\pdfoutput=1
\documentclass[a4paper,USenglish,cleveref,thm-restate]{lipics-v2021}

\hideLIPIcs  

\title{A Fast Deterministic Algorithm for \texorpdfstring{$(\Delta + 1)$}{(Δ+1)}-Edge Coloring in CONGEST} 

\titlerunning{A Fast Deterministic Algorithm for \texorpdfstring{$(\Delta + 1)$}{(Δ+1)}-Edge Coloring in CONGEST} 

\author{Sebastian Brandt}{CISPA Helmholtz Center for Information Security, Saarbrücken, Germany}{brandt@cispa.de}{https://orcid.org/0000-0001-5393-6636}{}

\author{Ananth Narayanan}{CISPA Helmholtz Center for Information Security, Saarbrücken, Germany}{ananth.narayanan@cispa.de}{https://orcid.org/0009-0002-6137-4025}{}

\author{Alexandre Nolin}{Telecom SudParis: Évry-Courcouronnes, Île-de-France, France}{alexandre.nolin@telecom-sudparis.eu}{https://orcid.org/0000-0002-3952-0586}{}

\authorrunning{S. Brandt, A. Narayanan and A. Nolin} 

\Copyright{Jane Open Access and Joan R. Public} 



\category{} 

\relatedversion{} 

\funding{Funded by the European Union. Views and opinions expressed are however those of the author(s) only and do not necessarily reflect those of the European Union or the European Research Council. Neither the European Union nor the granting authority can be held responsible for them. This work is supported by ERC grant \href{https://doi.org/10.3030/101162747}{OLA-TOPSENS} (grant agreement number 101162747) under the Horizon Europe funding programme.}

\acknowledgements{}

\nolinenumbers

\EventEditors{John Q. Open and Joan R. Access}
\EventNoEds{2}
\EventLongTitle{42nd Conference on Very Important Topics (CVIT 2016)}
\EventShortTitle{CVIT 2016}
\EventAcronym{CVIT}
\EventYear{2016}
\EventDate{December 24--27, 2016}
\EventLocation{Little Whinging, United Kingdom}
\EventLogo{}
\SeriesVolume{42}
\ArticleNo{23}

\usepackage{xcolor}
\usepackage{algorithm}
\usepackage{algpseudocode}
\usepackage{tcolorbox}

\usepackage{refcount}
\usepackage{mathtools}
\usepackage{aliascnt}

\usepackage{tikz}
\usetikzlibrary{arrows.meta,decorations.pathreplacing,decorations.pathmorphing,calc}

\definecolor{colone}{RGB}{31,110,190}
\definecolor{coltwo}{RGB}{200,80,20}
\definecolor{colthree}{RGB}{25,140,90}
\definecolor{colfour}{RGB}{130,70,175}
\definecolor{colfive}{RGB}{185,35,120}
\definecolor{colfanone}{RGB}{120,140,90}
\definecolor{colfantwo}{RGB}{90,125,160}
\definecolor{colfanthree}{RGB}{165,115,95}
\definecolor{colfanfour}{RGB}{55,150,165}
\definecolor{colfanfive}{RGB}{195,110,150}
\definecolor{colfansix}{RGB}{170,150,45}

\tikzset{
  vertex/.style={circle,draw=black,fill=white,line width=.7pt,inner sep=0pt,minimum size=5pt},
  nodelabel/.style={font=\small,inner sep=2pt},
  colored/.style={line width=1.2pt,draw=#1},
  colored/.default=black,
  blank/.style={line width=1.2pt,draw=black!40,densely dashed},
  colorlabel/.style={fill=white,inner sep=1.5pt,font=\small,text=#1},
  colorlabel/.default=black,
  edgelabel/.style={font=\scriptsize,text=black!70,inner sep=2pt},
  shiftto/.style={-{Straight Barb[length=2.2mm,width=2.6mm]},line width=1.3pt,draw=black!70},
  movearrow/.style={-{Straight Barb[length=1.6mm]},draw=black!45,line width=.6pt},
}

\newtheorem*{property*}{Property}
\newtheorem{fact}[theorem]{Fact}

\newtheorem{oq}{Question}

\newtheoremstyle{statementnoparen}
  {3pt}        
  {3pt}        
  {\itshape}   
  {}           
  {\bfseries}  
  {.}          
  {.5em}       
  {\thmname{#1}\thmnote{ #3}} 

\theoremstyle{statementnoparen}
\newtheorem*{statement}{Statement}

\crefname{case}{Case}{Cases}
\Crefname{case}{Case}{Cases}
\crefname{instance}{Property}{Properties}
\Crefname{instance}{Property}{Properties}
\crefname{fact}{fact}{facts}
\Crefname{fact}{Fact}{Facts}
\crefname{lemma}{lemma}{lemmas}
\Crefname{lemma}{Lemma}{Lemmas}
\crefname{observation}{observation}{observations}
\Crefname{observation}{Observation}{Observations}

\newcommand{\len}{\ell}
\newcommand{\numphases}{T}
\newcommand{\commoncolorprop}{common-color property}
\newcommand{\repeatedcolorprop}{repeated-color property}
\newcommand{\inheritedcolorprop}{inherited-color property}
\newcommand{\ABtuple}{(\alpha,\beta)}
\newcommand{\BAtuple}{(\beta,\alpha)}
\newcommand{\GDtuple}{(\G,\D)}

\newcommand{\ABSet}{\{\alpha,\beta\}}
\newcommand{\A}{\alpha}
\newcommand{\B}{\beta}
\newcommand{\G}{\gamma}
\newcommand{\D}{\delta}

\newcommand{\chainG}{\mathcal{C}}
\newcommand{\chainDelta}{\widetilde{\Delta}}

\DeclareMathOperator{\poly}{poly}
\DeclareMathOperator{\End}{End}
\DeclareMathOperator{\Start}{Start}
\newcommand{\vend}{\ensuremath{v}\!\End}
\newcommand{\eend}{\ensuremath{e}\!\End}
\newcommand{\vstart}{\ensuremath{v}\!\Start}
\newcommand{\estart}{\ensuremath{e}\!\Start}
\DeclarePairedDelimiter{\ceil}{\lceil}{\rceil}
\DeclarePairedDelimiter{\parens}{\lparen}{\rparen}
\DeclarePairedDelimiter{\card}{\lvert}{\rvert}
\DeclarePairedDelimiter{\set}{\lbrace}{\rbrace}
\DeclareMathOperator{\cen}{cen}
\newcommand{\Shift}{\operatorname{Shift}}
\newcommand{\id}{\operatorname{ID}}
\newcommand{\bichrom}{\operatorname{tail}}
\newcommand{\fA}{\mathcal{A}}
\newcommand{\suc}{\mathit{success}}

\date{}

\begin{document}

\maketitle

\begin{abstract}
    Vizing's theorem states that any graph of maximum degree $\Delta$ can be properly edge-colored with $\Delta + 1$ colors (which is optimal in general). A recent breakthrough result by Bernshteyn showed that such a $(\Delta + 1)$-edge coloring can be found deterministically in $\poly(\Delta,\log n)$ rounds in the LOCAL model of distributed computing, where $n$ denotes the number of vertices of the input graph [J.\ Comb.\ Theory 2022]. Since then, the exponent in the $\poly(\log n)$-part of the runtime has been improved by Christiansen [STOC 2023] and Bernshteyn and Dhawan [J.\ Comb.\ Theory, Series B, 2025].
    However, the algorithms used in all of these works use large messages, leaving open the question for efficient algorithms in the more restrictive CONGEST model.

    We answer this question by presenting the first $\poly(\Delta,\log n)$-round algorithm for $(\Delta + 1)$-edge coloring in the CONGEST model.
    Our algorithm is deterministic and the $n$-dependency of its runtime, $\tilde{O}(\log^5 n)$, matches the best published dependency in the LOCAL model.
    
\end{abstract}

\section{Introduction}
In this work, we study the problem of properly coloring the edges of a given $n$-vertex graph with $\Delta + 1$ colors, where $\Delta$ denotes the maximum degree of the graph.
By Vizing's theorem~\cite{vizing1964}, such a coloring always exists, while there are graphs (such as cliques with an odd number of vertices) that cannot be $\Delta$-edge-colored, making $\Delta + 1$ colors the optimum achievable in general.
As such, $(\Delta + 1)$-edge coloring constitutes one of the most fundamental graph problems and has been studied extensively across a wide range of models.
Especially the last few years have seen an explosion of works on $(\Delta + 1)$-edge coloring, making remarkable progress in different models, coming close to optimal algorithms in a number of cases (see, e.g., \cite{bhattacharya2024faster,assadi2025faster,bhattacharya2025even,assadi2025vizing,assadi2026vizing} for classical centralized algorithms, \cite{christiansen2023power,bhattacharya2024nibbling,christiansen2026deterministic} for dynamic algorithms, and \cite{bernshteyn2022fast,christiansen2023power,bernshteyn2025fast} for distributed algorithms).

We will study $(\Delta + 1)$-edge coloring in the CONGEST model~\cite{peleg2000distributed}, one of the two standard models for distributed algorithms.
Until now, research on distributed algorithms for $(\Delta + 1)$-edge coloring has focused on the more lenient LOCAL model~\cite{linial1987distributive}.
While it has been known for almost a decade that $(\Delta + 1)$-edge coloring requires $\Omega(\log n)$ rounds deterministically and $\Omega(\log \log n)$ rounds randomized~\cite{chang2019distributed}, obtaining LOCAL algorithms coming remotely close to these bounds remained elusive until a breakthrough result by Bernshteyn showed that $(\Delta + 1)$-edge coloring can be solved deterministically in $\poly(\Delta, \log n)$ rounds~\cite{bernshteyn2022fast}.
Subsequently, Christiansen~\cite{christiansen2023power} improved the runtime to $\tilde{O}(\poly(\Delta) \log^6 n)$ rounds\footnote{Throughout the paper, we will use $\tilde{O}(\cdot)$ to hide factors that are at most polylogarithmic in the argument.}, which was further improved by Bernshteyn and Dhawan~\cite{bernshteyn2025fast} to the current state of the art of $\tilde{O}(\poly(\Delta) \log^5 n)$ rounds.
Bernshteyn and Dhawan also provided the best currently known randomized $(\Delta + 1)$-edge coloring algorithm, which has a runtime of $O(\poly(\Delta) \log^2 n)$ rounds.

On graphs with sufficiently small maximum degree---in particular on constant-degree graphs, which constitute an exceptionally well-studied graph class in the LOCAL model--- the aforementioned deterministic algorithms are only polynomially slower than the mentioned lower bound, making them reasonably efficient and justifying the focus on runtimes of the form $\poly(\Delta,\log n)$ and on improvements in the exponent in the $\poly(\log n)$-part of the runtime.
Unfortunately, all of these algorithms make use of large messages, rendering them infeasible in the more restrictive CONGEST model and raising the following natural question.

\vspace{0.05cm}
\begin{tcolorbox}
	\begin{oq}\label{oq1}
    Can one obtain deterministic algorithms in the CONGEST model that are similarly efficient as the best known deterministic LOCAL algorithms for $(\Delta + 1)$-edge coloring?
    \end{oq}
\end{tcolorbox}
\vspace{0.05cm}

\subsection{Our Contributions}
We answer~\Cref{oq1} in the positive by proving the following theorem.

\begin{theorem}\label{deterministic_polylog}
    There exists an $\tilde{O}(\poly(\Delta) \cdot \log^5 n)$-round deterministic CONGEST algorithm for $(\Delta+1)$-edge coloring.
\end{theorem}

Our result yields the first randomized or deterministic algorithm for $(\Delta + 1)$-edge coloring in the CONGEST model with only a polynomial dependency on $\Delta$ and a polylogarithmic dependency on $n$, matching the breakthrough result Bernshteyn achieved in the LOCAL model~\cite{bernshteyn2022fast}.
Moreover, our dependency of $\tilde{O}(\log ^5 n)$ on the number $n$ of vertices matches the best published dependency in the LOCAL model~\cite{bernshteyn2025fast}, despite our result being achieved in the more restrictive CONGEST model.
We would like to note that combining the randomized result from~\cite{bernshteyn2025fast} with the state-of-the-art network decomposition algorithm in the LOCAL model~\cite{GG_focs24_nd_mis} and a well-known derandomization argument~\cite{GHK_focs18_derandomizing} would improve the $n$-dependency in the LOCAL model to $\tilde{O}(\log^4 n)$.
However, as current research seems to be far from understanding how a similar generic derandomization result could hold in the CONGEST model, there is little hope to extend this approach to the CONGEST model.

We also obtain the following simple corollary of \Cref{deterministic_polylog}, providing the state-of-the-art complexity in the CONGEST model for graphs of constant maximum degree, a graph class that has been the focus of an abundance of works in the field of distributed algorithms.

\begin{corollary}\label{cor:constant-degree}
    On constant-degree graphs, $(\Delta + 1)$-edge coloring can be solved deterministically in $\tilde{O}(\log^5 n)$ rounds in the CONGEST model.
\end{corollary}
In particular, on constant-degree graphs, \Cref{cor:constant-degree} reduces the gap between the best known upper bound and the best known lower bound of $\Omega(\log n)$ rounds~\cite{chang2019distributed} (which already holds on constant-degree graphs) to being polynomial.

Our algorithm combines methods from known LOCAL algorithms for $(\Delta + 1)$-edge coloring with novel ideas. To contrast our approach with the approaches taken in previous work (in the LOCAL model), we start by giving a high-level overview of the latter.

\subparagraph*{Inner workings of recent distributed algorithms for $(\Delta+1)$-edge coloring.}
State-of-the-art LOCAL algorithms for computing a $(\Delta+1)$-edge coloring all follow the same outline: construct a set of mutually disjoint augmenting subgraphs called \emph{multi-step Vizing chains} for a large enough fraction of uncolored edges, extend the current coloring to these uncolored edges by recoloring edges within the augmenting subgraphs, and repeat until no uncolored edge remains.
Informally speaking a multi-step Vizing chain is a concatenation of multiple Vizing chains, while a Vizing chain is the union of a star and a path on which the edges' colors can be shifted while keeping the partial coloring correct. The path part of a Vizing chain is typically bichromatic, i.e., its edges are colored by two alternating colors.
For precise definitions, we refer the reader to \Cref{sec:preliminaries}.

The aforementioned set of mutually disjoint multi-step Vizing chains for a large fraction of the uncolored edges is found in different ways in previous works~\cite{bernshteyn2022fast,christiansen2023power,bernshteyn2025fast}. 
In \cite{bernshteyn2022fast,bernshteyn2025fast}, the authors construct multi-step Vizing chains with the help of randomization.
For each uncolored edge $e$, they progressively build a Vizing chain edge by edge, similar to the construction in the classic proof of Vizing's theorem~\cite{vizing1964}, starting with a star centered on an endpoint of $e$ and then following a bichromatic path adjacent to the star. 
While the classic proof of Vizing's theorem follows the bichromatic path until its end, at which point the star and path can be recolored to extend the partial coloring to $e$, an efficient distributed algorithm cannot afford to follow arbitrarily long paths.
Instead, whenever the length of a Vizing chain's bichromatic path exceeds a certain threshold, the path is truncated.
More precisely, a vertex on the bichromatic path constructed so far is chosen at random and is set as the endpoint of the first Vizing chain.
A new Vizing chain is then constructed starting from this endpoint, adding a new step to the multi-step Vizing chain.
This process of truncating and growing the next Vizing chain is iterated until the concatenation of the computed Vizing chains is ``terminating''. That is, until the obtained multi-step Vizing chain can be recolored in a way that extends the current partial coloring to the uncolored edge $e$ from which the multi-step Vizing chain was started
Each operation of starting a new Vizing chain from the truncated path of the last chain is an additional ``step'' of the multi-step Vizing chain.

The authors showed that, with sufficiently large probability, $O(\log n)$ steps suffice to obtain a terminating multi-step Vizing chain, i.e., one that can be used to color an extra edge in the graph. 
The random truncation performed at each step limits the size of the built multi-step Vizing chains, which is crucial to obtain an efficient distributed algorithm.
Moreover, they showed that when performing the same process for all the uncolored edges in parallel, each multi-step Vizing chain intersects with only few other multi-step Vizing chains in expectation.
As a result, basic random sampling within the set of multi-step Vizing chains suffices to find a set of mutually disjoint multi-step Vizing chains that is large in expectation. Shifting colors along all these mutually disjoint chains colors a large fraction of the uncolored edges in expectation. Repeating  the process enough times colors all the edges in the graph with high probability. To get a deterministic algorithm, they use a standard distributed derandomization technique using the method of conditional expectations and graph network decompositions. 

In \cite{christiansen2023power}, Christiansen follows the same overall approach, but without using any randomization. Recall that the other works made use of randomness during the chains' construction for the random truncation, and after construction to find a large set of mutually disjoint chains within the set of all chains.
Although not explicitly phrased in this way, Christiansen's argument uses a ``flooding'' approach to construct a multi-step Vizing chain for an uncolored edge $e$.
Their construction of a Vizing chain also starts the same way as in the classic proof of Vizing's theorem~\cite{vizing1964}.
As previously, they stop the construction of the chain if its bichromatic path is too long. Instead of trimming the chain at a random location on the path as in \cite{bernshteyn2022fast,bernshteyn2025fast} and starting a new Vizing chain from the truncation point, they consider all the Vizing chains that could originate from all of the bichromatic path's  vertices. They repeat this process $O(\log n)$ times, resulting in a very large set of potential multi-step chains for each uncolored edge $e$. They show that the large set of Vizing chains originating from an uncolored edge $e$ must contain a terminating chain at the end of this process. This process is done for all the uncolored edges in parallel. Unlike in the other works, the process gives no guarantee that a given multi-step Vizing chain of an uncolored edge intersects with the multi-step Vizing chains of only a few others. However, from arguments about the density of the graph, the author is able to show that there exists a large enough maximum independent set of disjoint multi-step Vizing chains.
An algorithm for hypergraph maximal matching is used to compute a large enough approximation of this maximum independent set. This gives a deterministic procedure to color a large enough fraction of the uncolored edges.

\subparagraph*{Our approach.}
Running the above algorithms as they are in the CONGEST model is unfeasible, as they use large messages across multiple parts. In \cite{bernshteyn2022fast,bernshteyn2025fast}, derandomizing the algorithm using the method of conditional expectation is difficult since the randomness of all the vertices which are within $O(\log n)$ distance from each other might be correlated. Hence the method of conditional expectation does not work well in polylogarithmic time. It is also difficult to use $k$-wise independence since the random variables they use in the analysis would need to be $\Omega(\log n)$-wise independent at the very least. In \cite{christiansen2023power}, the initial difficulty comes in trying to construct a multi-step Vizing chain for all the uncolored edges in parallel which requires potentially very large messages. The second difficulty comes from the fact that hypergraph maximal matching is not feasible in the CONGEST model with the current algorithms especially when the rank of the hypergraph is of order $\poly(\Delta,\log n)$.  

Our approach mirrors that of Christiansen~\cite{christiansen2023power}, 
who showed for the first time that for every uncolored edge $e$ under a proper partial $(\Delta+1)$-edge coloring, there exists a $\poly(\Delta)\log n$-sized multi-step Vizing chain which can be recolored to extend the partial coloring to $e$, i.e., a $\poly(\Delta)\log n$-sized augmenting subgraph.
Similar to the density argument used in their proof, we use a ``selective flooding'' approach to construct a multi-step Vizing chain for every uncolored edge in parallel. We also start our construction of a Vizing chain as in the classic proof of Vizing's theorem \cite{vizing1964}. We stop following the bichromatic path beyond a certain length. Instead of constructing a Vizing chain from all of the points on the bichromatic path, we use a carefully chosen large enough subset of the points on this bichromatic path and construct a Vizing chain from these points. We repeat this process $O(\log n)$ times. Limiting the number of new Vizing chains we consider as extensions of existing multi-step Vizing chains helps us control the number of messages that need to go through any given edge at various points in our algorithm.
We use a similar analysis as in \cite{christiansen2023power} to show that this process finds a terminating multi-step Vizing chain for every uncolored edge. Moreover, our algorithm ensures that the multi-step Vizing chain of an uncolored edge intersects with only few other multi-step Vizing chains. This bound on the number of intersections that each chain can have with other chains allows us to find a large enough set of mutually disjoint multi-step Vizing chains without relying on an algorithm for hypergraph maximal matching, and to instead use an algorithm for computing large independent sets.

With each multi-step Vizing chain intersecting at most $\poly(\Delta)\log^2 n$ other chains, there necessarily exists an independent set of multi-step Vizing chains containing at least a $1/(\poly(\Delta) \log^2 n)$-fraction of all chains. Notably, any maximal independent set is at least this large.
To compute an independent set containing at least this many multi-step Vizing chains, we leverage recent results from the literature on computing independent sets and colorings in the CONGEST model~\cite{FGGKR_talg25_generalized_rounding,BG_disc24}.
Finding a large independent set of multi-step Vizing chains corresponds to computing an independent set over a virtual graph whose nodes are the selected multi-step Vizing chains and whose edges represent intersections of the chains. Simulating a round of communication over this virtual graph requires multiple rounds of communication over the communication network $G$. Direct simulation of the rounding-based algorithm from~\cite{FGGKR_talg25_generalized_rounding} would compute the independent set we need in $\tilde{O}(\poly(\Delta)\log^3 n)$ rounds.
We shave off an $O(\log n)$-factor using that if starting with a coloring, the algorithm works with smaller messages the less colors are used by the coloring. We compute a $\poly(\Delta,\log n)$-coloring of the virtual graph prior to computing the independent set, using a variant of Linial's classic coloring algorithm adapted to coloring graphs other than the communication network itself~\cite{linial1987distributive,BG_disc24}.

Having found a set of independent multi-step Vizing chains containing at least a $1/(\poly(\Delta) \log^2 n)$-fraction of all the chains we built, shifting colors along those chains reduces the set of uncolored edges by the same fraction. Repeating this process $\poly(\Delta) \log^3 n$ times ensures that no uncolored edge remains in the graph.

\subsection{Further Related Work}
A fundamental distinction for edge coloring problems in the LOCAL and CONGEST models comes from the number of colors allowed: If $2\Delta - 1$ colors are allowed, then
a proper edge coloring can be found deterministically in $O(\log^* n + f(\Delta))$ rounds for some function $f$ (see, e.g., \cite{goldberg1987parallel,panconesi2001some,barenboim2018locally}) whereas already $(2\Delta - 2)$-edge coloring has a randomized lower bound of $\Omega(\log_\Delta \log n)$ rounds and a deterministic lower bound of $\Omega(\log_\Delta n)$ rounds \cite{chang2019distributed} (already on constant-degree graphs in the LOCAL model).

The best currently known deterministic upper bounds for the complexity of $(2\Delta - 1)$-edge coloring in the LOCAL model are $O(\log^* n + \log^{12} \Delta)$ rounds for complexities of the form mentioned above (where randomness is known not to help) \cite{balliu2022distributed} and $\tilde{O}(\log^{5/3} n)$ rounds for complexities as a function of $n$ \cite{GG_focs24_nd_mis}.
An algorithm of complexity $O(\log^2 \Delta \log n)$ is also known~\cite{GK_focs21}.
In the CONGEST model, the state-of-the-art deterministic complexities are provided by the $O(\log^* n + \Delta)$-round algorithm from~\cite{barenboim2018locally} and the $\tilde{O}(\log^{2.5} n + \log^2 \Delta \log n)$-round algorithm from~\cite{blikstad2026deterministic}.
For randomized algorithms, the best known complexities are $\tilde{O}(\log^{5/3}\log n)$ rounds in the LOCAL model~\cite{GG_focs24_nd_mis} and $O(\log^4 \log n)$ rounds in the CONGEST model~\cite{halldorsson2023superfast}.

For sub-$(2\Delta - 1)$-edge coloring, recent years have seen a variety of new algorithms with the following state-of-the-art runtimes depending on the number of allowed colors (beyond the $(\Delta + 1)$-edge coloring results already discussed).
For $2\Delta - 2$ colors, \cite{JakobMS25} provides a reduction from $(2\Delta - 2)$- to $(2\Delta - 1)$-edge coloring resulting in deterministic LOCAL algorithms with runtimes of $\tilde{O}(\log^{5/3} n)$, $O(\log^{2} \Delta \log n)$, and $O(\log^{12} \Delta + \log n)$ rounds deterministically and $\tilde{O}(\log^{5/3} \log n)$, $O(\log^{2} \Delta \log \log n)$, and $O(\log^{12} \Delta + \log \log n)$ rounds randomized.
For $3/2\cdot \Delta$ colors, \cite{brandt2025locality} gives a deterministic LOCAL algorithm with runtime $O(\Delta^2 \log n)$ rounds; for $(3/2 + \varepsilon) \cdot \Delta$ colors (for any $\varepsilon > 0$), \cite{maus2026fast} provides a deterministic LOCAL algorithm with runtime $O(\varepsilon^{-1}\log^2 \Delta \log n + \varepsilon^{-2} \log n)$ rounds.
For $(1 + \varepsilon)\cdot \Delta + \Theta(\sqrt{\log n})$ colors (for any constant $\varepsilon > 0$), \cite{blikstad2026deterministic} provides an $O(\log^2 n) + \tilde{O}(\log^2 \Delta \log n)$-round deterministic LOCAL algorithm.
The $\tilde{O}(\log^{2.5} n + \log^2 \Delta \log n)$-round deterministic CONGEST algorithm from that work (mentioned above for $2\Delta - 1$ colors) also works down to $(1 + \varepsilon)\cdot \Delta + \Theta(\sqrt{\log n})$ colors.
For $(1 + \varepsilon) \cdot \Delta$ colors (for any constant $\varepsilon > 0$), \cite{HalldorssonMN22} provides an $O(\poly(\log \log n))$-round randomized CONGEST algorithm, as long as $\Delta$ is larger than a certain constant. Earlier works~\cite{EPS_soda15,halldorsson2023superfast} also showed $(1+O(1))\Delta$-edge coloring to admit $O(\log^* n)$-round randomized algorithms in both LOCAL and CONGEST when $\Delta \in \log^{1+\Omega(1)} n$.
For $(1 + \varepsilon) \cdot \Delta$ colors (though now under the restriction $\varepsilon \in \omega((\log^{2.5}\Delta)/\sqrt{\Delta}$), \cite{davies2023improved} gives a randomized LOCAL algorithm with runtime $O(\poly(1/\varepsilon, \log \log n))$ rounds.
The same work also provides an $O(\poly(\log \log n))$-round randomized LOCAL algorithm for $(\Delta + o(\Delta))$-edge coloring.
For an overview including earlier edge coloring algorithms (with a focus on the LOCAL model), see also the excellently written~\cite[Section 1.3]{JakobMS25}.

\section{Preliminaries}\label{sec:preliminaries}

\subparagraph*{Graph-theoretic notation.}
The input graph $G$ is a finite simple undirected graph.
For a graph $G$, we will use $V(G)$ and $E(G)$ to denote the set of vertices and the set of edges of $G$, respectively.
We will use $n$ to denote the number of vertices of a considered graph, i.e., $n := |V(G)|$, and $\Delta$ to denote its maximum degree.
We assume throughout that $\Delta\geq 2$; if $\Delta\leq 1$, then $G$ is a matching and coloring every edge with the same color is a proper $(\Delta+1)$-edge coloring.
We denote the distance between vertices $x$ and $y$ in $G$ by $d_G(x,y)$.
For a vertex $x$ and a positive integer $i$, we define $N^i(x) := \{ v \in V(G) \mid 1 \leq d_G(v,x) \leq i \}$ and $N(x) := N^1(x)$.
For $X\subseteq V(G)$ and a positive integer $i$, we set $N^i(X) := \bigcup_{x \in X} N^i(x)$ and $N(X) := N^1(X)$.
Define $N^i[x]:=N^i(x)\cup\{x\}$ and $N^i[X]:=N^i(X)\cup X$. Set $N[x]:=N^1[x]$ and $N[X]=N^1[X]$.
For an edge subset $E' \subseteq E(G)$, we denote by $V(E')$ the set of vertices in the subgraph induced by the edges in $E'$.
If $E'=\{e\}$ for some edge $e$, we may write $V(E')$ simply as $V(e)$.

\subparagraph*{Model of computation.}
The model of computation we consider in this work is the standard (deterministic) CONGEST model of distributed computation.
In the CONGEST model, each vertex of the input graph is considered as a computational entity, while edges are communication links over which the vertices (or the corresponding entities) communicate.
Time in the CONGEST model is split into synchronous rounds where in each round each vertex can send a message of $O(\log n)$ bits over each edge and, after receiving the messages sent by its neighbors in the same round, perform any arbitrarily complex computation on the information it has gathered up to this point.
Each vertex has to decide at some point to terminate; when a vertex terminates, it outputs its local part of the global solution and does not take part in any further sending of messages.
For our problem of $(\Delta + 1)$-edge coloring, when terminating each vertex has to output the color of each incident edge such that neighboring vertices agree on the color of the connecting edge.

Initially each vertex is aware of its own degree and the number of vertices in the graph as well as the value of $\Delta$ (which is required for each vertex to know the number of allowed colors).
Each vertex $v$ is also equipped with a globally unique $O(\log n)$-bit identifier $\id(v)$.
To simplify some of our arguments, we will also assign identifiers to edges: for each edge $e=uv\in E(G)$, let $\id(e):=\parens*{\min\{\id(u),\id(v)\},\max\{\id(u),\id(v)\}}$, so that both endpoints of $e$ infer the same identifier for it.
Note that each vertex can infer the identifiers of its incident edges in one round of communication (and that the edge identifiers are globally unique and of size $O(\log n)$ bits).

The runtime or complexity of an algorithm is the time until the last vertex terminates; the complexity of a problem is the runtime of an optimal algorithm.
Our objective is to improve the state-of-the-art upper bound on the \emph{worst-case} complexity of $(\Delta + 1)$-edge coloring as a function of $n$. 

While our results are obtained in the CONGEST model, at times we will also refer to the (stronger) LOCAL model of distributed computation.
The LOCAL model is simply the CONGEST model without any restriction on the size of the messages sent, i.e., vertices may send arbitrarily large messages to each other (though still only one per edge per round).
We will also refer to the CONGEST$(B)$ model, which is the same as the CONGEST model except that the upper bound for each message is $B$ bits instead of $\Theta(\log n)$ bits.
We will make use of the following simple observation throughout the paper.

\begin{observation} \label{obsvn:LOCAL_to_CONGEST}
     A $T$-round CONGEST$(B)$ algorithm can be simulated in the CONGEST model in $O(\ceil{\frac{B}{\log n}}\cdot T)$ rounds.
\end{observation}

\subsection{Chains and Shifts}
Next we define some standard objects considered in the context of $(\Delta + 1)$-edge coloring.

Let $\chi$ be a proper partial $(\Delta+1)$-edge coloring of $G$ for the remainder of this paper unless indicated otherwise.
We will write $\chi(e) = \neg$ to indicate that an edge $e$ is uncolored under $\chi$.
A color $c\in [(\Delta+1)]$ is called \emph{available (under $\chi$ at $v$)} if no edge adjacent to $v$ is colored $c$ under $\chi$. We denote the set of available colors at $v$ under $\chi$ by  $A(v,\chi)$. In case $\chi$ is clear from the context, we may omit the argument.

Consider two adjacent edges $f , g \in E(G)$ in a graph $G$ under $\chi$. Then we define another (not necessarily proper) partial edge coloring $\Shift(\chi,f,g)$ by setting: 
\[
\Shift(\chi,f,g)(e) = \begin{cases}
     \chi(g) &\quad\text{if } e = f \\
     \neg &\quad\text{if } e = g \\
    \chi(e) &\quad\text{if } e \notin \{f,g\} \\
     \end{cases}
\]
We call such a pair $(f,g)$ of adjacent edges $\chi$-\emph{shiftable} if $\chi(f) = \neg$, $\chi(g) \neq \neg$, and $\Shift(\chi,f,g)(e)$ is a proper partial edge coloring. In case the coloring $\chi$ is clear from the context, we may omit that argument.

A \emph{chain $C$ under $\chi$} of length $k\geq1$ is a tuple of edges
$C = (e_1, \dots, e_k)$ such that $e_{i}$ and $e_{i+1}$ are adjacent
for all $1\leq i\leq k-1$ and $\chi(e_i) = \neg$ iff $i = 1$.
For $0 \leq i \leq k-1$, we define partial colorings
$\Shift_{i}(\chi, C)$ via
\begin{align*}
    \Shift_{0}(\chi, C) &:= \chi, \text{ and} \\
    \Shift_{i}(\chi, C) &:= \Shift(\Shift_{i-1}(\chi, C), e_{i}, e_{i+1}).
\end{align*}
We may not specify $\chi$ when we refer to the chain $C$ if $\chi$ is clear from the context.
We say that $C$ is $\chi$-shiftable if every pair of edges $(e_{i}, e_{i+1})$ is
$\Shift_{i-1}(\chi,C)$-shiftable. We will abuse notation slightly and write $\Shift_{k-1}(\chi,C)$ as $\Shift(\chi,C)$.
Note that a chain of length $1$ is just an uncolored edge and is $\chi$-shiftable by definition.
\Cref{fig:shift} illustrates the effect of shifting a chain.

\begin{figure}[t]
    \centering
    \begin{tikzpicture}[scale=.95,every node/.style={transform shape}]
        \begin{scope}
            \node[vertex] (x) at (0,0) {};
            \node[vertex] (u) at (-1.3,.7) {};
            \node[vertex] (v) at (.1,1.6) {};
            \node[vertex] (w) at (1.4,.75) {};
            \node[vertex] (p) at (2.8,.1) {};
            \node[vertex] (q) at (4.1,.85) {};

            \draw[blank] (x) -- (u);
            \draw[colored=colone] (x) -- node[colorlabel=colone] {$c_1$} (v);
            \draw[colored=coltwo] (x) -- node[colorlabel=coltwo] {$c_2$} (w);
            \draw[colored=colthree] (w) -- node[colorlabel=colthree] {$c_3$} (p);
            \node[nodelabel] at (1.4,-.75) {$\chi$};
            \draw[colored=colfour] (p) -- node[colorlabel=colfour] {$c_4$} (q);
        \end{scope}

        \draw[shiftto] (5.0,.8) -- node[above=1pt,font=\small] {$\Shift$} (6.5,.8);

        \begin{scope}[xshift=8.4cm]
            \node[vertex] (x) at (0,0) {};
            \node[vertex] (u) at (-1.3,.7) {};
            \node[vertex] (v) at (.1,1.6) {};
            \node[vertex] (w) at (1.4,.75) {};
            \node[vertex] (p) at (2.8,.1) {};
            \node[vertex] (q) at (4.1,.85) {};

            \draw[colored=colone] (x) -- node[colorlabel=colone] {$c_1$} (u);
            \draw[colored=coltwo] (x) -- node[colorlabel=coltwo] {$c_2$} (v);
            \draw[colored=colthree] (x) -- node[colorlabel=colthree] {$c_3$} (w);
            \draw[colored=colfour] (w) -- node[colorlabel=colfour] {$c_4$} (p);
            \draw[blank] (p) -- (q);
            \node[nodelabel] at (1.4,-.75) {$\Shift(\chi,C)$};
        \end{scope}
    \end{tikzpicture}
    \caption{Shifting a chain $C=(e_1,\dots,e_k)$; dashed edges are uncolored.
    Every colored edge of $C$ hands its color to the preceding edge, so that
    $e_1$ receives $\chi(e_2)$ and $e_k$ becomes uncolored.}
    \label{fig:shift}
\end{figure}
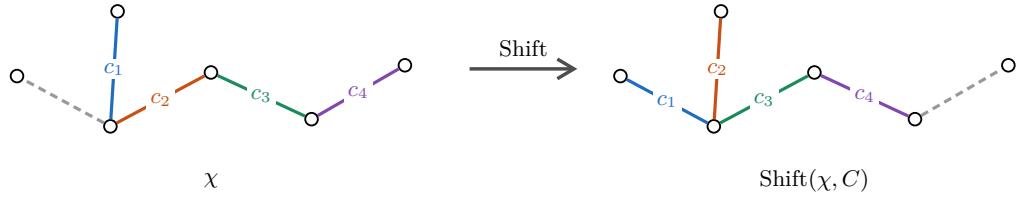

For two chains $C_1=(e_1,e_2,\dots,e_k)$ and $C_2=(f_1,f_2,\dots,f_q)$ satisfying $e_k = f_1$, we define the \emph{concatenation $C_1 + C_2$ of $C_1$ and $C_2$} as $C_1+C_2 := (e_1,e_2,\dots,e_k=f_1,f_2,\dots f_q)$.
If $C_1$ is $\chi$-shiftable and $C_2$ is $\Shift(\chi,C_1)$-shiftable, then $C_1 + C_2$ is a $\chi$-shiftable chain since
$$
\Shift(\chi, C_1 + C_2)
    = \Shift(\Shift(\chi, C_1), C_2).
$$
The \emph{initial segment} of length $r$ of a chain $C$ is the chain $C|r := (e_1, \dots, e_r)$. Note that if $C$ is $\chi$-shiftable, $C|r$ is also $\chi$-shiftable. If $r$ is greater than the number of edges in $C$, then $C|r$ is defined as $C$ itself.

We may refer to chain $C=(e_1,e_2,\dots,e_k)$ as chain $C$ \emph{on $e_1$} to emphasize the first edge in the chain.
Moreover, we call a chain $C$ \emph{simple} if all edges in $C$ are distinct.
A chain $C$ may be considered as a subgraph of $G$ in the natural way. Accordingly, we may use notation such as $V(C)$, $E(C)$, or $d_C(x,y)$.

\subsection{Types of Chains}
\suppressfloats[t]

\begin{figure}[t]
    \centering
    \begin{tikzpicture}[scale=.85,every node/.style={transform shape}]
        \def\rad{2.2}
        \begin{scope}
            \node[vertex,label={[nodelabel]-90:$x=\cen(F)$}] (x) at (0,0) {};
            \node[vertex,label={[nodelabel]185:$z_1=y$}] (z1) at (185:\rad) {};
            \node[vertex,label={[nodelabel]150:$z_2$}]   (z2) at (150:\rad) {};
            \node[vertex,label={[nodelabel]115:$z_3$}]   (z3) at (115:\rad) {};
            \node[vertex,label={[nodelabel]47:$z_{k-1}$}] (z4) at (47:\rad) {};
            \node[vertex,label={[nodelabel]12:$z_k=\vend(F)$}] (z5) at (12:\rad) {};
            \node[font=\small] at (80:1.85) {$\cdots$};

            \draw[blank] (x) -- (z1);
            \draw[colored=colone]   (x) -- node[colorlabel=colone,pos=.68]   {$\eta_1$} (z2);
            \draw[colored=coltwo]   (x) -- node[colorlabel=coltwo,pos=.68]   {$\eta_2$} (z3);
            \draw[colored=colfour]  (x) -- node[colorlabel=colfour,pos=.68]  {$\eta_{k-2}$} (z4);
            \draw[colored=colfive]  (x) -- node[colorlabel=colfive,pos=.68]  {$\eta_{k-1}$} (z5);
            \node[nodelabel] at (0,-1.15) {$\chi$};
        \end{scope}

        \draw[shiftto] (4.5,1.6) -- node[above=3pt,font=\small] {$\Shift$} (5.9,1.6);

        \begin{scope}[xshift=9.3cm]
            \node[vertex,label={[nodelabel]-90:$x$}] (x) at (0,0) {};
            \node[vertex,label={[nodelabel]185:$z_1=y$}] (z1) at (185:\rad) {};
            \node[vertex,label={[nodelabel]150:$z_2$}]   (z2) at (150:\rad) {};
            \node[vertex,label={[nodelabel]115:$z_3$}]   (z3) at (115:\rad) {};
            \node[vertex,label={[nodelabel]47:$z_{k-1}$}] (z4) at (47:\rad) {};
            \node[vertex,label={[nodelabel]12:$z_k$}]     (z5) at (12:\rad) {};
            \node[font=\small] at (80:1.85) {$\cdots$};

            \draw[colored=colone]   (x) -- node[colorlabel=colone,pos=.68]   {$\eta_1$} (z1);
            \draw[colored=coltwo]   (x) -- node[colorlabel=coltwo,pos=.68]   {$\eta_2$} (z2);
            \draw[colored=colthree] (x) -- node[colorlabel=colthree,pos=.68] {$\eta_3$} (z3);
            \draw[colored=colfive]  (x) -- node[colorlabel=colfive,pos=.68]  {$\eta_{k-1}$} (z4);
            \draw[blank] (x) -- (z5);
            \node[nodelabel] at (0,-1.15) {$\Shift(\chi,F)$};
        \end{scope}
    \end{tikzpicture}
    \caption{Shifting a fan chain $F=(xz_1,\dots,xz_k)$ on the uncolored edge
    $e=xy$, centered at $x=\cen(F)$ and ending at $z_k=\vend(F)$; dashed edges
    are uncolored. The color $\eta_i=\chi(xz_{i+1})$ of the edge $xz_{i+1}$ is
    the representative available color at $z_i$, which is what makes $F$
    shiftable: shifting moves $\eta_i$ onto $xz_i$ for every $i<k$ and leaves
    $xz_k$ uncolored.}
    \label{fig:fan}
\end{figure}
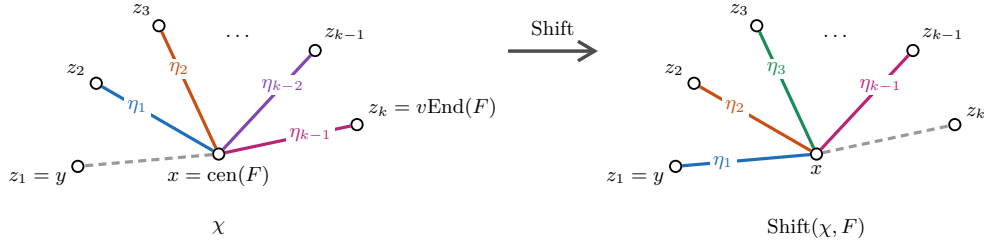

We define two basic types of chains---\emph{fan chains} and \emph{path chains} (under the proper partial coloring $\chi$). The following definitions are virtually the same as in \cite{bernshteyn2022fast,bernshteyn2025fast,christiansen2023power}. 

\subparagraph*{Fan chains.}
We define a \emph{fan chain} (or simply a \emph{fan}), on an uncolored edge $e=xy$. A fan chain $F$ \emph{on $e$} is a simple chain $(xy=xz_1, \dots, xz_k)$ such that $\chi(xz_{i+1}) \in A(z_i)$ for all $1 \leq i \leq k - 1$.\footnote{The intuition behind this definition is that the property $\chi(xz_{i+1}) \in A(z_i)$ guarantees that $F$ is $\chi$-shiftable.} We say $\chi(xz_{i+1})$ is the
\emph{representative available color} at $z_i$ in $F$. We call $x$ the \emph{center} of the fan chain. We will denote $x$ as $\cen(F)$ and $z_k$ as $\vend(F)$.
Note that one can similarly define a fan chain $F$ on $e$ with $y$ as the center.

We will also select an arbitrary color from $A(\vend(F))$ as the representative available color at $\vend(F)$.
If this representative available color is also available at $x$ or
if it is the representative available color for some $z_j$ with $j < k$,
then we say that the fan is \emph{maximal}.
See \Cref{fig:fan} for an illustration of a fan chain.

\subparagraph*{Path chains.}
A \emph{trail} $P$ is a tuple of edges $P=(e_1,e_2,\dots,e_k)$ such that
$e_i\neq e_j$ for $1\leq i<j \leq k$, and $e_i$ and $e_{i+1}$ are
adjacent for $1\leq i \leq k-1$. We call $k$ the \emph{length} of
$P$. We say $P$ \emph{starts at vertex $x$} if $x\in V(e_1)$ and $x\notin V(e_1)\cap V(e_2)$ and we denote $x$ as $\vstart(P)$. We say $P$ \emph{ends at vertex $x$} if $x\in V(e_k)$ and $x\notin V(e_k)\cap V(e_{k-1})$ and we denote $x$ as $\vend(P)$; for $k=1$, which endpoint of $e_1$ is $\vstart(P)$ and which is $\vend(P)$ is part of the trail. Similarly, we say $P$ \emph{starts on edge $e_1$} and \emph{ends at edge $e_k$}, and we denote $e_1$ and $e_k$ as $\estart(P)$ and $\eend(P)$ respectively.
A trail $P$ is a \emph{path} if its vertices are pairwise distinct, i.e., if there are pairwise distinct vertices $v_0,\dots,v_k$ with $e_i=v_{i-1}v_i$ for $1\leq i\leq k$; then $\vstart(P)=v_0$ and $\vend(P)=v_k$. For vertices $u,v$ visited by $P$, we write $d_P(u,v)$ for the number of edges $P$ traverses between its first visits to $u$ and to $v$.

We denote the prefix $(e_1,e_2,\dots e_r)$ of $P$ as $P|r$, the trail $(e_2,\dots e_k)$ as $\bichrom(P)$, and the reversed trail $(e_k,\dots,e_1)$ as $\operatorname{rev}(P)$; if $P$ is a path, then so is each of these.
Suppose $P_1=(f_1,\dots f_p)$ and $P_2=(g_1,\dots,g_q)$ such that $(f_1,\dots f_p,g_1,\dots g_q)$ is a trail, say $P_3$, then we
write $P_3=P_1+P_2$. 
For $\alpha, \beta \in [\Delta+1]$, $\A\neq \B$, we say $P$ is an \emph{$\ABtuple$-bichromatic path} or equivalently a \emph{$\BAtuple$-bichromatic path} if $P$ is a path consisting
only of edges colored with colors from $\ABSet$. We say that an $\ABtuple$-bichromatic path is \emph{maximal} if, for endpoints $x$ and $y$ of $P$, $|A(x)\cap\ABSet|=1$ and $|A(y)\cap\ABSet|=1$. We say that $P$ is an \emph{$x$-maximal} $\ABtuple$-bichromatic path if $x$ is an endpoint of $P$ and $|A(x)\cap\ABSet|=1$.
Since $\chi$ is proper, the subgraph formed by the edges colored $\A$ or $\B$ has maximum degree at most $2$, so each of its components is a path or a cycle. This yields the following observation.

\begin{observation} \label{obsvn:bichromatic_component_is_path}
    Let $\A\neq\B$ and let $x$ be a vertex with $|A(x,\chi)\cap\ABSet|=1$. Then there is a unique maximal $\ABtuple$-bichromatic path under $\chi$ having $x$ as an endpoint, namely the component of $x$ in the subgraph formed by the edges colored $\A$ or $\B$ under $\chi$.
\end{observation}

We say vertices $x$ and $y$ are \emph{$\ABtuple$-related} (or equivalently \emph{$\BAtuple$}-related) under $\chi$ if there exists an $\ABtuple$-bichromatic path under $\chi$ between $x$ and $y$. We say vertices $x$ and $y$ are \emph{$P$-related} if $x$ and $y$ are vertices on $P$.

We now define the notion of \emph{path chains}. We emphasize that unlike fan chains which may be referred to as fans as well, a path chain and a path are defined differently.
A tuple of edges $P=(e_1,e_2,\dots e_k)$ is a \emph{path chain} under $\chi$ if $P$ is a simple chain under $\chi$ and $P$ is a path in the graph $G$.
We define the \emph{length} or \emph{size} of a path chain $P = (e_1, \dots, e_k)$ as $k$. 

An \emph{$\ABtuple$-bichromatic path chain} is a path chain $P = (e_1, \dots, e_k)$ such that 
\begin{itemize}
    \item $\chi(e_1) = \neg$,
    \item $\chi(e_2) = \chi(e_4) = \dots = \A$,
    \item $\chi(e_3) = \chi(e_5) = \dots = \B$,
    \item $\A \in A(x)$, and
    \item $\B\in A(y)$,
\end{itemize}
where $x:=\vstart(P)$ and $y$ is the other endpoint of $e_1$. For $k\geq 2$ this means that $x$ is the endpoint of $e_1$ not contained in $e_2$ and that $y$ is the endpoint shared with $e_2$; for $k=1$ we have $y=\vend(P)$, so that $x$ and $y$ are determined by the orientation of $e_1$, which is part of the path chain.
To indicate the initial edge and the vertex shared between the first two edges of $P$, we refer to $P$ as the \emph{$\ABtuple$-bichromatic path chain on $(e_1,y)$}. We write $\vend(P)$ and $\eend(P)$ for the last vertex and the last edge of $P$, and $\vstart(P)=x$ and $\estart(P)=e_1$ for the first vertex and the first edge of $P$. We say $P$ is \emph{maximal} under $\chi$ if $\{\A,\B\} \cap A(\vend(P),\chi) \neq \emptyset$.
We may refer to $\bichrom(P)$ as the \emph{bichromatic part} of the path chain. Note that $\bichrom(P)=\emptyset$ in case $P$ is a bichromatic path chain of length $1$.
We say a path chain $P$ is a \emph{bichromatic path chain} if it is an $\ABtuple$-bichromatic path chain for some $\alpha, \beta \in [\Delta+1]$ satisfying $\alpha \neq \beta$.
\Cref{fig:pathchain} depicts a bichromatic path chain.

\begin{figure}[t]
    \centering
    \begin{tikzpicture}[scale=.88,every node/.style={transform shape}]
        \begin{scope}[yshift=3.1cm]
            \node[nodelabel,anchor=east] at (-1.65,0) {$\chi$};

            \node[vertex,label={[nodelabel]90:$x=\vstart(P)$}] (x)  at (0,0) {};
            \node[vertex,label={[nodelabel]90:$y$}]            (y)  at (2.2,0) {};
            \node[vertex,label={[nodelabel]90:$v_2$}]          (v2) at (3.9,0) {};
            \node[vertex,label={[nodelabel]90:$v_3$}]          (v3) at (5.6,0) {};
            \node[vertex,label={[nodelabel]90:$v_{k-1}$}]      (v4) at (7.9,0) {};
            \node[vertex,label={[nodelabel]90:$v_k=\vend(P)$}] (v5) at (9.6,0) {};

            \draw[blank]          (x)  -- node[edgelabel,below=6pt] {$e_1$} (y);
            \draw[colored=colone] (y)  -- node[colorlabel=colone] {$\A$} (v2);
            \draw[colored=coltwo] (v2) -- node[colorlabel=coltwo] {$\B$} (v3);
            \node[font=\small] at ($(v3)!.5!(v4)$) {$\cdots$};
            \draw[colored=coltwo] (v4)
                -- node[colorlabel=coltwo] {$\B$}
                   node[edgelabel,below=6pt] {$e_k$} (v5);

            \node[nodelabel,text=colone,below=11pt] at (x) {$\A\in A(x)$};
            \node[nodelabel,text=coltwo,below=11pt] at (y) {$\B\in A(y)$};

            \draw[decorate,decoration={brace,amplitude=4pt,mirror},draw=black!55]
                (2.2,-.85) -- (9.6,-.85) node[midway,below=6pt,font=\small] {$\bichrom(P)$};
        \end{scope}

        \draw[shiftto] (4.8,1.6) -- (4.8,.5);
        \node[font=\small,anchor=west] at (5.05,1.05) {$\Shift$};

        \begin{scope}
            \node[nodelabel,anchor=east] at (-1.65,0) {$\Shift(\chi,P)$};

            \node[vertex,label={[nodelabel]-90:$x$}]       (x)  at (0,0) {};
            \node[vertex,label={[nodelabel]-90:$y$}]       (y)  at (2.2,0) {};
            \node[vertex,label={[nodelabel]-90:$v_2$}]     (v2) at (3.9,0) {};
            \node[vertex,label={[nodelabel]-90:$v_3$}]     (v3) at (5.6,0) {};
            \node[vertex,label={[nodelabel]-90:$v_{k-1}$}] (v4) at (7.9,0) {};
            \node[vertex,label={[nodelabel]-90:$v_k$}]     (v5) at (9.6,0) {};

            \draw[colored=colone] (x)  -- node[colorlabel=colone] {$\A$} (y);
            \draw[colored=coltwo] (y)  -- node[colorlabel=coltwo] {$\B$} (v2);
            \draw[colored=colone] (v2) -- node[colorlabel=colone] {$\A$} (v3);
            \node[font=\small] at ($(v3)!.5!(v4)$) {$\cdots$};
            \draw[blank] (v4) -- (v5);
        \end{scope}
    \end{tikzpicture}
    \caption{Shifting an $\ABtuple$-bichromatic path chain $P=(e_1,\dots,e_k)$
    on $(e_1,y)$; dashed edges are uncolored. Under $\chi$ the first edge
    $e_1=xy$ is uncolored and the colors of the remaining edges alternate
    between $\A$ and $\B$, where $\A$ is available at $x$ and $\B$ is available
    at $y$; this is what makes $P$ shiftable. Shifting moves the color of
    $e_{i+1}$ onto $e_i$ for every $i<k$ and leaves $e_k$ uncolored. The
    bichromatic part $\bichrom(P)$ is a bichromatic path, and $P$ is maximal if
    $\{\A,\B\}\cap A(\vend(P))\neq\emptyset$.}
    \label{fig:pathchain}
\end{figure}
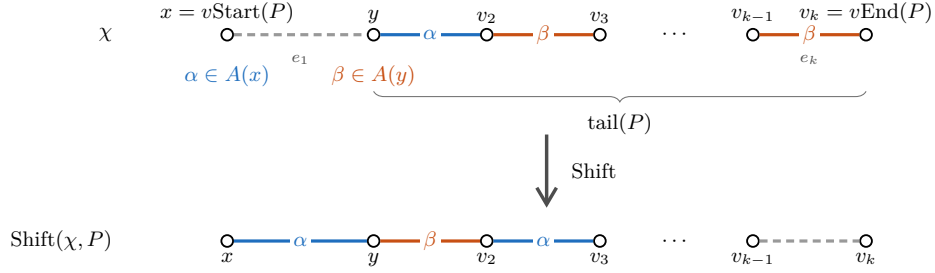

We continue by collecting some observations that follow immediately from the definitions.

\begin{observation}
    A bichromatic path chain as well as a fan chain are $\chi$-shiftable. Any initial segment of a path chain or a fan chain is $\chi$-shiftable as well.
\end{observation}

\begin{observation} \label{obsvn:fan_available_color}
    Let $F=(xz_1,\dots,xz_k)$ be a fan chain under $\chi$, and for $1\leq i\leq k-1$ let $\eta_i$ be the representative available color at $z_i$ in $F$. Then,
    \begin{itemize}
        \item  $A(x,\chi)=A(x,\Shift(\chi,F))$,
        \item  $A(\vend(F),\chi)\subseteq A(\vend(F),\Shift(\chi,F))$, and
        \item  $A(z_i,\chi)\setminus\{\eta_i\}\subseteq A(z_i,\Shift(\chi,F))$ for $1\leq i\leq k-1$, and
        \item  $A(v,\chi)=A(v,\Shift(\chi,F))$ for every $v\notin V(F)$.
    \end{itemize}
\end{observation}

\begin{observation} \label{obsvn:path_available_color}
    Let $P=(v_1v_2,v_2v_3,\dots ,v_{i-1}v_i)$ be a bichromatic path chain under $\chi$. Then $A(v_r,\chi)=A(v_r,\Shift(\chi,P))$ if $3 \leq r \leq i-2$.
\end{observation}
We now define a useful type of chain called a ($1$-step) \emph{Vizing Chain}, obtained by concatenating a fan chain and a path chain.
\begin{definition}
A \emph{Vizing chain}, or equivalently a \emph{$1$-step Vizing chain}, on an uncolored edge $e=xy$ under a
proper partial coloring $\chi$ of $G$ is a $\chi$-shiftable chain $F+P$, where
$F$ is a fan chain centered at $x$ and $P$
is a bichromatic path chain on $(x\vend(F),\vend(F))$ (possibly of length $1$) under $\Shift(\chi,F)$.
\end{definition}

If we concatenate multiple Vizing chains, we obtain the notion of \emph{$i$-step Vizing chains}.

\begin{definition}
An \emph{$i$-step Vizing} chain under a proper partial coloring $\chi$ of $G$ is a
$\chi$-shiftable chain $C$ of the form $F_1 + P_1 + F_2 + P_2 + \dots + F_i + P_i$,
where $F_1+P_1$ is a Vizing chain under $\chi$, and for each $2\leq j\leq i$, $F_j + P_j$ is a Vizing chain on $\eend(P_{j-1})$ under $\Shift(\chi,F_1+P_1+\dots+F_{j-1}+P_{j-1})$. Moreover, if the length of each $P_j$ is at most $\len+1$, we call $C$ an $i$-step Vizing chain \emph{of step length $\len$}. We say $C$ is \emph{terminating} if the endpoints of $\eend(P_i)$ share an available color under $\operatorname{Shift(\chi,C)}$.
\end{definition}

Note that for $1\leq j\leq i$, $\bichrom(P_j)$ is a bichromatic path under $\Shift(\chi,F_1+P_1+\dots+F_{j-1}+P_{j-1}+F_j)$.
\Cref{fig:vizingchain} shows a $2$-step Vizing chain and the effect of shifting it.

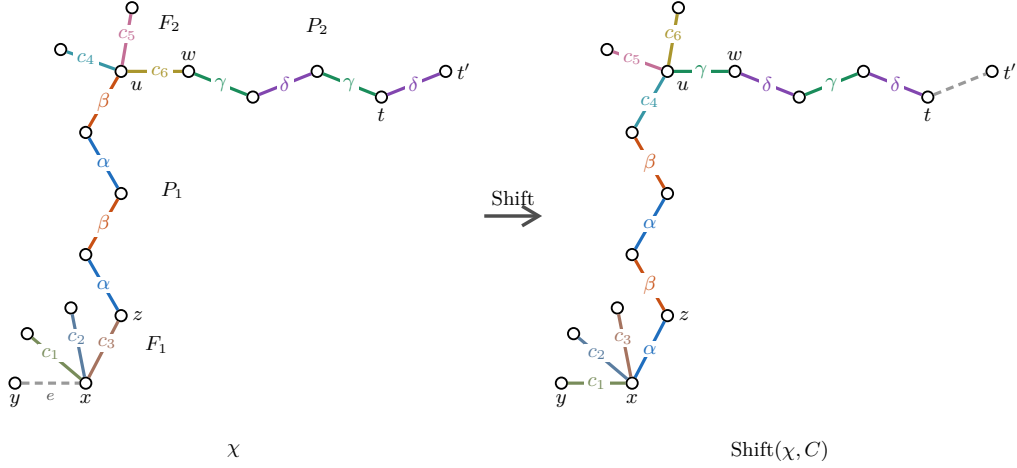
\begin{figure}[t]
    \centering
    \begin{tikzpicture}[scale=.85,every node/.style={transform shape}]
        \begin{scope}
            \node[vertex,label={[nodelabel]-90:$y$}]  (y)  at (0,0) {};
            \node[vertex,label={[nodelabel]-90:$x$}]  (x)  at (1.1,0) {};
            \node[vertex]                             (z2) at (.19,.77) {};
            \node[vertex]                             (z3) at (.87,1.17) {};
            \node[vertex,label={[nodelabel]0:$z$}]    (z)  at (1.65,1.05) {};
            \node[vertex]                             (p1) at (1.1,2.0) {};
            \node[vertex]                             (p2) at (1.65,2.95) {};
            \node[vertex]                             (p3) at (1.1,3.9) {};
            \node[vertex,label={[nodelabel]-45:$u$}]  (u)  at (1.65,4.85) {};
            \node[vertex]                             (w1) at (.71,5.19) {};
            \node[vertex]                             (w2) at (1.82,5.84) {};
            \node[vertex,label={[nodelabel]90:$w$}]   (w)  at (2.7,4.85) {};
            \node[vertex]                             (t1) at (3.7,4.45) {};
            \node[vertex]                             (t2) at (4.7,4.85) {};
            \node[vertex,label={[nodelabel]-90:$t$}]  (t3) at (5.7,4.45) {};
            \node[vertex,label={[nodelabel]0:$t'$}]   (t4) at (6.7,4.85) {};

            \draw[blank]            (y)  -- node[edgelabel,below=2pt] {$e$} (x);
            \draw[colored=colfanone]   (x)  -- node[colorlabel=colfanone,pos=.62] {$c_1$} (z2);
            \draw[colored=colfantwo]   (x)  -- node[colorlabel=colfantwo,pos=.62] {$c_2$} (z3);
            \draw[colored=colfanthree] (x)  -- node[colorlabel=colfanthree,pos=.62] {$c_3$} (z);
            \draw[colored=colone]   (z)  -- node[colorlabel=colone] {$\A$} (p1);
            \draw[colored=coltwo]   (p1) -- node[colorlabel=coltwo] {$\B$} (p2);
            \draw[colored=colone]   (p2) -- node[colorlabel=colone] {$\A$} (p3);
            \draw[colored=coltwo]   (p3) -- node[colorlabel=coltwo] {$\B$} (u);
            \draw[colored=colfanfour]  (u)  -- node[colorlabel=colfanfour,pos=.6] {$c_4$} (w1);
            \draw[colored=colfanfive]  (u)  -- node[colorlabel=colfanfive,pos=.6] {$c_5$} (w2);
            \draw[colored=colfansix]   (u)  -- node[colorlabel=colfansix,pos=.6] {$c_6$} (w);
            \draw[colored=colthree] (w)  -- node[colorlabel=colthree] {$\G$} (t1);
            \draw[colored=colfour]  (t1) -- node[colorlabel=colfour] {$\D$} (t2);
            \draw[colored=colthree] (t2) -- node[colorlabel=colthree] {$\G$} (t3);
            \draw[colored=colfour]  (t3) -- node[colorlabel=colfour] {$\D$} (t4);

            \node[font=\small\itshape] at (2.2,.6)  {$F_1$};
            \node[font=\small\itshape] at (2.45,3.0) {$P_1$};
            \node[font=\small\itshape] at (4.7,5.55) {$P_2$};
            \node[font=\small\itshape] at (2.4,5.6)  {$F_2$};

            \node[nodelabel] at (3.4,-1.05) {$\chi$};
        \end{scope}

        \draw[shiftto] (7.3,2.6) -- node[above=1pt,font=\small] {$\Shift$} (8.2,2.6);

        \begin{scope}[xshift=8.5cm]
            \node[vertex,label={[nodelabel]-90:$y$}]  (y)  at (0,0) {};
            \node[vertex,label={[nodelabel]-90:$x$}]  (x)  at (1.1,0) {};
            \node[vertex]                             (z2) at (.19,.77) {};
            \node[vertex]                             (z3) at (.87,1.17) {};
            \node[vertex,label={[nodelabel]0:$z$}]    (z)  at (1.65,1.05) {};
            \node[vertex]                             (p1) at (1.1,2.0) {};
            \node[vertex]                             (p2) at (1.65,2.95) {};
            \node[vertex]                             (p3) at (1.1,3.9) {};
            \node[vertex,label={[nodelabel]-45:$u$}]  (u)  at (1.65,4.85) {};
            \node[vertex]                             (w1) at (.71,5.19) {};
            \node[vertex]                             (w2) at (1.82,5.84) {};
            \node[vertex,label={[nodelabel]90:$w$}]   (w)  at (2.7,4.85) {};
            \node[vertex]                             (t1) at (3.7,4.45) {};
            \node[vertex]                             (t2) at (4.7,4.85) {};
            \node[vertex,label={[nodelabel]-90:$t$}]  (t3) at (5.7,4.45) {};
            \node[vertex,label={[nodelabel]0:$t'$}]   (t4) at (6.7,4.85) {};

            \draw[colored=colfanone]   (y)  -- node[colorlabel=colfanone] {$c_1$} (x);
            \draw[colored=colfantwo]   (x)  -- node[colorlabel=colfantwo,pos=.62] {$c_2$} (z2);
            \draw[colored=colfanthree] (x)  -- node[colorlabel=colfanthree,pos=.62] {$c_3$} (z3);
            \draw[colored=colone]   (x)  -- node[colorlabel=colone] {$\A$} (z);
            \draw[colored=coltwo]   (z)  -- node[colorlabel=coltwo] {$\B$} (p1);
            \draw[colored=colone]   (p1) -- node[colorlabel=colone] {$\A$} (p2);
            \draw[colored=coltwo]   (p2) -- node[colorlabel=coltwo] {$\B$} (p3);
            \draw[colored=colfanfour]  (p3) -- node[colorlabel=colfanfour] {$c_4$} (u);
            \draw[colored=colfanfive]  (u)  -- node[colorlabel=colfanfive,pos=.6] {$c_5$} (w1);
            \draw[colored=colfansix]   (u)  -- node[colorlabel=colfansix,pos=.6] {$c_6$} (w2);
            \draw[colored=colthree] (u)  -- node[colorlabel=colthree] {$\G$} (w);
            \draw[colored=colfour]  (w)  -- node[colorlabel=colfour] {$\D$} (t1);
            \draw[colored=colthree] (t1) -- node[colorlabel=colthree] {$\G$} (t2);
            \draw[colored=colfour]  (t2) -- node[colorlabel=colfour] {$\D$} (t3);
            \draw[blank] (t3) -- (t4);

            \node[nodelabel] at (3.4,-1.05) {$\Shift(\chi,C)$};
        \end{scope}
    \end{tikzpicture}
    \caption{A $2$-step Vizing chain $C=F_1+P_1+F_2+P_2$ on the uncolored edge
    $e=xy$, and the coloring obtained by shifting it; dashed edges are
    uncolored. The fan $F_1$ is centered at $x$ and ends at $z=\vend(F_1)$; the
    $\ABtuple$-bichromatic path chain $P_1$ continues from $z$ and ends at
    $u=\vend(P_1)$; the fan $F_2$ is a fan on $\eend(P_1)$ centered at $u$ and
    ends at $w=\vend(F_2)$; and the $\GDtuple$-bichromatic path chain $P_2$
    closes the chain with $\eend(P_2)=tt'$. The edges drawn in muted
    tones carry the representative available colors $c_1,\dots,c_6$ of the
    two fans. Shifting $C$ moves the color of every edge of
    $C$ onto its predecessor, so that $e$ becomes colored and $tt'$ becomes
    uncolored; $C$ is terminating if $t$ and $t'$ share an available color
    under $\Shift(\chi,C)$.}
    \label{fig:vizingchain}
\end{figure}

The following observation gives a sufficient condition for an $i$-step Vizing chain to be terminating.

\begin{observation} \label{obsvn:maximal_implies_terminating}
    Consider an $i$-step Vizing chain $C=F_1+P_1+\dots+F_i+P_i$ under $\chi$ for $i\geq 1$. If $P_i$ is a maximal bichromatic path chain under $\Shift(\chi,F_1+P_1+\dots+F_i)$, then $C$ is terminating.
\end{observation}

\begin{proof}
If $P_i$ is of length $1$, the observation follows trivially from $P_i$ being maximal.

Now, let $e_{k-1}$ and $e_k=\eend(P_i)$ be the last two edges of $P_i$, and let $v_k$ and $v_{k+1} = \vend(P_i)$ be its last two vertices. Let $\alpha$ and $\beta$ be the colors s.t.\ $P_i$ is $\ABtuple$-bichromatic under $\chi'=\Shift(\chi,F_1+P_1+\dots+F_{i-1}+P_{i-1}+F_i)$, and without loss of generality, let us assume that $\chi'(e_{k})=\alpha$ and $\chi'(e_{k-1})=\beta$.
If $P_i$ is maximal under $\chi'$, then $\card{A(v_{k+1},\chi') \cap \set{\alpha,\beta}} = 1$ since $\eend(P_i)$ is colored $\alpha$, $\beta \in A(v_{k+1},\chi')$.
As we shift the entire chain $C$, and notably the last two edges on the path $P_i$, the color $\beta$ becomes available at $v_{k}$. This is because the two edges incident on $v_k$ that change color in the shift, $e_{k-1}$ and $e_k$, change color from $\beta$ to $\alpha$ and become uncolored, respectively. The color $\beta$ is therefore available on both endpoints of $e_k = \eend(P_i)$ after the shift, making $C$ terminating.
\end{proof}

Note that the converse is not true: the last path in a terminating chain $C$ is not necessarily maximal. For example, this is the case if the endpoints of the last edge on $C$ have an available color in common simply from their incident edges being mostly uncolored.

\Cref{obsvn:maximal_implies_terminating} provides a conceptually simple recipe for extending a partial coloring to an uncolored edge $e$: find a terminating Vizing chain on $e$, apply the natural shift along the Vizing chain (which in particular colors $e$ and uncolors the last edge of the chain), and color the last edge of the chain with some color available at both endpoints (whose existence is guaranteed by the fact that the Vizing chain is terminating).
In particular, we obtain the following observation.

\begin{observation} \label{obsvn:recolor_terminating_chain}
     Consider an $i$-step Vizing chain $C=F_1+P_1+\dots+F_i+P_i$ on $e$ under $\chi$ for $i>0$. Suppose $C$ is terminating with $\A \in A(\Shift(\chi,C),x)\cap A(\Shift(\chi,C),y)$ where $xy=\eend(P_i)$. Then the coloring $\chi'$ defined by
     \[
    \chi'(f) = \begin{cases}
     \Shift(\chi,C)(f) &\quad\text{if } f\in E(C) \setminus \{\eend(P_i)\} \\
     \A &\quad\text{if } f=\eend(P_i) \\
      \chi(f) &\quad\text{otherwise}
     \end{cases}
    \]
 is a proper partial $(\Delta+1)$-edge coloring.
\end{observation}

For $i\geq 1$, we may refer to an $i$-step Vizing chain as a \emph{multi-step Vizing chain} as well.
Now, consider an $i$-step Vizing chain $C=F_1+P_1+\dots+F_j+P_j$ on $e$. If $e_j$ is the last edge of $P_{j-1}$ and $C$ is clear from the context, we may refer to $F_j+P_j$ as the \emph{Vizing chain on $e_j$} since $F_j+P_j$ is a Vizing chain on $e_j$ under $\Shift(\chi,F_1+P_1+\dots+F_{j-1}+P_{j-1})$.
 
We restrict our focus to chains of the form $F_1 + P_1 + F_2 + P_2 + \dots + F_i + P_i$ with the property that the chain does not ``self-intersect''. The notion of this self-intersection is captured by chains defined as \emph{self-avoiding} $i$-step Vizing chains.

\begin{definition} 
Consider an $i$-step Vizing chain $C=F_1+P_1+\dots+F_{i}+P_{i}$ under $\chi$.
We call $C$ \emph{self-avoiding} if 
\begin{enumerate}
    \item for all $1\leq k\leq i-1$, we have $E(P_k)\cap E(F_{k+1}+P_{k+1})=\{\eend(P_{k})\}$,
    \item for all $m, k$ satisfying $3\leq m+2\leq k\leq i$, we have $V(F_m+P_m)\cap V(F_k+P_k)=\emptyset$,
    \item for all $1\leq k\leq i-1$, we have $V(F_k)\cap V(F_{k+1})=\emptyset$, and
    \item for all $1\leq k \leq i$, we have that $\cen(F_k)$ and $\vend(F_k)$ are not $\bichrom(P_k)$-related. (In case $\bichrom(P_k)=\emptyset$, they are \emph{trivially} not related.)
\end{enumerate}
\end{definition}

Note that by the definition of a self-avoiding $i$-step Vizing chain $C=F_1+P_1+\dots+F_i+P_i$, $E(F_{m-1}+P_{m-1})$ and $E(F_m+P_m)$ intersect, but only in $\eend(P_{m-1})$ for all $1<m\leq i$. Moreover edges of $F_{k}+P_{k}$ and $F_{m}+P_{m}$ do not intersect (even in vertices) if $3\leq m+2 \leq k\leq i$. Hence edges in $C$ are not repeated and $C$ is a simple chain.
Although $C$ does not intersect itself in edges, it may intersect itself in vertices although only in restricted ways. If $k\neq m$ for $1\leq k,m\leq i$, $F_k$ and $F_m$ do not intersect each other even in vertices. Two paths $P_k$ and $P_m$ with $k\leq m$ can intersect in vertices only if $m-1=k$. Similarly a path $P_k$ and a fan $F_m$ can intersect only if $k\in \{m-1,m\}$. 

Note that for any multi-step Vizing chain $C$, a vertex not in $V(C)$ has the same set of available colors under $\chi$ and $\Shift(\chi,C)$. If $C$ is also self-avoiding, almost all vertices in $V(C)$ that are not part of a fan chain also have this property. We collect these insights in the following observation.

\begin{observation} \label{obsvn:available_colors_after_shift}
    Let $C=F_1+P_1+\dots+F_i+P_i$ be a self-avoiding $i$-step Vizing chain on $e$. Assume $v\in V(G)\setminus \bigl((\bigcup_{k=1}^{i}V(F_k))\cup V(\eend(P_i))\bigr)$. Then, $A(v,\chi)=A(v,\Shift(\chi,C))$.
\end{observation}
\begin{proof}
By assumption, $v$ belongs to no fan of $C$, and it is not an endpoint of $\eend(P_k)$ for any $k$. Indeed, if $k<i$, then $\eend(P_k)\in E(F_{k+1})$; and if $k=i$, this follows from $v\notin V(\eend(P_i))$.

Consequently, every edge of $C$ incident with $v$ is an interior edge of some bichromatic part $\bichrom(P_k)$. Thus, the edges of $C$ incident with $v$ can be grouped into consecutive pairs $e_j,e_{j+1}$, where each pair lies in the same bichromatic part. Since such a part alternates between two colors, we have $\chi(e_{j+2})=\chi(e_j)$. After shifting, the colors on $e_j$ and $e_{j+1}$ at $v$ change from $\chi(e_j),\chi(e_{j+1})$ to $\chi(e_{j+1}),\chi(e_{j+2})$, that is, to $\chi(e_{j+1}),\chi(e_j)$. Hence the same set of colors is used at $v$ before and after the shift. Therefore, $A(v,\cdot)$ remains unchanged.
\end{proof}

We will now state two lemmas that will be useful for constructing fan chains with certain properties.
The first lemma appeared as part of a proof of Vizing's theorem in \cite{vizing1964}. We rephrase it to our setting.

\begin{lemma} \label{lemma:fan-addition-phase-1}
Let $\chi$ be a proper partial edge coloring, and let $e = xy$ be an uncolored edge. Assume $x$ knows $A(z)$ for $z\in N[x]$.  Then vertex $x$ can internally construct a fan $F=(xz_1=xy,xz_2,\dots,xz_p)$ on $e$ under $\chi$ such that exactly one of the following properties hold.
\begin{enumerate}
    \item \emph{(Common-color property.)} $A(x)\cap A(z_p)\neq \emptyset$, or \label[instance]{instance:fan_construction_phase_1_instance_1}
    \item \emph{(Repeated-color property.)} $A(x)\cap A(z_i)= \emptyset$ for every $i\in [p]$, and there exist distinct colors $\A,\B\in [\Delta+1]$ with $1\leq q < p$ such that $\A\in A(x,\chi)$ and the representative available color at both $z_p$ and $z_q$ is $\B$. \label[instance]{instance:fan_construction_phase_1_instance_2}
\end{enumerate}
\end{lemma}

\begin{proof}
    A fan cannot have both the \commoncolorprop{} and the \repeatedcolorprop{}, since the former asks for $A(x)\cap A(z_p)\neq\emptyset$ while the latter asks for $A(x)\cap A(z_i)=\emptyset$ for every $i\in[p]$. Hence it suffices to construct a fan satisfying at least one of them.

    Fix a color $\A\in A(x)$, which exists since $e$ is uncolored. Vertex $x$ sets $z_1:=y$ and then, for $i=1,2,\dots$, having already chosen $z_1,\dots,z_i$ and assigned a representative available color to each of $z_1,\dots,z_{i-1}$, proceeds according to the first of the following three cases that applies.
    \begin{enumerate}
        \item If $A(x)\cap A(z_i)\neq\emptyset$, then $x$ finishes the construction with $\vend(F)=z_i$, and we set $p:=i$.
        \item If some color of $A(z_i)\setminus\{\A\}$ has already been assigned as the representative available color of $z_k$ for some $k<i$, then $x$ assigns one such color $\B$ as the representative available color of $z_i$ as well and finishes the construction with $\vend(F)=z_i$; we set $p:=i$ and $q:=k$.
        \item Otherwise, $x$ assigns an arbitrary color $\eta_i\in A(z_i)\setminus\{\A\}$ as the representative available color of $z_i$, lets $z_{i+1}$ be the neighbor of $x$ with $\chi(xz_{i+1})=\eta_i$, and continues with $i+1$.
    \end{enumerate}

    The third case is well defined. It is reached only when $A(x)\cap A(z_i)=\emptyset$, so $\A\notin A(z_i)$ and hence $A(z_i)\setminus\{\A\}=A(z_i)$, which is nonempty because $z_i$ has at most $\Delta$ incident edges while there are $\Delta+1$ colors. Moreover $\eta_i\in A(z_i)$ and $A(x)\cap A(z_i)=\emptyset$ give $\eta_i\notin A(x)$, so $x$ does have a neighbor $z_{i+1}$ with $\chi(xz_{i+1})=\eta_i$.

    The process terminates. The third case assigns to $z_i$ a color that is not the representative available color of any $z_k$ with $k<i$, so the colors assigned during the process are pairwise distinct. As there are $\Delta+1$ colors in total, the process reaches the first or the second case within $\Delta+1$ iterations.

    Let $F=(xz_1,xz_2,\dots,xz_p)$ be the resulting chain. Its vertices are pairwise distinct: for $i\geq 1$ the vertex $z_{i+1}$ is determined by $\chi(xz_{i+1})=\eta_i$ and the colors $\eta_1,\dots,\eta_{p-1}$ are pairwise distinct, so $z_{i+1}\neq z_{k+1}$ whenever $i\neq k$; and $z_{i+1}\neq z_1$, since $xz_1=e$ is uncolored whereas $xz_{i+1}$ is colored $\eta_i$. In particular no edge is repeated in $F$, so $F$ is a fan chain on $e$ under $\chi$.

    Finally, we check that $F$ satisfies one of the two properties. If the process stopped in the first case, then $A(x)\cap A(z_p)\neq\emptyset$, which is the \commoncolorprop{}. If it stopped in the second case, then $A(x)\cap A(z_i)=\emptyset$ for every $i\in[p]$, since each of $z_1,\dots,z_{p-1}$ was passed through the third case and $z_p$ was reached without the first case applying. Furthermore $\A\in A(x,\chi)$, the colors $\A$ and $\B$ are distinct because $\B\in A(z_p)\setminus\{\A\}$, the index $q$ satisfies $1\leq q<p$, and $\B$ is the representative available color of both $z_p$ and $z_q$. This is the \repeatedcolorprop{}.
\end{proof}

The next lemma has appeared in different forms in \cite{grebik2020measurable,  christiansen2023power,bernshteyn2022fast, bernshteyn2025fast}. 

\begin{lemma} \label{lemma:fan-addition-phase-j}
Let $\chi$ be a proper partial edge coloring, and let $e = xy$ be an uncolored edge. Assume $x$ knows $A(z)$ for $z\in N[x]$.
Let $\A\in A(x) \text{ and }\B\in A(y)$.
Then $x$ can internally construct a fan $F=(xz_1=xy,xz_2,\dots,xz_p)$ on $e$ under $\chi$ such that exactly one of the following properties hold.
\begin{enumerate}
    \item \emph{(Common-color property.)} $A(x)\cap A(z_p)\neq \emptyset$, or \label[instance]{instance:fan_construction_phase_j_instance_1}
    \item \emph{(Repeated-color property.)} $A(x)\cap A(z_i)= \emptyset$ for every $i\in [p]$, and
    there exist distinct colors $\G,\D\in[\Delta+1]\setminus \ABSet$ and an index $q$ with  $1\leq q < p$ such that $\G\in A(x)$ and the representative available color at both $z_p$ and $z_q$ is $\D$,  or \label[instance]{instance:fan_construction_phase_j_instance_2}
    \item \emph{(Inherited-color property.)} $A(x)\cap A(z_i)= \emptyset$ for every $i\in [p]$, $z_p\neq y$, $A(z_p)=\{\B\}$, $\A\neq \B$. \label[instance]{instance:fan_construction_phase_j_instance_3}
\end{enumerate}
\end{lemma}

\begin{proof}
    We construct a fan $F$ below and then verify that it satisfies exactly one of the three properties.

    Note that $|A(x)|>1$ and $|A(y)|>1$, since $e=xy$ is uncolored. Vertex $x$ sets $z_1:=y$ and then, for $i=1,2,\dots$, having already chosen $z_1,\dots,z_i$ and assigned a representative available color to each of $z_1,\dots,z_{i-1}$, proceeds according to the first of the following four cases that applies.
    \begin{enumerate}
        \item If $A(x)\cap A(z_i)\neq\emptyset$, then $x$ finishes the construction with $\vend(F)=z_i$, and we set $p:=i$.
        \item If some color of $A(z_i)\setminus\ABSet$ has already been assigned as the representative available color of $z_k$ for some $k<i$, then $x$ assigns one such color $\D$ as the representative available color of $z_i$ as well and finishes the construction with $\vend(F)=z_i$; we set $p:=i$ and $q:=k$.
        \item If $A(z_i)\setminus\ABSet=\emptyset$, then $x$ finishes the construction with $\vend(F)=z_i$, and we set $p:=i$.
        \item Otherwise, $x$ assigns an arbitrary color $\eta_i\in A(z_i)\setminus\ABSet$ as the representative available color of $z_i$, lets $z_{i+1}$ be the neighbor of $x$ with $\chi(xz_{i+1})=\eta_i$, and continues with $i+1$.
    \end{enumerate}

    The fourth case is well defined. It is reached only when $A(z_i)\setminus\ABSet\neq\emptyset$, so the color $\eta_i$ exists. It is also reached only when $A(x)\cap A(z_i)=\emptyset$, so $\eta_i\in A(z_i)$ gives $\eta_i\notin A(x)$, and hence $x$ does have a neighbor $z_{i+1}$ with $\chi(xz_{i+1})=\eta_i$.

    The process terminates. The fourth case assigns to $z_i$ a color that is not the representative available color of any $z_k$ with $k<i$, so the colors assigned during the process are pairwise distinct. As there are $\Delta+1$ colors in total, the process reaches one of the first three cases within $\Delta+1$ iterations.

    Let $F=(xz_1,xz_2,\dots,xz_p)$ be the resulting chain. Its vertices are pairwise distinct: for $i\geq 1$ the vertex $z_{i+1}$ is determined by $\chi(xz_{i+1})=\eta_i$ and the colors $\eta_1,\dots,\eta_{p-1}$ are pairwise distinct, so $z_{i+1}\neq z_{k+1}$ whenever $i\neq k$; and $z_{i+1}\neq z_1$, since $xz_1=e$ is uncolored whereas $xz_{i+1}$ is colored $\eta_i$. In particular no edge is repeated in $F$, so $F$ is a fan chain on $e$ under $\chi$.

    We now determine which property $F$ satisfies. If the process stopped in the first case, then $A(x)\cap A(z_p)\neq\emptyset$, which is the \commoncolorprop{}. In the two remaining cases we have $A(x)\cap A(z_i)=\emptyset$ for every $i\in[p]$, since each of $z_1,\dots,z_{p-1}$ was passed through the fourth case and $z_p$ was reached without the first case applying. In particular $\A\notin A(z_i)$ for every $i\in[p]$ because $\A\in A(x)$, and $\B\notin A(x)$ because $\B\in A(y)=A(z_1)$.

    Suppose the process stopped in the second case. As $|A(x)|>1$ and $\B\notin A(x)$, we may pick a color $\G\in A(x)\setminus\ABSet$. The color $\D$ lies in $A(z_p)\setminus\ABSet$, so $\D\notin\ABSet$, and $\G\neq\D$ because $\G\in A(x)$, $\D\in A(z_p)$ and $A(x)\cap A(z_p)=\emptyset$. Together with $1\leq q<p$ and the fact that $\D$ is the representative available color of both $z_p$ and $z_q$, this is the \repeatedcolorprop{}.

    Suppose instead that the process stopped in the third case, so that $A(z_p)\subseteq\ABSet$. Since $\A\notin A(z_p)$, and $A(z_p)$ is nonempty because $z_p$ has at most $\Delta$ incident edges while there are $\Delta+1$ colors, we get $A(z_p)=\{\B\}$, and in particular $\A\neq\B$. It remains to check that $z_p\neq y$. The third case cannot apply at $i=1$: there $\A\notin A(z_1)$ and $|A(z_1)|>1$, so $A(z_1)\setminus\ABSet$ is nonempty. Hence $p\geq 2$ and $z_p\neq z_1=y$, so the \inheritedcolorprop{} holds.

    It remains to observe that $F$ cannot satisfy two of the three properties at once. The \commoncolorprop{} is incompatible with each of the other two, since it asks for $A(x)\cap A(z_p)\neq\emptyset$ while they ask for $A(x)\cap A(z_i)=\emptyset$ for every $i\in[p]$. The \repeatedcolorprop{} and the \inheritedcolorprop{} are incompatible as well: the former asserts that the representative available color at $z_p$ is a color $\D\notin\ABSet$, and a representative available color at a vertex is by definition available at that vertex, so it gives $\D\in A(z_p)\setminus\ABSet$; the latter asserts $A(z_p)=\{\B\}\subseteq\ABSet$. Hence $F$ satisfies exactly one of the three properties.
\end{proof}

\section{The Algorithm}
In this section, we describe our overall algorithm for obtaining~\Cref{deterministic_polylog}.
Our algorithm proceeds in iterations, where in each iteration an $\Omega(1/(\poly(\Delta) \log^2 n))$-fraction of the remaining uncolored edges is colored, while all colored edges remain colored (though they might change their colors).
In particular, this implies that after $O(\poly(\Delta)\log^3 n)$ iterations all edges are colored.
Since, after each iteration, the obtained edge coloring is proper, the resulting coloring is a proper $(\Delta + 1)$-edge coloring as desired.

Therefore, all we have to do to obtain~\Cref{deterministic_polylog} is to describe how we can, in $\tilde{O}(\poly(\Delta)\log^2 n)$ rounds, turn any partial $(\Delta + 1)$-edge coloring into another partial $(\Delta + 1)$-edge coloring where the set of uncolored edges has shrunk in size by an $\Omega(1/(\poly(\Delta)\log^2 n))$-fraction.
We achieve this goal in three steps.

The first step is to construct, on each uncolored edge, a terminating self-avoiding multi-step Vizing chain.
We describe in \Cref{sec:construct-chains} how we construct these chains.

The second step is to choose a subset of the set of these Vizing chains such that
\begin{enumerate}
    \item\label{item:no-two} no two chosen Vizing chains intersect in a vertex and
    \item\label{item:at-least} at least an $\Omega(1/(\poly(\Delta)\log^2 n))$-fraction of all computed Vizing chains is chosen.
\end{enumerate}
On a high level, we choose such a subset by computing an independent set on the virtual graph that contains all computed Vizing chains as nodes and has an edge between two nodes if the corresponding Vizing chains intersect in a vertex.
In \Cref{sec:comp-indep-chains}, we describe this part of our algorithm in detail (and prove that the chosen subset is as desired).

The third step is to change the partial coloring by recoloring the edges in the chosen ``independent'' terminating Vizing chains while still keeping the partial coloring proper.
In particular, this includes coloring the uncolored edge at the start of the Vizing chain.
We perform the recoloring in the way commonly used for Vizing chains in the literature: each edge in a chosen Vizing chain $C_e$ except for the last edge is colored with the color of the subsequent edge in $C_e$ under the current coloring.
Since $C_e$ is terminating, the endpoints of the last edge of $C_e$ will have a common available color after this shifting of colors, which is then assigned to the edge. Since the recoloring of $C_e$ only affects the lists of available colors for the vertices lying on $C_e$, Property~\ref{item:no-two} of our selection of Vizing chains guarantees that all those Vizing chains can be recolored at the same time without creating conflicts.
Property~\ref{item:at-least} then guarantees that the number of uncolored edges has decreased additively by a $\Omega(1/(\poly(\Delta)\log^2 n))$-fraction.

\Cref{alg:overall} summarizes the resulting algorithm.
Its third line needs no communication beyond what the first line already provides: by \Cref{obsvn:recolor_terminating_chain}, every edge of $C_e$ other than the last one takes the color of its successor in $C_e$, which its endpoints learn while $C_e$ is constructed, and the last edge takes a color available at both of its endpoints.

\begin{algorithm}[!htbp]
\caption{The overall algorithm.}
\label{alg:overall}
\begin{algorithmic}[1]
\State $\chi\gets$ the empty partial $(\Delta+1)$-edge coloring and $U\gets E(G)$
\While{$U\neq\emptyset$}
    \State construct a terminating self-avoiding multi-step Vizing chain $C_e$ on $e$, for every $e\in U$ \Comment{\Cref{sec:implementation-phases}}
    \State compute $I\subseteq U$ with $\card{I}\in\Omega(\card{U}/(\poly(\Delta)\log^2 n))$ such that the chains $C_e$ with $e\in I$ are pairwise vertex-disjoint \Comment{\Cref{sec:comp-indep-chains}}
    \State recolor $C_e$ as in \Cref{obsvn:recolor_terminating_chain}, for every $e\in I$ in parallel, and set $U\gets U\setminus I$
\EndWhile
\end{algorithmic}
\end{algorithm}

\subsection{Constructing Self-Avoiding Multi-Step Vizing Chains}\label{sec:construct-chains}
In this section, we describe the first of the discussed three steps, i.e., the part of our overall algorithm that produces for each uncolored edge $e$ a terminating self-avoiding multi-step Vizing chain on $e$.
We will denote this subroutine by $\fA$.
Throughout this section, we will assume that we are given a partially $(\Delta + 1)$-edge-colored graph $G$.
Let $\chi$ denote this partial $(\Delta+1)$-edge coloring and let $\len$ and $\numphases$ be integers to be defined later.

 We start by providing an informal high-level overview of our subroutine $\fA$.
 As already outlined above, the task of $\fA$ is to produce a terminating self-avoiding multi-step Vizing chain for each uncolored edge $e$.
 Each of these Vizing chains will be of step length $\ell$.
 In the following, we describe how we obtain the desired terminating self-avoiding multi-step Vizing chain for a fixed uncolored edge $e$.
 Note that $\fA$ computes the desired Vizing chains for all uncolored edges simultaneously.

 Fix an uncolored edge $e$.
 The main part of $\fA$ consists of $\numphases=O(\log n)$ \emph{construction phases}.
 In construction phase $i$, we construct multiple $i$-step Vizing chains on $e$ of step length $\len$.
 More precisely, these $i$-step Vizing chains are obtained by appending a $1$-step self-avoiding Vizing\footnote{Technically speaking, this chain we append is not (necessarily) a Vizing chain under $\chi$ but a Vizing chain under the partial coloring obtained from $\chi$ by a shift along the respective $(i-1)$-step self-avoiding Vizing chain.} chain of step
 length $\len$ to each of the $(i-1)$-step self-avoiding Vizing chains of step length $\len$ on $e$ constructed by the end of phase $i-1$.
 However, for each such $(i - 1)$-step Vizing chain, we do not extend it to a $i$-step Vizing chain in only a single way: instead, we append multiple $1$-step Vizing chains to such a $(i - 1)$-step Vizing chain, thereby transforming each $(i - 1)$-step Vizing chain on $e$ into a collection of $i$-step Vizing chains on $e$.
 For each of the mentioned $(i - 1)$-step Vizing chains, these appended different $1$-step Vizing chains are quite similar to each other: they are obtained from a common $1$-step Vizing chain $F + P$ by truncating $P$ at different vertices (which we call a \emph{candidate chain for $e$}).
 We note that not all of these $i$-step Vizing chains may be self-avoiding; we remove all that are not self-avoiding and iterate only on the self-avoiding ones in construction phase $i + 1$.
 
 We show that at least one of two useful properties must hold at the end of construction phase $i$: the number of computed $i$-step self-avoiding Vizing chains on $e$ is an $\Omega(\poly(\Delta))$-factor larger than the number of computed $(i-1)$-step self-avoiding Vizing chains on $e$, or at least one of the computed self-avoiding $i$-step Vizing chains on $e$ is terminating. 
 Using these properties, we show that after $O(\log n)$ construction phases, a terminating self-avoiding multi-step Vizing chain on $e$ must have been found, as desired.
 In fact, our algorithm may produce many such Vizing chains on each edge $e$; hence, Algorithm $\fA$ contains a postprocessing step in which for each edge $e$, precisely one of the computed terminating self-avoiding multi-step Vizing chains on $e$ is selected.
 
On the technical side, our approach requires to define a number of objects, which we outline in the following (still for a fixed edge $e$).
In construction phase $i$, we construct sets $R_i(e,\len)$ and $S_i(e,\len)$ which are subsets of $V(G)$.
A vertex $x$ is in set $R_i(e,\len)$ if it lies on one of the aforementioned path chains $P$ that we truncate to obtain our collection of $i$-step Vizing chains.
The set $S_i(e,\len)$ is the subset of $R_i(e,\len)$ containing all those vertices at which the respective path chain $P$ is truncated, i.e., the vertices in $S_i(e,\len)$ are precisely the endpoints of the $i$-step Vizing chains that have been constructed at the end of construction phase $i$. 
In other words, $S_i(e,\len)$ contains precisely those vertices at which, in construction phase $i + 1$, one of the aforementioned $1$-hop Vizing chains $F + P$ is appended (and possibly $P$ truncated) to obtain an $(i + 1)$-step Vizing chain.

We will now partition $S_i(e,\len)$ into three (disjoint) subsets $D_i(e,\len)$, $B_i(e,\len)$, and
$Q_i(e,\len)$.
To which of the sets a vertex $v \in S_i(e,\len)$ is assigned depends on whether the $(i + 1)$-step Vizing chain $C$ obtained by appending the respectively considered $1$-hop Vizing chain $F + P$ at $v$ is overlapping and/or terminating.
More precisely, $v$ is assigned to $B_i(e,\len)$ if $C$ is not self-avoiding\footnote{In fact, the condition under which a vertex $v$ may be added to $B_i(e,\len)$ is even weaker. We add  vertex $v$ to $B_i(e,\len)$ even if $C$ is only ``potentially'' not self-avoiding.}, to $D_i(e,\len)$ if $C$ is self-avoiding and terminating, and to $Q_i(e,\len)$ otherwise.
Note that, due to the dependency of the assignment on $(i + 1)$-step Vizing chains, this partition of the set $S_i(e,\len)$ can only be computed in construction phase $i + 1$ of $\fA$, not already in construction phase $i$.

Using the defined objects, we can now more formally state our argument that shows that after $O(\log n)$ construction phases a terminating self-avoiding Vizing chain on $e$ must have been found.
Concretely, we prove that if $D_i(e,\len)$ is empty, then $|Q_i(e,\len)|$ is an $\Omega(\poly(\Delta))$-factor larger than $|Q_{i - 1}(e,\len)|$ (which in turn is shown by proving that $|S_i(e,\len)|$ is an $\Omega(\poly(\Delta))$-factor larger than $|S_{i - 1}(e,\len)|$ and that $|B_i(e,\len)|$ is small compared to $|S_i(e,\len)|$).
As, for any $i$, $|Q_i(e,\len)|$ cannot exceed $n$, it follows that $D_j(e,\len) \neq \emptyset$ for some $j\leq \numphases$, which implies that our algorithm finds a terminating self-avoiding $(j+1)$-step Vizing chain of step length $\len$ on $e$ for some $j \in O(\log n)$.

We conclude with the postprocessing step mentioned above, in which each uncolored edge $e$ selects one of the terminating self-avoiding multi-step Vizing chains constructed on $e$.
Each vertex knows only the $1$-step Vizing chain that it appended, so the selected chain has to be assembled from these local pieces.
We do so by marking, for each $e$, those vertices whose appended chain can be continued to a terminating multi-step Vizing chain on $e$.
Initially, the marked vertices are exactly those in the sets $D_i(e,\len)$, as such a vertex already knows that its appended chain completes a terminating chain on $e$.
The remaining marks are propagated backwards through the construction phases, in $\numphases-1$ \emph{marking phases} indexed in decreasing order.
In marking phase $i$, a vertex in $S_i(e,\len)$ becomes marked if its appended chain reaches a marked vertex of $S_{i+1}(e,\len)$, which informs it of the mark along that chain.
Running the marking phases in decreasing order of the index ensures that the marks in $S_{i+1}(e,\len)$ are already final when marking phase $i$ begins, so a single phase per index suffices.
Following the marks forward from $e$ then picks out a single terminating multi-step Vizing chain $C_e$ for each uncolored edge $e$.
More precisely, $C_e$ is the concatenation of the appended chains of the marked vertices met along the way, each truncated at the next such vertex.
Finally, a further $\numphases$ \emph{relay phases} propagate this information along $C_e$, so that every vertex of $C_e$ learns that it lies on $C_e$ and receives enough information to later shift $C_e$.
We describe these steps in \Cref{subsec:selecting-learning}.

\subsection{Selecting Independent Chains}

The previous section detailed our algorithm for finding a terminating multi-step
Vizing chain for each uncolored edge in the graph under a given
partial $(\Delta+1)$-edge coloring $\chi$.  Each of these multi-step Vizing
chains is a recipe to recolor some edges in the graph so as to obtain
a new partial $(\Delta+1)$-edge coloring with one less uncolored edge.
Throughout this section, for any uncolored edge $e$ in the given
partial coloring $\chi$, let $C_e$ denote the terminating multi-step Vizing chain
computed by the previous algorithm for $e$.

The previous section focused on building terminating self-avoiding multi-step Vizing
chains. This self-avoiding property ensures that the
multi-step Vizing chains we found can actually be used to shift colors
along the chains, i.e., so that by recoloring edges in $C_e$ we can
extend the current partial coloring $\chi$ to $e$. However, while the
multi-step Vizing chains we found are self-avoiding, they can overlap
significantly, which prohibits using all our multi-step Vizing chains
simultaneously to remove all uncolored edges at once.  Our goal in
this section is to take care of these intersections between the chains
that we computed, that is, to extract from our set of self-avoiding multi-step
Vizing chains a somewhat large subset of Vizing chains that
additionally avoid each other.

We say two multi-step Vizing chains are independent if they do not overlap on any
vertex. This ensures that if performing the recoloring of both chains at
the same time we do not have to worry about obtaining an invalid
partial coloring, and we do not have to make choices such as which
color to give an edge that is in the intersection of two chains\footnote{Some chains that overlap on one or more vertices can be recolored simultaneously without creating an invalid partial coloring. One could define these chains to be independent, but it is simpler and without significant loss to define independence more stringently as simply not overlapping on vertices.}.

The task can equivalently be described as finding a large independent set in the virtual multi-graph $\chainG = (V_{\chainG},E_{\chainG})$ whose nodes represent our computed multi-step Vizing chains, and where edges represent overlap between chains. Two multi-step Vizing chains overlapping on multiple vertices in the original graph $G$ results in multiple edges between the two nodes representing them in $\chainG$. As each node in $\chainG$ represents a subgraph of $G$ of diameter $O(\ell \log n)$, in LOCAL, simulating an arbitrary algorithm on the virtual graph $\chainG$ using rounds of communication over $G$ would only come at a multiplicative cost of $O(\ell \log n)$. As we only have access to CONGEST communication over $G$, we pay a higher cost to simulate an algorithm on $\chainG$. We also need to exploit the specificities of the algorithm we simulate to keep the simulation overhead low. Notably, an important property is that the algorithm we simulate should work with small messages, e.g., in CONGEST.

How large an independent set of chains can we hope to compute? The algorithm that found the self-avoiding multi-step Vizing chains bounded their overlap. More precisely, each of our multi-step Vizing chains overlaps with at most $O(\Delta^6 \ell \log^2 n)$ other Vizing chains, even when counting overlaps with multiplicity. A graph with $N$ nodes of maximum degree $D$ necessarily contains an independent set of size at least $N/(D+1)$. For our use case, there is no drawback to computing an independent set of size $\Omega(N/D)$ instead of $N/(D+1)$, and can be done slightly faster with the current state-of-the-art. Concretely, we compute a set of independent Vizing chains that represents a $1/\Omega(\Delta^6 \ell \log^2 n)$ fraction of all the self-avoiding Vizing chains we previously found.

We simulate an algorithm by Faour et al.~\cite{FGGKR_talg25_generalized_rounding} to compute our large-ish independent set.
A feature of this algorithm is that it requires smaller messages if starting from a coloring of the graph with few colors. More precisely, the algorithm only needs messages of size $O(\log \zeta)$ if given a valid $\zeta$-coloring.
Our first step is thus to color $\chainG$ with $\poly(\Delta,\log n)$ colors.
To do so, we exploit the fact that the classic algorithm by Linial~\cite{linial1987distributive} for $O(\Delta^2)$-coloring can be implemented as a sequence of broadcast and aggregation operations (that is: where each node sends the same message to all its neighbors, and only needs to know something like the sum of its neighbors' messages instead of all their messages).
This was highlighted in a recent work by Barenboim and Goldenberg~\cite{BG_disc24} where this property was used to adapt Linial's classic algorithm to the distance-2 setting, with congestion, and in a work by Flin, Halldórsson, and Nolin~\cite{FHN_disc24_virtual_graphs} with randomized coloring of virtual graphs as application.
Such broadcast/aggregation operations are much faster to perform than to simulate a proper round of CONGEST over $\chainG$.

Once the coloring is done, we simulate the algorithm by Faour et al.~rather naïvely, by directly simulating CONGEST communication over $\chainG$. The algorithm is kept somewhat efficient by the fact that we only need the chains to deliver small messages of size $O(\log \Delta + \log \log n)$ to each other. 

\subsection{Runtime of the Overall Algorithm}
In this section, we prove the runtime of our overall algorithm (and thereby \Cref{deterministic_polylog}) via the statements below that we prove in later sections.
Throughout, $U$ denotes the set of all uncolored edges and, for each uncolored edge $e$, the terminating self-avoiding multi-step Vizing chain on $e$ computed during the construction phases is denoted by $C_e$.
We bound the cost of the three steps performed in each iteration of \Cref{alg:overall} in turn.

\subparagraph*{Constructing the chains.}
First, we consider the runtimes of construction phase $1$ and construction phases $2$ to $T$.
\begin{restatable}{lemma}{RuntimePhaseOne} \label{lemma:runtime_1st_phase}
    Construction phase $1$ runs in $O(\Delta^6\len)$ rounds in CONGEST.
\end{restatable}

\begin{restatable}{lemma}{RuntimePhaseJ}\label{lemma:runtime_phase_j}
      For $2\leq j\leq T$, construction phase $j$ runs in $O(\Delta^7\len)$ rounds in CONGEST.
\end{restatable}

Consider next the runtime of marking phases $T-2$ to $0$ and relay phases $1$ to $T$.

\begin{restatable}{lemma}{RuntimeMarkingPhase} \label{lemma:marking_phase_runtime}
    Any marking phase runs in $O(\Delta^6\len)$ rounds in CONGEST.
\end{restatable}

\begin{restatable}{lemma}{RuntimeRelayPhase} \label{lemma:RuntimeRelayPhase}
    Each relay phase runs in $O(\Delta^6\len)$ rounds in CONGEST.
\end{restatable}
\Cref{lemma:runtime_1st_phase} is proved in \Cref{subsec:runtime-phase-1}, \Cref{lemma:runtime_phase_j} in \Cref{sec:runtime-phase-j}, \Cref{lemma:marking_phase_runtime} in \Cref{subsec:marking-phase-properties}, and \Cref{lemma:RuntimeRelayPhase} in \Cref{subsec:runtime-learning}.
Since there are $T-1$ construction phases of the second kind, $T-1$ marking phases and $T$ relay phases, and since $T=\ceil{\log_2 n}$, the four lemmas together bound the cost of the first line.

\begin{corollary} \label{cor:construction-runtime}
    Computing a terminating self-avoiding multi-step Vizing chain $C_e$ for every $e\in U$ takes $O(\Delta^7\len\log n)$ rounds in CONGEST.
\end{corollary}

\subparagraph*{Selecting independent chains.}
\begin{restatable}{lemma}{LemLargeISChainGraph}
  \label{lem:large-is-chain-graph}
  There is an algorithm of complexity $O(\Delta^{12} \len^2 \log^2 n (\log^3 \log n + \log^3 \Delta))$ that computes a subset $I \subseteq U$ of size  $\Omega(\frac{\card{U}}{\Delta^6 \ell \log^2 n})$ of uncolored edges s.t.\ for any two edges $e,f \in I$, their chains $C_e$ and $C_f$ are independent.
\end{restatable}
\Cref{lem:large-is-chain-graph} is proved in \Cref{ssec:chain-graph-simulation}.

\subparagraph*{Shifting the chains.}
\begin{observation} \label{obsvn:shift-runtime}
    Once the relay phases are concluded, recoloring the chains $C_e$ with $e\in I$ as in \Cref{obsvn:recolor_terminating_chain} takes $O(1)$ rounds in CONGEST.
\end{observation}
\begin{proof}
    Each vertex in a chain already knows its role within the chain, that is, which $F_i$ or $P_i$ it belongs to, and with which colors its incident edges should be recolored when performing the shift.
    More precisely, edges in any fan in a chain can be recolored by the center of the fan, while, for an edge $f=vw$ on the bichromatic part of a path chain with $w$ being the successor of $v$, vertex $v$ colors $f$.
    Every vertex therefore sends at most one message per incident edge.
\end{proof}

Now we have all the ingredients to prove our main theorem.
\begin{proof}[Proof of \Cref{deterministic_polylog}]
    The algorithm consists of several epochs, one for each iteration of the \textbf{while} loop of \Cref{alg:overall}, where in each epoch a $c/(\Delta^6 \ell \log^2 n)$-fraction of all uncolored edges gets colored, for some constant $c$.
    Since the graph contains a total of at most $n \Delta / 2$ edges, repeating the process, after $(2\Delta^6 \ell \log^3 n) / c$ epochs, the number of uncolored edges is at most:
    \[
    \frac{n\Delta}{2}\parens*{1- \frac{c}{\Delta^6 \ell \log^2 n}}^{\frac{2\Delta^6 \ell \log^3 n}{c}}
    \leq \frac{n\Delta}{2} e^{-\frac{c}{\Delta^6 \ell \log^2 n} \cdot \frac{2\Delta^6 \ell \log^3 n}{c} }
    \leq \frac{n\Delta}{2} e^{-2\log n} < 1\ ,
    \]
    that is, the entire graph is colored. It remains to bound the cost of a single epoch.

    By \Cref{cor:construction-runtime}, constructing the chains takes $O(\Delta^7 \len \log n)$ rounds, and by \Cref{lemma:degree_of_terminating_chain_graph} the chains it computes overlap with at most $O(\Delta^6 \len \log^2 n)$ other chains each.
    By \Cref{lem:large-is-chain-graph}, selecting the independent set $I$ takes
    $O(\Delta^{12} \ell^2 \log^2 n (\log^3 \log n + \log^3 \Delta)) = O(\Delta^{13} \ell^2 \log^2 n \log^3 \log n)$
    rounds, where we absorbed the $\log^3\Delta$ term into an additional factor of $O(\Delta)$.
    By \Cref{obsvn:shift-runtime}, shifting the chains $C_e$ with $e\in I$ takes $O(1)$ rounds.
    The selection of $I$ dominates, so an epoch takes $O(\Delta^{13} \ell^2 \log^2n \log^3 \log n)$ rounds.

    Performing $\Theta(\Delta^6 \len \log^3 n)$ epochs to ensure that no uncolored edge remains therefore takes no more than $O(\Delta^{19} \ell^3 \log^5 n \log^3 \log n)$ rounds. As $\ell \in O(\Delta^{22})$ we get a total runtime of $O(\Delta^{85} \log^5 n \log^3 \log n)$, i.e., $\tilde{O}(\poly(\Delta) \log^5 n)$ rounds.
\end{proof}

\section{Constructing Terminating Multi-Step Vizing Chains}\label{sec:implementation-phases}

This section describes algorithm $\fA$ in full. Following the overview of \Cref{sec:construct-chains}, $\fA$ constructs a collection of multi-step Vizing chains on each uncolored edge, selects one terminating chain per edge, and makes the vertices of that chain aware of it. \Cref{subsec:formal-description} gives the formal description, deferring the implementation of the individual phases to the remaining parts: the \emph{construction phases} (\Cref{subsec:implementation-phase-1,subsec:implementation-phase-j}), which build the candidate chains, and the \emph{selection} and \emph{learning} of a single terminating chain $C_e$ per edge (\Cref{subsec:selecting-learning}). All correctness and runtime analysis is deferred to \Cref{sec:analysis-construction-phases,sec:existence-of-MSVC,sec:analysis-selection}.

\subsection{Formal Description}\label{subsec:formal-description}
 We now describe algorithm $\mathcal{A}$ more formally.
 However, to keep things simple, we will defer the precise implementation of a number of steps of $\fA$ to~\Cref{subsec:implementation-phase-1} and \Cref{subsec:implementation-phase-j}.

 Recall the proper partial $(\Delta+1)$-edge coloring $\chi$, and let $U$ be the set of uncolored edges under $\chi$. Let $\len=1584(\Delta+1)^{22}$ and $\numphases=\ceil{\log_2 n}$.
 For each $x\in V(G)$, fix an arbitrary injective function $\mu_x:N(x)\to [\Delta]$. Note that $\mu_x$ is simply an ordering of the neighbors of $x$. Define $W(x):=\{(y,\mu_x(y))\mid y\in N(x)\}$.

\subparagraph*{Preprocessing steps.}
Algorithm $\fA$ start with a number of preprocessing steps.
First, every vertex $x$ collects the set $A(z)$ of available colors at $z$ for every neighbor $z\in N(x)$. This can be done in $O(\Delta)$ rounds in CONGEST since $A(z)\subseteq[\Delta+1]$ can be described by its characteristic vector of $\Delta+1$ bits.

The objective of the next preprocessing step is that each vertex $v$ that lies on a maximal bichromatic path $P$ learns the set $W(x)$ for each endpoint $x$ of $P$ that is at distance at most $\ell + 1$ from $v$ along $P$. 
To this end, we iterate through all pairs of colors $\A, \B \in [\Delta + 1]$ with $\A \neq \B$, where in each iteration we proceed as follows.

Every vertex $x$ that is an endpoint of some maximal $\ABtuple$-bichromatic path $P$ sends out the message $\tau_{\operatorname{initial}}= (x,W(x))$ along $P$ up to distance $\ell + 1$ along $P$ (unless the other endpoint of $P$ is reached before $\len+1$ steps, in which case $\tau$ stops at that endpoint).
When $\tau$ passes through a vertex $v\in V(P)$, $v$ learns $W(x)$ and can calculate $d_{P}(x,v)$ from the round in which $\tau$ reached $v$.
As $W(x)$ can be described by $O(\Delta\log n)$ bits,
this process can be implemented in the LOCAL model in $O(\len)$ rounds with messages of size $O(\Delta\log n)$ (since messages only travel along $\ABtuple$-bichromatic paths and hence at most 2 messages of the type $\tau$ outlined above are sent over any given edge at once).
Therefore, the same process can be implemented in CONGEST in $O(\Delta\len)$ rounds.
It follows that the total runtime of this preprocessing step (over all iterations) is $O(\Delta^3\len)$ rounds (as there are $O(\Delta^2)$ color pairs).\footnote{In the interest of simplicity, we do not focus on optimizing the dependency on $\Delta$ in our algorithm. For instance, one could reduce the runtime of the discussed preprocessing step by a $\Delta$-factor to $O(\Delta^2 \ell)$ rounds by arguing that when a message is sent over some edge $e'$, this message must come from an iteration where the considered color pair $(\alpha,\beta)$ satisfies that $\alpha$ or $\beta$ is the color of $e'$, observing that there are only $O(\Delta)$ such color pairs, and pipelining the iterations suitably.}

Next, we initialize a number of variables.
For all $e\in U$, $x\in V(G)$, $\A,\B \in [\Delta+1]$ satisfying $\A\neq\B$, and $0\leq i \leq \numphases$, initialize the sets $R_i(e,\len)$, $S_i(e,\len)$, $Q_i(e,\len)$, $B_i(e,\len)$, $D_i(e,\len)$, and $M_i(x)$ to $\emptyset$.
We now update the values of $R_0(e,\len)$, $S_0(e,\len)$ and $M_0(x)$ for all $e\in U$ and $x\in V(G)$.
For any $e=xy \in U$ with $\id(x)>\id(y)$, set $R_0(e,\len):=\{x\}$ and $S_0(e,\len):=\{x\}$. 
For any $x\in V(G)$, set $M_0(x):=\{e\in U\mid x\in S_0(e,\len)\}$.
Intuitively, $M_0(x)$ collects all edges $e$ for which the fan of the $1$-step Vizing chain on $e$ that we want to construct (in the first construction phase) has center $x$. 

Next, we compute a proper vertex coloring $\psi$ of $G^2$ with $(\Delta^2+1)$ colors. This coloring will be used at different points of the algorithm. By \cite[Corollary 3.6]{BG_disc24}, such a coloring can be obtained in $O(\Delta^{3/2}\log\Delta+\log^*n)$, and in particular in $O(\Delta^2+\log^*n)$, rounds in the CONGEST model.

The goal of the final preprocessing step is to let each vertex $v$ learn if $v$ is in $N^2[R_0(e,\len)]$ for any uncolored edge $e$. 
We iterate through the colors of $\psi$ in arbitrary order.
When processing a color $c$, each vertex $x$ of color $c$ sends the message $(x,M_0(x))$ to all vertices in its $2$-hop neighborhood.
This process can be implemented in the LOCAL model in $O(1)$ rounds with messages of size $O(\Delta\log n)$ bits since $|M_0(x)|\leq \Delta$ and each edge has a unique id of $O(\log n)$ bits. Hence, the process can be implemented in $O(\Delta)$ rounds in CONGEST. 
It follows that the total runtime of this step (over all iterations) is $O(\Delta^3)$ rounds in the CONGEST model (as there are $O(\Delta^2)$ vertex colors).

\subparagraph*{Construction phases.}
After performing the above preprocessing steps, Algorithm $\fA$ constructs Vizing chains on the uncolored edges.
More precisely, $\fA$ proceeds by first executing construction phase $1$ as described in \Cref{subsec:implementation-phase-1} and then executing construction phases $j$ as described in \Cref{subsec:implementation-phase-j}, in increasing order of $j$ for $2\leq j\leq \numphases$.

\subparagraph*{Postprocessing steps.}
After the construction phases are concluded, each uncolored edge $e$ selects one terminating self-avoiding Vizing chain among the many constructed terminating self-avoiding Vizing chains on $e$, and the vertices of the selected chain learn that they lie on it.
To this end, $\fA$ executes $\numphases-1$ \emph{marking} phases, with the index running from $\numphases-2$ down to $0$, then selects a unique terminating multi-step Vizing chain $C_e$ for each $e\in U$, and finally executes $\numphases$ \emph{relay} phases.
We emphasize that $C_e$ is stored in a distributed fashion and not at a single vertex.
All three steps are described in \Cref{subsec:selecting-learning}.
This concludes the description of Algorithm $\mathcal{A}$.

\subsection{Construction Phase 1} \label{subsec:implementation-phase-1}

We begin with construction phase $1$, in which the first Vizing chain on each uncolored edge is built.
For any $x\in V(G)$, fix an injective function $\lambda_{x,0}$ from $M_0(x)$ to $[\Delta]$. 
This is possible since $|M_0(x)|\leq \Delta$ for any $x\in V(G)$.

The rest of construction phase $1$ proceeds in $O(\Delta^3)$ subphases. Each subphase 
is uniquely determined by a pair $(c,r)$ where $1\leq c \leq \Delta^2+1$ and 
$1 \leq r\leq \Delta$. In subphase $(c,r)$, every vertex $x$ with $\psi(x)=c$ that has an edge
$e\in M_0(x)$ with $\lambda_{x,0}(e)=r$ constructs a $1$-step self-avoiding
Vizing chain on that edge. Since $\lambda_{x,0}$ is an injective function for any vertex 
$x\in V(G)$, iterating through all possible pairs $(c,r)$ in an arbitrary order and 
running subphase $(c,r)$ ensures that a $1$-step Vizing chain is constructed on any 
$e\in M_0(y)$ for any $y\in V(G)$.

Iterate through all possible pairs of $(c,r)$ and run subphase $(c,r)$ which is 
described in \Cref{subsubsec:implementation-sub-phase-in-phase-1}.

During the subphases, vertices accumulate tuples into the sets $M_1(v)$ for 
$v\in V(G)$, where each tuple has the form $(e,x,\A,\B,vw)$ with $e\in U$, 
$x\in V(G)$, $\A,\B\in[\Delta+1]$, and $w\in N(v)$. Once all subphases of 
construction phase $1$ are complete, we deduplicate $M_1(v)$ as follows. For each 
uncolored edge $e\in U$, define the \emph{$(e,1)$-fiber} of $v$ by
$$
M_1(v;e):=\{(e',x,\A,\B,vw)\in M_1(v)\mid e'=e\}.
$$
We say that $M_1(v)$ has \emph{duplicates of $e$} if $|M_1(v;e)|\ge 2$. Vertex $v$
replaces each nonempty fiber $M_1(v;e)$ by an arbitrary singleton subset of it. We
slightly abuse the notation and continue to denote the resulting set by $M_1(v)$.

\Cref{alg:phase-1} summarizes construction phase $1$.

\begin{algorithm}[!htbp]
\caption{Construction phase $1$.}
\label{alg:phase-1}
\begin{algorithmic}[1]
\State every $x\in V(G)$ fixes an injective $\lambda_{x,0}\colon M_0(x)\to[\Delta]$ \Comment{possible since $|M_0(x)|\leq\Delta$}
\ForAll{$(c,r)\in[\Delta^2+1]\times[\Delta]$, in an arbitrary order}
    \State run subphase $(c,r)$ \Comment{\Cref{alg:subphase-phase-1}}
\EndFor
\ForAll{$v\in V(G)$ and $e\in U$ with $M_1(v;e)\neq\emptyset$}
    \State $v$ replaces $M_1(v;e)$ by an arbitrary singleton subset of it \Comment{deduplication}
\EndFor
\end{algorithmic}
\end{algorithm}

\subsubsection{Subphase \texorpdfstring{$(c,r)$}{(c,r)} of Construction Phase 1}
\label{subsubsec:implementation-sub-phase-in-phase-1}

Let $(c,r)\in[\Delta^2+1]\times[\Delta]$. We now describe subphase $(c,r)$ of construction phase $1$.
We emphasize that subphase $(c,r)$ is a single global part of the algorithm. It is not defined separately for each vertex. Instead, it consists of the local actions performed in parallel by all vertices $x$ with $\psi(x)=c$ for which there exists an edge $e\in M_0(x)$ satisfying $\lambda_{x,0}(e)=r$.

The subphase has two stages, executed one after the other: the chain construction stage and the set construction stage. The chain construction stage is itself divided into four substages, called the preparation, search, notification, and confirmation substages, which are also executed in this order. We refer to the preparation and search substages together as the \emph{fan chain construction}, and to the notification and confirmation substages together as the \emph{path chain construction}.

Now fix a vertex $x$ with $\psi(x)=c$, and suppose there exists an edge $e\in M_0(x)$ such that $\lambda_{x,0}(e)=r$. During subphase $(c,r)$, vertex $x$ attempts to construct a $1$-step self-avoiding Vizing chain $F+P$ of step length $\len$ on $e$ under $\chi$. We call this chain the candidate chain for $e$ constructed by $x$ in construction phase $1$.

At the end of the chain construction stage, vertex $x$ either \emph{finishes} or \emph{abandons} the construction of $F+P$. It returns a variable $\suc$ with a value in $\{1,0\}$ in case it finishes and a value $-1$ in case it abandons the construction. Intuitively, $\suc=1$ records that the construction reached a point at which the chain can be stopped, $\suc=0$ that it was instead cut off at the step length and must therefore be continued in a later construction phase, and $\suc=-1$ that it ran into a part of the graph already used by an earlier construction phase and had to be given up. Later we will show that whenever $x$ finishes the construction, the resulting chain $F+P$ is indeed a $1$-step self-avoiding Vizing chain of step length $\len$ on $e$ under $\chi$.

If $x$ finishes the construction, then we call $F+P$ a \emph{successful candidate chain} when $\suc=1$, and a \emph{hopeful candidate chain} when $\suc=0$. If $x$ abandons the construction, then we call $F+P$ an \emph{unsuccessful candidate chain}. Note that in construction phase $1$ the confirmation substage never returns $\suc=-1$, so vertex $x$ always finishes the construction and no unsuccessful candidate chains arise in this phase; we nevertheless introduce all three terms here since they will be reused in construction phases $j\geq 2$, where $x$ may abandon the construction.

All vertices that participate in subphase $(c,r)$ execute each stage and substage in parallel. Different vertices may complete their local work within a given substage at different times. However, the algorithm proceeds to the next substage only after every participating vertex has completed the current one.

For a concise pseudocode of the subphase, see \Cref{alg:subphase-phase-1}. Inside its two subroutines, a line labelled \textit{(Preparation)}, \textit{(Search)}, \textit{(Notification)} or \textit{(Confirmation)} is executed in the substage named by that label, and a line carrying no label is executed in the same substage as the line preceding it. Throughout the chain construction stage the description is from the perspective of $x$; in the set construction stage, the additions to $R_1(e,\len)$, $S_1(e,\len)$ and $M_1(\cdot)$ are performed by the vertices of $P$ themselves.

We now describe the actions performed during these stages in detail, for an arbitrary vertex $x$ with $\psi(x)=c$ and an edge $e\in M_0(x)$ satisfying $\lambda_{x,0}(e)=r$. Note that $e\in M_0(x)$ implies $x\in S_0(e,\len)$.

\begin{algorithm}[!htbp]
\caption{Subphase $(c,r)$ of construction phase $1$, at a vertex $x$ with $\psi(x)=c$ holding $e\in M_0(x)$ with $\lambda_{x,0}(e)=r$.}
\label{alg:subphase-phase-1}
\small
\begin{algorithmic}[1]
\Statex \textsc{Chain construction stage}
\State \Call{FanChainConstruction}{} \Comment{yields the fan chain $F$}
\State \Call{PathChainConstruction}{} \Comment{yields the path chain $P$ and $\suc$}
\State $x$ finishes the construction of $F+P$
\State $F+P$ is \emph{successful} if $\suc=1$, and \emph{hopeful} if $\suc=0$
\Statex
\Statex \textsc{Set construction stage}
\State $x$ joins $D_0(e,\len)$ if $F+P$ is successful, and $Q_0(e,\len)$ if $F+P$ is hopeful
\State every $v\in V(P)\setminus\{x\}$ joins $R_1(e,\len)$ and notifies its $2$-hop neighborhood
\If{$F+P$ is hopeful} \Comment{only in the \nameref{Phase_1_construction_case_2}, so $\ABSet$ is defined}
    \ForAll{$v\in V(\bichrom(P))\setminus N^2[x]$, setting $vw\gets\eend(P|d_P(x,v))$}
        \If{$\frac{\len}{10\Delta^4}\cdot b\leq d_P(x,v)<\frac{\len}{10\Delta^4}\cdot(b+1)$, where $b=\Delta^3\cdot\mu_{\vend(F)}(x)+r$}
            \State $v$ joins $S_1(e,\len)$
            \State $v$ adds $(e,x,\A,\B,vw)$ to $M_1(v)$ if $\chi(vw)=\A$, and $(e,x,\B,\A,vw)$ otherwise
        \EndIf
    \EndFor
\EndIf
\Procedure{FanChainConstruction}{}
    \State \textit{(Preparation)} $F'=(xz_1,\dots,xz_p)\gets$ fan chain on $e$ under $\chi$ \Comment{\Cref{lemma:fan-addition-phase-1}}
    \If{$F'$ has the \commoncolorprop{}} \Comment{\nameref{Phase_1_construction_case_1}}
        \State $F\gets F'$
        \State pick $\eta\in A(x,\chi)\cap A(\vend(F),\chi)$
        \State \textit{(Search)} $x$ takes no action
    \Else \Comment{$F'$ has the \repeatedcolorprop{}; \nameref{Phase_1_construction_case_2}}
        \State $\B\gets$ the representative available color of $z_q$ and $z_p$, where $q<p$
        \State $\A\gets$ any color in $A(x,\chi)$
        \State \textit{(Search)} $K\gets(xz_p)+K'$ \Comment{$K'$: maximal $\ABtuple$-bichromatic path from $z_p$ under $\chi$}
        \State $x$ propagates $\tau_{\operatorname{search}}=(e,x,z_q,\A,\B)$ along $K$
        \State the propagation stops at $x$, at $z_q$, at an endpoint of $K$, or at hop $\len+1$
        \State if it stops at $x$ or at $z_q$, that vertex sends $\tau_{\operatorname{return}}=(e,x)$ to $x$ \Comment{one hop, as $z_q\in N(x)$}
        \State $F\gets F'|q$ if $x$ receives $\tau_{\operatorname{return}}$, and $F\gets F'|p$ otherwise \Comment{it does iff some $v\in\{x,z_q\}$ lies on $K'$ with $d_{K'}(z_p,v)\leq\len$}
    \EndIf
\EndProcedure
\Procedure{PathChainConstruction}{}
    \If{$F'$ has the \commoncolorprop{}} \Comment{\nameref{Phase_1_construction_case_1}}
        \State \textit{(Notification)} $x$ sends $\tau_{\operatorname{notify}}=(e,x,\vend(F),\neg,\eta)$ to $\vend(F)$
        \State $P\gets(x\vend(F))$ \Comment{a path chain of length $1$ under $\Shift(\chi,F)$}
        \State \textit{(Confirmation)} $\vend(P)=\vend(F)$ replies $\tau_{\operatorname{reply}}=(e,x,\vend(F),\neg,\eta,\suc=1)$
    \Else \Comment{\nameref{Phase_1_construction_case_2}}
        \State \textit{(Notification)} $Q\gets(x\vend(F))+Q'$ \Comment{$Q'$: maximal $\ABtuple$-bichromatic path from $\vend(F)$ under $\chi$}
        \State $x$ propagates $\tau_{\operatorname{notify}}=(e,x,\vend(F),\A,\B)$ along $Q$
        \State the propagation stops at an endpoint of $Q$ or at hop $\len+1$
        \State $P\gets$ the traversed prefix of $Q$ \Comment{a path, by \Cref{obsvn:x_vend(F)_far_along_Q'_phase_1}}
        \State \textit{(Confirmation)} $\vend(P)$ replies $\tau_{\operatorname{reply}}$ along $\operatorname{rev}(P)$, where
        \If{$\vend(P)$ is an endpoint of $Q$}
            \State $\tau_{\operatorname{reply}}=(e,x,\vend(F),\A,\B,\suc=1)$
        \Else
            \State $\tau_{\operatorname{reply}}=(e,x,\vend(F),\A,\B,\suc=0)$
        \EndIf
    \EndIf
\EndProcedure
\end{algorithmic}
\end{algorithm}

\paragraph*{Chain construction stage for $x$.}
In the search, notification, and confirmation substages, vertex $x$ initiates a message that propagates along a fan or a bichromatic path; each vertex on that path forwards it to its successor using the color of the edge on which it arrived together with the contents of the message. Every such propagation is therefore well defined, and we do not remark on this again.

In the preparation substage, vertex $x$ constructs a fan chain $F'=(xz_1,xz_2,\dots,xz_p)$ using \Cref{lemma:fan-addition-phase-1}. It then executes the \nameref{Phase_1_construction_case_1} or the \nameref{Phase_1_construction_case_2} according as $F'$ has the \commoncolorprop{} or the \repeatedcolorprop{}.
    
\subparagraph{Common-Color Case.} \label[case]{Phase_1_construction_case_1} Suppose $F'$ is the fan with the \commoncolorprop{}. Then $\eta\in A(x,\chi)\cap A(\vend(F'),\chi)$ for some $\eta\in[\Delta+1]$, and hence $\eta\in A(x,\Shift(\chi,F'))\cap A(\vend(F'),\Shift(\chi,F'))$ by \Cref{obsvn:fan_available_color}. In the preparation substage, vertex $x$ sets $F:=F'$.
       
Vertex $x$ performs no action in the search substage. In the notification substage it sends $\tau_{\operatorname{notify}}=(e,x,\vend(F),\neg,\eta)$ to $\vend(F)$ and then sets $P:=(x\vend(F))$, a path chain of length $1$ under $\Shift(\chi,F)$. In the confirmation substage $\vend(P)$($=\vend(F)$) replies with $$\tau_{\operatorname{reply}}=(e,x,\vend(F),\neg,\eta,\suc=1).$$
Vertex $x$ then finishes the construction of $F+P$.        
        
\subparagraph{Repeated-Color Case.} \label[case]{Phase_1_construction_case_2}  Suppose $F'$ is the fan with the \repeatedcolorprop{}. Then there exist $\A,\B\in [\Delta+1]$ with $\A\neq \B$ and an index $q\in [p-1]$ such that
   \begin{itemize}
        \item $A(x,\chi)\cap A(z_i,\chi)=\emptyset$ for $i\in [p]$,  and
        \item $\A\in A(x,\chi)$, $\B\in A(z_p,\chi)\cap A(z_q,\chi)$.
    \end{itemize}

In the search substage, the conditions above make both $z_p$ and $z_q$ endpoints of $\ABtuple$-bichromatic paths under $\chi$. Vertex $x$ propagates a message $\tau_{\operatorname{search}}=(e,x,z_q,\A,\B)$ along $K=(xz_p)+K'$ where $K'$ is the maximal $\ABtuple$-bichromatic path under $\chi$ with $z_p$ as an endpoint. $\tau_{\operatorname{search}}$ stops propagating at a vertex $v$ according to the following subcases.
   \begin{enumerate}
        \item If $v\in\{x,z_q\}$, then $\tau_{\operatorname{search}}$ stops at $v$, and $v$ sends the message $\tau_{\operatorname{return}}=(e,x)$ to $x$ along the edge $xv$. This is a single hop, since $z_q\in N(x)$; in case $v=x$, vertex $v$ sends the message to itself.

        \item If subcase 1 does not apply and $v$ is an endpoint of $K$, then $\tau_{\operatorname{search}}$ stops at $v$.

        \item If subcases 1 and 2 do not apply, then $\tau_{\operatorname{search}}$ stops at $v$ if $d_{K}(x,v)=\len+1$.
    \end{enumerate}
If $x$ receives $\tau_{\operatorname{return}}$ (which happens when some $v\in\{x,z_q\}$ is on $K'$ and $d_{K'}(z_p,v)\leq \len$), $x$ sets $F$ to be the fan chain $F'|q$. Else $x$ sets $F=F'|p$.  Let $Q'$ be the maximal $\ABtuple$-bichromatic path under $\chi$ starting at $\vend(F)$ and $Q$ be the trail $(x\vend(F))+Q'$ under $\chi$.

We make the following observation about $x$ and $\vend(F)$ along $Q'$.
\begin{observation} \label{obsvn:x_vend(F)_far_along_Q'_phase_1}
    Either $x\notin V(Q')$, or $x$ and $\vend(F)$ are at distance greater than $\len$ along $Q'$.
\end{observation}
\begin{proof}
    We have $\vend(F)\in\{z_p,z_q\}$. Note that $z_p$, $z_q$ and $x$ are distinct and are endpoints of maximal $\ABtuple$-bichromatic paths under $\chi$. Hence all three of them cannot lie on the same maximal $\ABtuple$-bichromatic path under $\chi$.
    
    If $\vend(F)=z_q$, then $x$ received $\tau_{\operatorname{return}}$, so $\tau_{\operatorname{search}}$ stopped at $x$ or at $z_q$. Suppose it stopped at $x$. Then $x$ and $z_p$ both lie on $K'$, so $z_q\notin V(K')$ and hence $Q'\neq K'$; as two maximal $\ABtuple$-bichromatic paths under $\chi$ are either equal or vertex-disjoint, this gives $x\notin V(Q')$. Suppose instead it stopped at $z_q$. Then $z_q\in V(K')$ and hence $Q'=K'$; as $z_p$ and $z_q$ both lie on $K'$, this again gives $x\notin V(Q')$.

    If $\vend(F)=z_p$, then $Q'=K'$. By construction, $\vend(F)=z_p$ is chosen only when $x$ does not receive $\tau_{\operatorname{return}}$, which in particular means that either $x\notin V(K')$ or $d_{K'}(z_p,x)>\len$. In either case, either $x\notin V(Q')$ or $x$ and $\vend(F)$ are at distance greater than $\len$ along $Q'$.
\end{proof}
 
In the notification substage, vertex $x$ sends a message $\tau_{\operatorname{notify}}=(e,x,\vend(F),\A,\B)$ that propagates until it traverses $\len+1$ edges along $Q$ or until it reaches an endpoint of $Q$, whichever is earlier. Let
$\tau_{\operatorname{notify}}$ stop its propagation at vertex $v$. Since $\tau_{\operatorname{notify}}$ traverses at most $\len$ edges of $Q'$, \Cref{obsvn:x_vend(F)_far_along_Q'_phase_1} implies that the traversed portion of $Q'$ does not contain $x$. Hence $Q|d_Q(x,v)$ is a path, and we let $P$ be the path $Q|d_Q(x,v)=(x\vend(F))+Q'|d_{Q'}(\vend(F),v)$.
        
In the confirmation substage, $\vend(P)$ sends a message $\tau_{\operatorname{reply}}$ along the path $\operatorname{rev}(P)$ back to $x$. The content of $\tau_{\operatorname{reply}}$ is defined according to the following exhaustive subcases. 
    \begin{itemize}
        \item If $\vend(P)$ is an endpoint of $Q$, $\tau_{\operatorname{reply}}=(e,x,\vend(F),\A,\B,\suc=1)$. 
        \item If $\vend(P)$ is not an endpoint of $Q$,  $\tau_{\operatorname{reply}}=(e,x,\vend(F),\A,\B,\suc=0)$.
    \end{itemize}
Vertex $x$ \emph{finishes} the construction of $F+P$. 
         
The following observation about $\bichrom(P)$ follows from the construction in the \nameref{Phase_1_construction_case_2}.
\begin{observation} \label{obsvn:vend(F)_maximal_point}
If $x$ constructs $F+P$ by the \nameref{Phase_1_construction_case_2},
    \begin{itemize}
        \item $\bichrom(P)$ is a $\vend(F)$-maximal $\ABtuple$-bichromatic path of length at most $\len$ under $\chi$, and
        \item $x$ and $\vend(F)$ are not $\bichrom(P)$-related.
    \end{itemize}
\end{observation} 

This completes the description of the chain construction stage. In both cases, $F$ together with an initial segment of $P$ forms a self-avoiding chain. We state this below and defer the proof to \Cref{subsec:properties-phase-1}.

\begin{restatable}{lemma}{IsSelfAvodingVizingChainPhaseOne}
    \label{lemma:F+P_is_self-avoiding_Vizing_chain}
        Let $F+P$ be a candidate chain for an uncolored edge $e\in U$ constructed by vertex $x$ in construction phase $1$. Suppose $v\in V(P)\setminus \{x\}$. Then $F+P|d_P(x,v)$ is a self-avoiding $1$-step Vizing chain  of step length $\len$ on $e$ under $\chi$.
\end{restatable}

Moreover, a successful candidate chain is a terminating Vizing chain, so that the construction on $e$ may stop at it. We state this below and defer the proof to \Cref{subsec:properties-phase-1} as well.

\begin{restatable}{lemma}{SuccessfulIsTerminatingPhaseOne}
    \label{lemma:successful_is_terminating_phase_1}
        Let $F+P$ be a successful candidate chain for an uncolored edge $e\in U$ constructed by vertex $x$ in construction phase $1$. Then $F+P$ is a terminating self-avoiding $1$-step Vizing chain of step length $\len$ on $e$ under $\chi$.
\end{restatable}

\paragraph*{Set construction stage for $x$.}
From the chain construction stage, $F+P$ is a successful candidate chain if $\tau_{\operatorname{reply}}$ carries $\suc=1$, and a hopeful candidate chain if it carries $\suc=0$.
Each vertex $v\in V(P)$ knows whether $F+P$ is successful or hopeful since each $v\in V(P)$  either sends or receives $\tau_{\operatorname{reply}}$.

Each vertex $v\in V(P)\setminus\{x\}$, as well as vertex $x$, adds themselves to different sets according to the following cases.
    \begin{itemize}
        \item If $F+P$ is a successful candidate chain, $x$ then adds itself to $D_{0}(e,\len)$. If $F+P$ is a hopeful candidate chain, $x$ adds itself to $Q_{0}(e,\len)$.

         \item  If $F+P$ is a successful candidate chain or a hopeful candidate chain, $v$ adds itself to $R_1(e,\len)$. Note that each $v$ can locally determine whether $v\in V(P)\setminus\{x\}$, since $v\in V(P)\setminus\{x\}$ iff $v\neq x$ and $v$ sent or received $\tau_{\operatorname{reply}}$.

        \item  Suppose $F+P$ is a hopeful candidate chain (which implies $F+P$ was constructed by the \nameref{Phase_1_construction_case_2}). Let $v\in V(\bichrom(P))\setminus N^2[x]$.  Let $vw=\eend(P|d_{P}(x,v))$.
        Note that $vw\in E(\bichrom(P))$ since $v\notin N^2[x]$.
        Vertex $v$ now 
        
        \begin{itemize}
            \item  adds itself to $S_1(e,\len)$ if $$ \frac{\len}{10\Delta^4}\cdot(\Delta^3\cdot\mu_{\vend(F)}(x)+r) \leq d_P(v,x) < \frac{\len}{10\Delta^4}\cdot(\Delta^3\cdot\mu_{\vend(F)}(x)+r+1),$$

            \item  and, if $v\in S_1(e,\len)$, adds to $M_1(v)$ the tuple $(e,x,\theta,\theta',vw)$, where $\{\theta,\theta'\}=\ABSet$ is the color pair of $\bichrom(P)$ and $\theta=\chi(vw)$.

        \end{itemize}
        
        We claim that this step is well-defined.
        Vertex $v$ knows the set $W(\vend(F))$ from the preprocessing step of algorithm $\fA$ (since $\bichrom(P)$ is a $\vend(F)$-maximal bichromatic path of length at most $\len$ by \Cref{obsvn:vend(F)_maximal_point}). Hence, it can compute $\mu_{\vend(F)}(x)$.
        Moreover, $\chi(vw)=\Shift(\chi,F)(vw)$ holds if $E(\bichrom(P))\cap E(F)=\emptyset$. This is true since $F+P$ is a self-avoiding $1$-step Vizing chain by \Cref{lemma:F+P_is_self-avoiding_Vizing_chain}.
        
\end{itemize}
A vertex $v$ in $R_1(e,\len)$ notifies its $2$-hop neighborhood that it is part of $R_1(e,\len)$ by sending a message $(v,e)$.

\subsection{Construction Phase \texorpdfstring{$j$}{j} for \texorpdfstring{$2\leq j \leq \numphases$}{2 ≤ j ≤ \numphases}} \label{subsec:implementation-phase-j}

Having described the first construction phase, we turn to the later ones. Construction phase $j$ extends the chains built in the previous construction phases.
Before stating formally what we assume about those chains, we explain what a vertex needs in order to carry out such an extension, and how the three statements below supply it.

At the end of construction phase $j-1$, the algorithm has built, for each uncolored edge $e\in U$, a collection of $(j-1)$-step self-avoiding Vizing chains on $e$.
None of these chains is stored anywhere: such a chain can have $\Theta(j\len)$ edges and is spread across the network, and no vertex can afford to learn it.
What remains of a chain $C$ after construction phase $j-1$ is a single tuple $(e,z,\A,\B,xy)$, held by the vertex $x$ at which $C$ ends; we call $x$ the \emph{tip} of $C$.
In construction phase $j$, it is this tip that must extend $C$, by appending a fan and a path chain to it, and it must do so knowing nothing about $C$ beyond its own tuple.

The fan is constructed using \Cref{lemma:fan-addition-phase-j}, applied to the edge $xy$ under the coloring $\Shift(\chi,C)$ obtained by shifting $C$.
That lemma asks for two things: the tip must know which colors are available at every vertex of $N[x]$ under $\Shift(\chi,C)$, and it must have a color available at $x$ and a color available at $y$ under that coloring.
Statements $\Pi_1(j-1)$ and $\Pi_3(j-1)$ together supply these, and statement $\Pi_2(j-1)$ ensures that a vertex can proceed in this way for every chain of which it is the tip.

Statement $\Pi_1(j-1)$ says that the tuple held by a tip stands for an actual chain, and describes the coloring at that chain's own two ends.
It gives a unique $(j-1)$-step self-avoiding Vizing chain $C$ on $e$, namely the one recovered from the tuple segment by segment: the vertex $z$ constructed the last segment $F_{j-1}+P_{j-1}$ in construction phase $j-1$, the tuple that $z$ itself holds for $e$ names the vertex that constructed the segment before it, and so on down to construction phase $1$. Speaking of ``the chain a tuple stands for'' is therefore meaningful even though no vertex knows that chain.
It further gives that $\A$ is available at $x$ and that $\B$ is available at $y$ under $\Shift(\chi,C)$, which are the two colors \Cref{lemma:fan-addition-phase-j} asks for.

Statement $\Pi_3(j-1)$ says that the tip of a chain on $e$ keeps its distance from the vertices reached in earlier construction phases: it lies at distance more than $2$ from $\bigcup_{k=0}^{j-2}R_k(e,\len)$.
Those earlier phases have therefore left the neighbourhood of the tip untouched, so shifting $C$ alters the available colors at no vertex of $N[x]$ other than $x$ and $y$ themselves, where the alterations are exactly the two colors named in the tuple (made precise in \Cref{lemma:last_node_nbrhd_available_colors}).
The tip can thus recover the available colors under $\Shift(\chi,C)$ at every vertex of $N[x]$ from those it collected during the preprocessing steps, which is what \Cref{lemma:fan-addition-phase-j} asks for besides the two colors above.
As the fan the tip goes on to construct lies in its own neighbourhood, the same clearance keeps that fan away from the parts of $C$ built in the earlier phases.
The statement additionally records that $x$ lies in $S_{j-1}(e,\len)$ precisely when it holds a corresponding tuple.

Statement $\Pi_2(j-1)$ says that a vertex holds at most $2\Delta^2$ tuples.
A vertex may be the tip of chains on many different uncolored edges, and it must extend each of them.
This bound is what lets it do so for all of its tuples: as described below, construction phase $j$ has each vertex work through its tuples one at a time, in separate subphases, and the bound caps how many subphases are needed.

We now state the three statements formally. We will later show in \Cref{subsec:statement_Pi1_Pi2_Pi3_proofs} that these statements can be safely assumed at the start of construction phase $j$.
\begin{statement}[$\Pi_1(j-1)$]
Assume construction phase $j-1$ is over. Consider any vertex $x\in V(G)$ and any $(e,z,\A,\B,xy)\in M_{j-1}(x)$. Then there is a unique $(j-1)$-step self-avoiding Vizing chain $C=F_1+P_1+\dots+F_{j-1}+P_{j-1}$ of step length $\len$ on $e$ under $\chi$ such that
    \begin{itemize}
        \item for $1\leq k \leq j-1$, $F_k=F_k^*$ and $P_k$ is a subpath of $P_k^*$ where $F_k^*+P_k^*$ is the unique candidate chain constructed for $e$ by $\cen(F_k)$ during construction phase $k$,
        \item for $1\leq k \leq j-1$, the set $M_k(\vend(P_k))$ contains a tuple with first entry $e$ and second entry $\cen(F_k)$, and 
        \item $xy=\eend(P_{j-1})$ and $x=\vend(P_{j-1})$. 
    \end{itemize}
        Moreover, this chain $C$ has the additional properties that
        \begin{itemize}
        \item $\cen(F_{j-1})=z$,
         
        \item for $1\leq k \leq j-1$, $V(\bichrom(P_k))\subseteq R_k(e,\len)$, 

        \item $\bichrom(P_{j-1})$ is a $\vend(F_{j-1})$-maximal $\ABtuple$-bichromatic path under $\chi$,

        \item $P_{j-1}$ is either a $\BAtuple$-bichromatic path chain or an $\ABtuple$-bichromatic path chain under $\Shift(\chi,F_1+P_1+\dots+F_{j-2}+P_{j-2}+F_{j-1})$ and length of $P_{j-1}$ is greater than $2$, 

        \item $\A\in A(x,\Shift(\chi,C))$, $\B\in A(y,\Shift(\chi,C))$ and $\Shift(\chi,F_1+P_1+\dots+F_{j-1})(xy)=\A$.

        \end{itemize}     

\end{statement}
\begin{statement}[$\Pi_2(j-1)$]
     Assume construction phase $j-1$ is over. Then, $|M_{j-1}(x)|\leq 2\Delta^2$ for each $x\in V(G)$.
\end{statement}

\begin{statement}[$\Pi_3(j-1)$]
     Assume construction phase $j-1$ is finished. Consider any $x\in V(G)$, $e\in U$.
    \begin{itemize}
        \item If $x\in S_{j-1}(e,\len)$, then $x\notin N^2[\bigcup_{k=0}^{j-2}R_k(e,\len)]$.
        \item  $x\in S_{j-1}(e,\len) \iff (e,z,\A,\B,xy)\in M_{j-1}(x)$ for a unique $z\in S_{j-2}(e,\len)$, $\A,\B\in[\Delta+1]$ and $y\in N(x)$.
        
    \end{itemize}
\end{statement}

With these statements in hand, we now describe construction phase $j$ itself.

For any $x\in V(G)$, fix an injective function $\lambda_{x,j-1}$ from $M_{j-1}(x)$ to $[2\Delta^2]$. This is possible since $|M_{j-1}(x)|\leq 2\Delta^2$ for any $x\in V(G)$ by statement $\Pi_2(j-1)$.
   
The rest of construction phase $j$ proceeds in $2\Delta^2\cdot(\Delta^2+1)$ subphases. Each subphase is uniquely determined by a pair $(c,r)$ where $1\leq c \leq \Delta^2+1$ and $1 \leq r\leq 2\Delta^2$.
In subphase $(c,r)$, every vertex $x$ with $\psi(x)=c$ that has a tuple $t\in M_{j-1}(x)$ with $\lambda_{x,j-1}(t)=r$ extends the chain that $t$ stands for. Since $\lambda_{x,j-1}$ is injective for any vertex $x\in V(G)$, iterating through all possible pairs $(c,r)$ in an arbitrary order and running subphase $(c,r)$ ensures that every tuple in $M_{j-1}(y)$ is served, for any $y\in V(G)$.
   
Iterate through all possible pairs of $(c,r)$ and run subphase $(c,r)$ which is described in \Cref{subsubsec:implementation-sub-phase-in-phase-j}.

During the subphases, vertices accumulate tuples into the sets $M_j(v)$ for $v\in V(G)$. Once all subphases of construction phase $j$ are complete, we deduplicate $M_j(v)$ as follows. For each uncolored edge $e\in U$, define the \emph{$(e,j)$-fiber} of $v$ by
\[
M_j(v;e):=\{(e',x,\A,\B,vw)\in M_j(v)\mid e'=e\}.
\]
We say that $M_j(v)$ has \emph{duplicates of $e$} if $|M_j(v;e)|\ge 2$. Vertex $v$ replaces each nonempty fiber $M_j(v;e)$ by an arbitrary singleton subset of it. We slightly abuse the notation and continue to denote the resulting set by $M_j(v)$.
This concludes construction phase $j$.

\Cref{alg:phase-j} summarizes construction phase $j$.

\begin{algorithm}[!ht]
\caption{Construction phase $j$, for $2\leq j\leq\numphases$.}
\label{alg:phase-j}
\begin{algorithmic}[1]
\State Every $x\in V(G)$ fixes an injective $\lambda_{x,j-1}\colon M_{j-1}(x)\to[2\Delta^2]$ \Comment{possible by statement $\Pi_2(j-1)$}
\ForAll{$(c,r)\in[\Delta^2+1]\times[2\Delta^2]$, in an arbitrary order}
    \State Run subphase $(c,r)$ \Comment{\Cref{alg:subphase-phase-j}}
\EndFor
\ForAll{$v\in V(G)$ and $e\in U$ with $M_j(v;e)\neq\emptyset$}
    \State $v$ replaces $M_j(v;e)$ by an arbitrary singleton subset of it \Comment{deduplication}
\EndFor
\end{algorithmic}
\end{algorithm}

\subsubsection{Subphase \texorpdfstring{$(c,r)$}{(c,r)} of Construction Phase \texorpdfstring{$j$}{j} for \texorpdfstring{$2\leq j\leq \numphases$}{2 ≤ j ≤ \numphases}} \label{subsubsec:implementation-sub-phase-in-phase-j}

We now describe subphase $(c,r)$ in construction phase $j$. Subphase $(c,r)$ is a single global part of the algorithm. It is not defined separately for each vertex. Rather, it consists of the local actions performed in parallel by all vertices $x\in V(G)$ with $\psi(x)=c$ for which there exists a tuple $t\in M_{j-1}(x)$ satisfying $\lambda_{x,j-1}(t)=r$. In the same way, each stage and substage below should be understood as the union of the corresponding local actions of all such vertices.

The subphase has two stages, executed one after the other: the \emph{chain construction stage} and the \emph{set construction stage}. The chain construction stage is further divided into four substages, executed in the order: the \emph{preparation} substage, the \emph{search} substage, the \emph{notification} substage, and the \emph{confirmation} substage. As in construction phase $1$, we refer to the preparation and search substages together as the \emph{fan chain construction}, and to the notification and confirmation substages together as the \emph{path chain construction}.

Now consider a vertex $x\in V(G)$ with $\psi(x)=c$, and let $(e,x',\A,\B,xy)\in M_{j-1}(x)$ be such that
$$
\lambda_{x,j-1}((e,x',\A,\B,xy))=r.
$$

By statements $\Pi_1(j-1)$, $\Pi_2(j-1)$, and $\Pi_3(j-1)$, we have the following facts.
\begin{fact}\label[fact]{fact_1}
   Let $(e,x',\A,\B,xy)\in M_{j-1}(x)$ for some $x\in V(G)$. Then there is a unique $(j-1)$-step self-avoiding Vizing chain
    $$
    C=F_1+P_1+\dots+F_{j-1}+P_{j-1}
    $$
    of step length $\len$ on $e$ under $\chi$ such that
    \begin{itemize}
        \item for each $1\leq k\leq j-1$, we have $F_k=F_k^*$ and $P_k$ is a subpath of $P_k^*$, where $F_k^*+P_k^*$ is the unique candidate chain constructed for $e$ by $\cen(F_k)$ during construction phase $k$,
        \item for each $1\leq k\leq j-1$, the set $M_k(\vend(P_k))$ contains a tuple with first entry $e$ and second entry $\cen(F_k)$, and
        \item $xy=\eend(P_{j-1})$ and $x=\vend(P_{j-1})$.
    \end{itemize}
    Moreover, this chain $C$ has the following additional properties:
    \begin{itemize}
        \item $\cen(F_{j-1})=x'$,
        \item for each $1\leq k\leq j-1$, we have $V(\bichrom(P_k))\subseteq R_k(e,\len)$,
        \item $\bichrom(P_{j-1})$ is a $\vend(F_{j-1})$-maximal $\ABtuple$-bichromatic path under $\chi$,
        \item $P_{j-1}$ is either a $\BAtuple$-bichromatic path chain or an $\ABtuple$-bichromatic path chain under $\Shift(\chi,F_1+P_1+\dots+F_{j-2}+P_{j-2}+F_{j-1})$, and the length of $P_{j-1}$ is greater than $2$,
        \item $\A\in A(x,\Shift(\chi,C))$, $\B\in A(y,\Shift(\chi,C))$, and
        $$
        \Shift(\chi,F_1+P_1+\dots+F_{j-1})(xy)=\A.
        $$
    \end{itemize}  
\end{fact}

\begin{fact}
    \label[fact]{fact_2}
    Let $(e,x',\A,\B,xy)\in M_{j-1}(x)$ for some $x\in V(G)$. Then
    \begin{itemize}
        \item $|M_{j-1}(x)|\leq 2\Delta^2$, and
        \item $x\notin N^2\left[\bigcup_{i=0}^{j-2}R_i(e,\len)\right]$.
    \end{itemize}
\end{fact}

The goal of vertex $x$ during the chain construction stage is to construct a $1$-step self-avoiding Vizing chain $F+P$ on $xy$ under $\Shift(\chi,C)$ such that $C+F+P$ is a $j$-step self-avoiding Vizing chain on $e$. We refer to $F+P$ as the \emph{candidate chain constructed by $x$ for $e$ in construction phase $j$}.

At the end of the chain construction stage, vertex $x$ either \emph{finishes} or \emph{abandons} the construction of $F+P$. Later we will show that whenever $x$ finishes the construction, the chain $F+P$ is indeed a $1$-step self-avoiding Vizing chain of step length $\len$ on $xy$ under $\Shift(\chi,C)$.

If $x$ finishes the construction, then we call $F+P$ a \emph{successful candidate chain} when the value $\suc$ returned by the chain construction stage is $1$, and a \emph{hopeful candidate chain} when it is $0$. If $x$ abandons the construction, then we call $F+P$ an \emph{unsuccessful candidate chain}.

Different vertices $x,x'\in V(G)$ with $\psi(x)=\psi(x')=c$ may complete the same substage at different times. However, the algorithm proceeds to the next substage only after every such vertex has completed the current one. Thus, all vertices with color class $c$ begin each stage and substage in parallel, even though they may finish their local work at different times.

For a concise pseudocode of the subphase, see \Cref{alg:subphase-phase-j}. Inside its two subroutines, a line labelled \textit{(Preparation)}, \textit{(Search)}, \textit{(Notification)} or \textit{(Confirmation)} is executed in the substage named by that label, and a line carrying no label is executed in the same substage as the line preceding it. Throughout the chain construction stage the description is from the perspective of $x$; in the set construction stage, the additions to $R_j(e,\len)$, $S_j(e,\len)$ and $M_j(\cdot)$ are performed by the vertices of $P$ themselves.

We now describe the two stages in detail, from the perspective of an arbitrary vertex $x\in V(G)$ with $\psi(x)=c$ and an element $(e,x',\A,\B,xy)\in M_{j-1}(x)$ satisfying
$$
\lambda_{x,j-1}((e,x',\A,\B,xy))=r.
$$

\begin{algorithm}[!ht]
\caption{Subphase $(c,r)$ of construction phase $j$ for $2\leq j\leq\numphases$, at a vertex $x$ with $\psi(x)=c$ holding $(e,x',\A,\B,xy)\in M_{j-1}(x)$ with $\lambda_{x,j-1}((e,x',\A,\B,xy))=r$; here $C$ is the $(j-1)$-step chain this tuple stands for (\Cref{fact_1}).}
\label{alg:subphase-phase-j}
\small
\begin{algorithmic}[1]
\Statex \textsc{Chain construction stage}
\State \Call{FanChainConstruction}{} \Comment{constructs the fan chain $F$}
\State \Call{PathChainConstruction}{} \Comment{constructs path chain $P$ and $\suc$; \Cref{alg:subphase-phase-j-path}}
\State $x$ finishes the construction of $F+P$ if $\suc\in\{0,1\}$, and abandons it if $\suc=-1$
\State $F+P$ is \emph{successful}, \emph{hopeful} or \emph{unsuccessful} as to whether $\suc$ equals $1$, $0$ or $-1$
\Statex
\Statex \textsc{Set construction stage}
\State $x$ joins $D_{j-1}(e,\len)$, $Q_{j-1}(e,\len)$ or $B_{j-1}(e,\len)$ according as $F+P$ is successful, hopeful or unsuccessful
\If{$F+P$ is successful or hopeful}
    \State every $v\in V(P)\setminus\{x\}$ joins $R_j(e,\len)$ and notifies its $2$-hop neighborhood
\EndIf
\If{$F+P$ is hopeful}
    \ForAll{$v\in V(\bichrom(P))\setminus N^2[x]$, setting $vw\gets\eend(P|d_P(x,v))$}
        \If{$v\notin N^2[\bigcup_{k=0}^{j-1}R_k(e,\len)]$ and $\frac{\len}{10\Delta^4}\cdot b\leq d_P(x,v)<\frac{\len}{10\Delta^4}\cdot(b+1)$, where $b=\Delta^3\cdot\mu_{\vend(F)}(x)+r$}
            \State $v$ joins $S_j(e,\len)$
            \State $v$ adds $(e,x,\theta,\theta',vw)$ to $M_j(v)$ \Comment{$\{\theta,\theta'\}$: color pair of $\bichrom(P)$ under $\chi$, with $\theta=\chi(vw)$}
        \EndIf
    \EndFor
\EndIf
\Procedure{FanChainConstruction}{}
    \State \textit{(Preparation)} $x$ computes $A(z,\Shift(\chi,C))$ for every $z\in N[x]$ \Comment{\Cref{lemma:last_node_nbrhd_available_colors}}
    \State $F'=(xz_1,\dots,xz_p)\gets$ fan chain on $xy$ under $\Shift(\chi,C)$, with $\A\in A(x,\Shift(\chi,C))$ and $\B\in A(y,\Shift(\chi,C))$ \Comment{\Cref{lemma:fan-addition-phase-j}}
    \If{$F'$ has the \commoncolorprop{}} \Comment{\nameref{j_phase_construction_case_1}}
        \State $F\gets F'$
        \State pick $\eta\in A(x,\Shift(\chi,C))\cap A(z_p,\Shift(\chi,C))$
        \State \textit{(Search)} $x$ takes no action
    \ElsIf{$F'$ has the \repeatedcolorprop{}} \Comment{\nameref{j_phase_construction_case_2}}
        \State $\D\gets$ the representative available color of $z_q$ and $z_p$, where $q<p$ \Comment{$\D\notin\ABSet$}
        \State $\G\gets$ any color in $A(x,\Shift(\chi,C))\setminus\ABSet$
        \State \textit{(Search)} $K\gets(xz_p)+K'$ \Comment{$K'$: maximal $\GDtuple$-bichromatic path from $z_p$ under $\chi$}
        \State $x$ propagates $\tau_{\operatorname{search}}=(e,x,z_q,\G,\D)$ along $K$
        \State the propagation stops at $x$, at $z_q$, at an endpoint of $K$, or at hop $\len+1$
        \State if it stops at $x$ or at $z_q$, that vertex sends $\tau_{\operatorname{return}}=(e,x)$ to $x$ \Comment{one hop, as $z_q\in N(x)$}
        \State $F\gets F'|q$ if $x$ receives $\tau_{\operatorname{return}}$, and $F\gets F'|p$ otherwise \Comment{it does iff some $v\in\{x,z_q\}$ lies on $K'$ with $d_{K'}(z_p,v)\leq\len$}
    \Else \Comment{$F'$ has the \inheritedcolorprop{}; \nameref{j_phase_construction_case_3}}
        \State $F\gets F'$ \Comment{here $\vend(F)=z_p\neq y$ and $\B\in A(z_p,\Shift(\chi,C))$}
        \State \textit{(Search)} $x$ takes no action
    \EndIf
\EndProcedure
\end{algorithmic}
\end{algorithm}

\begin{algorithm}[!ht]
\caption{The path chain construction of the subphase in \Cref{alg:subphase-phase-j}.}
\label{alg:subphase-phase-j-path}
\small
\begin{algorithmic}[1]
\Procedure{PathChainConstruction}{}
    \If{$F'$ has the \commoncolorprop{}} \Comment{\nameref{j_phase_construction_case_1}}
        \State \textit{(Notification)} $x$ sends $\tau_{\operatorname{notify}}=(e,x,z_p,\neg,\eta)$ to $z_p$
        \State $P\gets(xz_p)$ \Comment{a path chain of length $1$ under $\Shift(\chi,C+F)$}
        \State \textit{(Confirmation)} $\vend(P)=z_p$ replies, taking the first case that applies:
        \State \quad $\tau_{\operatorname{reply}}=(e,x,z_p,\neg,\eta,\suc=-1)$ if $z_p\in N^2[\bigcup_{i=0}^{j-2}R_i(e,\len)]$
        \State \quad the same tuple with $\suc=1$ otherwise
    \ElsIf{$F'$ has the \repeatedcolorprop{}} \Comment{\nameref{j_phase_construction_case_2}}
        \State \textit{(Notification)} $Q\gets(x\vend(F))+Q'$ \Comment{$Q'$: maximal $\GDtuple$-bichromatic path from $\vend(F)$ under $\chi$}
        \State $x$ propagates $\tau_{\operatorname{notify}}=(e,x,\vend(F),\G,\D)$ along $Q$
        \State the propagation stops in $N^2[\bigcup_{i=0}^{j-2}R_i(e,\len)]$, at an endpoint of $Q$, or at hop $\len+1$
        \State $P\gets$ the traversed prefix of $Q$ \Comment{a path, by \Cref{obsvn:x_vend(F)_far_along_Q'_phase_j_case_2}}
        \State \textit{(Confirmation)} $\vend(P)$ replies along $\operatorname{rev}(P)$, taking the first case that applies:
        \State \quad $\tau_{\operatorname{reply}}=(e,x,\vend(F),\G,\D,\suc=-1)$ if $\vend(P)\in N^2[\bigcup_{i=0}^{j-2}R_i(e,\len)]$
        \State \quad the same tuple with $\suc=1$ if $\vend(P)$ is an endpoint of $Q$, and with $\suc=0$ otherwise
    \Else \Comment{\nameref{j_phase_construction_case_3}}
        \State \textit{(Notification)} $Q\gets(x\vend(F))+Q'$ \Comment{$Q'$: maximal $\ABtuple$-bichromatic path from $\vend(F)$ under $\chi$}
        \State $x$ propagates $\tau_{\operatorname{notify}}=(e,x,\vend(F),\A,\B)$ along $Q$
        \State the propagation stops in $N^2[\bigcup_{i=0}^{j-2}R_i(e,\len)]\cup\{x\}$, at an endpoint of $Q$, or at hop $\len+1$
        \State $P\gets$ the traversed prefix of $Q$
        \State \textit{(Confirmation)} $\vend(P)$ replies along $\operatorname{rev}(P)$, by the first case that holds:
        \State \quad $\tau_{\operatorname{reply}}=(e,x,\vend(F),\A,\B,\suc=-1)$ if $\vend(P)\in N^2[\bigcup_{i=0}^{j-2}R_i(e,\len)]\cup\{x\}$,
        \State \quad the same tuple with $\suc=1$ if $\vend(P)$ is an endpoint of $Q$, and with $\suc=0$ otherwise
    \EndIf
\EndProcedure
\end{algorithmic}
\end{algorithm}

\paragraph*{Chain construction stage for $x$.}

In the search, notification, and confirmation substages, vertex $x$ initiates a message that propagates along a fan or a bichromatic path; each vertex on that path forwards it to its successor using the color of the edge on which it arrived together with the contents of the message. Every such propagation is therefore well defined, and we do not remark on this again.

For each $z\in N[x]$, vertex $x$ can compute $A(z,\Shift(\chi,C))$ internally, since it already knows $A(z,\chi)$ from the preprocessing steps. This follows from the next lemma, whose proof is deferred to \Cref{subsec:chain-properties-phase-j}.
\begin{restatable}{lemma}{claimSetOfAvailableColorsIsComputable} \label{lemma:last_node_nbrhd_available_colors}
    Let $z\in N[x]$. Then 
    \begin{itemize}
        \item $A(z,\chi)=A(z,\Shift(\chi,C))$ if $z\notin\{x,y\}$
        \item $A(y,\Shift(\chi,C))=A(y,\chi)\cup\{\B\}$, and
        \item  $A(x,\Shift(\chi,C))=A(x,\chi)\cup\{\A\}$.
    \end{itemize}

\end{restatable}

In the preparation substage, vertex $x$ computes $A(z,\Shift(\chi,C))$ internally for each $z\in N[x]$ and constructs a fan chain $F'=(xz_1,\dots,xz_p)$ on $xy$ under $\Shift(\chi,C)$, with $\A\in A(x,\Shift(\chi,C))$ and $\B\in A(y,\Shift(\chi,C))$, by \Cref{lemma:fan-addition-phase-j}. Depending on whether $F'$ has the \commoncolorprop{}, the \repeatedcolorprop{}, or the \inheritedcolorprop{}, vertex $x$ executes the \nameref{j_phase_construction_case_1}, the \nameref{j_phase_construction_case_2} or the \nameref{j_phase_construction_case_3} respectively.

\subparagraph*{Common-Color Case.} \label{j_phase_construction_case_1} Suppose $F'$ is the fan with the \commoncolorprop{}.  Let $\eta\in[\Delta+1]$ such that $\eta\in A(x,\Shift(\chi,C))\cap A(z_p,\Shift(\chi,C))$; hence $\eta\in A(x,\Shift(\chi,C+F'))\cap A(z_p,\Shift(\chi,C+F'))$ by \Cref{obsvn:fan_available_color}. 
In the preparation substage, vertex $x$ sets $F:=F'$.
        
Vertex $x$ performs no action in the search substage. In the notification substage it sends $\tau_{\operatorname{notify}}=(e,x,z_p,\neg,\eta)$ to $z_p$ and then sets $P:=(xz_p)$, a path chain of length $1$ under $\Shift(\chi,C+F)$.

In the confirmation substage, $z_p$ sends to $x$ a message $\tau_{\operatorname{reply}}$ whose content is according to the following cases.
    \begin{itemize}
        \item If $z_p\in N^2[\bigcup_{i=0}^{j-2} R_i(e,\len)]$, $\tau_{\operatorname{reply}}=(e,x,z_p,\neg,\eta,\suc=-1)$.
        \item Otherwise, $\tau_{\operatorname{reply}}=(e,x,z_p,\neg,\eta,\suc=1)$.
    \end{itemize}
If $x$ receives  $\tau_{\operatorname{reply}}=(e,x,z_p,\neg,\eta,\suc=1)$, $x$ finishes the construction of $F+P$. If $x$ receives  $\tau_{\operatorname{reply}}=(e,x,z_p,\neg,\eta,\suc=-1)$, $x$ abandons the construction of $F+P$. 
         
\subparagraph*{Repeated-Color Case.} \label{j_phase_construction_case_2} 
 Suppose $F'$ is the fan with the \repeatedcolorprop{}. Then, there exist $\G,\D\in[\Delta+1]\setminus \ABSet$ with $\G\neq \D$ and an index $q\in[p-1]$ such that
    \begin{itemize}
        \item  $A(x,\Shift(\chi,C))\cap A(z_i,\Shift(\chi,C))=\emptyset$ for $i\in[p]$,
            
        \item $\G\in A(x,\Shift(\chi,C))$ and $\D\in A(z_p,\Shift(\chi,C))\cap A(z_q,\Shift(\chi,C))$.
        \end{itemize}

In the search substage, we will now see that $z_p$ and $z_q$ are endpoints of maximal  $\GDtuple$-bichromatic paths under $\chi$. We have that $\G \notin A(z_p,\Shift(\chi,C))\cup A(z_q,\Shift(\chi,C))$ since $A(x,\Shift(\chi,C))\cap A(z_i,\Shift(\chi,C))=\emptyset$ for $i\in[p]$. 
By \Cref{lemma:last_node_nbrhd_available_colors}, $\G\notin  A(z_p,\chi)\cup A(z_q,\chi)$. On the other hand, we have $\D \in A(z_p,\Shift(\chi,C))\cap A(z_q,\Shift(\chi,C))$. By \Cref{lemma:last_node_nbrhd_available_colors}, $\D \in A(z_p,\chi)\cap A(z_q,\chi)$. Hence, both $z_p$ and $z_q$ are endpoints of maximal  $\GDtuple$-bichromatic paths under $\chi$.
        
Vertex $x$ propagates out a message $\tau_{\operatorname{search}}=(e,x,z_q,\G,\D)$ along $K=(xz_p)+K'$ where $K'$ is the maximal $\GDtuple$-bichromatic path under $\chi$ with $z_p$ as an endpoint.
Propagation stops at a vertex $v$ according to the following subcases.
    \begin{enumerate}
        \item If $v\in\{x,z_q\}$, $\tau_{\operatorname{search}}$ stops at $v$, and $v$ sends the message $\tau_{\operatorname{return}}=(e,x)$ to $x$ along the edge $xv$. This is a single hop, since $z_q\in N(x)$; in case $v=x$, vertex $v$ sends the message to itself.

        \item If subcase 1 is not true and $v$ is an endpoint of $K$, $\tau_{\operatorname{search}}$ stops at $v$.

        \item If subcases 1 and 2 are not true, $\tau_{\operatorname{search}}$ stops at $v$ if $d_{K}(x,v)=\len+1$. Vertex $v$ knows $d_{K}(x,v)$ depending on the round $\tau_{\operatorname{search}}$ reached $v$.
    \end{enumerate}
If $x$ receives $\tau_{\operatorname{return}}$ (which happens when some $v\in\{x,z_q\}$ is on $K'$ and $d_{K'}(z_p,v)\leq \len$), $x$ sets $F$ to the fan chain $F'|q$. Else $x$ sets $F$ to the fan chain $F'|p$.

Let $Q'$ be the maximal $\GDtuple$-bichromatic path under $\chi$ starting at $\vend(F)$. We make the following observation about $x$ and $\vend(F)$ along $Q'$.
\begin{observation} \label{obsvn:x_vend(F)_far_along_Q'_phase_j_case_2}
    Either $x\notin V(Q')$, or $x$ and $\vend(F)$ are at distance greater than $\len$ along $Q'$.
\end{observation}
\begin{proof}
    We have $\vend(F)\in\{z_p,z_q\}$. Note that $z_p$, $z_q$ and $x$ are distinct and are endpoints of maximal $\GDtuple$-bichromatic paths under $\chi$. Hence all three of them cannot lie on the same maximal $\GDtuple$-bichromatic path under $\chi$.
    
    If $\vend(F)=z_q$, then $x$ received $\tau_{\operatorname{return}}$, so $\tau_{\operatorname{search}}$ stopped at $x$ or at $z_q$. Suppose it stopped at $x$. Then $x$ and $z_p$ both lie on $K'$, so $z_q\notin V(K')$ and hence $Q'\neq K'$; as two maximal $\GDtuple$-bichromatic paths under $\chi$ are either equal or vertex-disjoint, this gives $x\notin V(Q')$. Suppose instead it stopped at $z_q$. Then $z_q\in V(K')$ and hence $Q'=K'$; as $z_p$ and $z_q$ both lie on $K'$, this again gives $x\notin V(Q')$.

    If $\vend(F)=z_p$, then $Q'=K'$. By construction, $\vend(F)=z_p$ is chosen only when $x$ does not receive $\tau_{\operatorname{return}}$, which in particular means that either $x\notin V(K')$ or $d_{K'}(z_p,x)>\len$. In either case, either $x\notin V(Q')$ or $x$ and $\vend(F)$ are at distance greater than $\len$ along $Q'$.
\end{proof}
        
In the notification substage, vertex $x$ sends a message $\tau_{\operatorname{notify}}=(e,x,\vend(F),\G,\D)$ that propagates along $Q$ away from $x$ where $Q:=(x\vend(F))+Q'$.
Propagation stops at a vertex $v$ according to the following subcases.  
    \begin{enumerate}
        \item If $v\in N^2[\bigcup_{i=0}^{j-2}R_i(e,\len)]$, $\tau_{\operatorname{notify}}$ stops at $v$. Let $P=Q|d_{Q}(x,v)$.

        \item If subcase 1 is not true and $v$ is an endpoint of $Q$, $\tau_{\operatorname{notify}}$ stops at $v$. Let $P=Q|d_{Q}(x,v)$.

        \item If subcases 1 and 2 are not true, $\tau_{\operatorname{notify}}$ stops at $v$ if $d_{Q}(x,v)=\len+1$. Let $P=Q|d_{Q}(x,v)$.
    \end{enumerate}

In each subcase, $\tau_{\operatorname{notify}}$ traverses at most $\len$ edges of $Q'$, so \Cref{obsvn:x_vend(F)_far_along_Q'_phase_j_case_2} implies that the traversed portion of $Q'$ does not contain $x$. Hence $P$ is a path.
        
In the confirmation substage, vertex $\vend(P)$ sends a message $\tau_{\operatorname{reply}}$ along the path $\operatorname{rev}(P)$ to $x$ according to the following exhaustive cases (with earlier cases taking precedence over later ones). 
    \begin{itemize}
        \item If $\vend(P)\in N^2[\bigcup_{i=0}^{j-2}R_i(e,\len)]$, $\tau_{\operatorname{reply}}=(e,x,\vend(F),\G,\D,\suc=-1)$. 
        \item If $\vend(P)$ is an endpoint of $Q$, $\tau_{\operatorname{reply}}=(e,x,\vend(F),\G,\D,\suc=1)$.
        \item If $\vend(P)$ is not an endpoint of $Q$,  $\tau_{\operatorname{reply}}=(e,x,\vend(F),\G,\D,\suc=0)$.
    \end{itemize}
If $x$ receives either $\tau_{\operatorname{reply}}=(e,x,\vend(F),\G,\D,\suc=0)$ or $\tau_{\operatorname{reply}}=(e,x,\vend(F),\G,\D,\suc=1)$, $x$ finishes the construction of $F+P$. If $x$ receives $\tau_{\operatorname{reply}}=(e,x,\vend(F),\G,\D,\suc=-1)$, $x$ abandons the construction of $F+P$.

Now that $x$ has either finished or abandoned the construction in the \nameref{j_phase_construction_case_2}, we make the following observation about $\bichrom(P)$ which follows from the construction.        
\begin{observation} \label{obsvn:is_bichromatic_path_phase_j_case_2}
If $x$ finished constructing $F+P$ by the \nameref{j_phase_construction_case_2},
    \begin{itemize}
        \item $\bichrom(P)$ is a $\vend(F)$-maximal $\GDtuple$-bichromatic path of length at most $\len$ under $\chi$, and
               
        \item $x$ and $\vend(F)$ are not $\bichrom(P)$-related.
    \end{itemize}
\end{observation}
 
\subparagraph*{Inherited-Color Case.} \label{j_phase_construction_case_3} Suppose $F'$ is the fan with the \inheritedcolorprop{}. Then,
 \begin{itemize}
    \item $A(x,\Shift(\chi,C))\cap A(z_i,\Shift(\chi,C))=\emptyset$ for $i\in[p]$,
            
    \item $\B\in A(z_p,\Shift(\chi,C))$, and

    \item $z_p\neq z_1$ i.e. $\vend(F')\neq y$.
\end{itemize}
Note that $\A\neq \B$ from these conditions. In the preparation substage, vertex $x$ sets $F:=F'$. Vertex $x$ performs no action in the search substage.

In the notification substage, we will now see that $\vend(F)$ is an endpoint of a maximal $\ABtuple$-bichromatic path under $\chi$. To see why, note that $\A\notin A(\vend(F),\Shift(\chi,C))$ since $A(x,\Shift(\chi,C))\cap A(\vend(F),\Shift(\chi,C))=\emptyset$. Hence $\vend(F)$ is an endpoint of a maximal $\ABtuple$-bichromatic path under $\Shift(\chi,C)$. This implies that $\B\in A(\vend(F),\chi)$ by \Cref{lemma:last_node_nbrhd_available_colors}. Hence, $\vend(F)$ is an endpoint of a maximal $\ABtuple$-bichromatic path under $\chi$.

Let  $Q'$ be the maximal $\ABtuple$-bichromatic path under $\chi$ starting at $\vend(F)$. Let $Q=(x\vend(F))+Q'$. Vertex $x$ sends out a message $\tau_{\operatorname{notify}}=(e,x,\vend(F),\A,\B)$ that propagates along $Q$ away from $x$. Propagation stops at a vertex $v$ according to the following subcases. 
  \begin{enumerate}
        \item If $v\in N^2[\bigcup_{i=0}^{j-2}R_i(e,\len)]\cup\{x\}$, $\tau_{\operatorname{notify}}$ stops at $v$. Let $P=Q|d_{Q}(x,v)$.

        \item If subcase 1 is not true and $v$ is an endpoint of $Q$, $\tau_{\operatorname{notify}}$ stops at $v$. Let $P=Q|d_{Q}(x,v)$.

        \item If subcases 1 and 2 are not true, $\tau_{\operatorname{notify}}$ stops at $v$ if $d_{Q}(x,v)=\len+1$. Let $P=Q|d_{Q}(x,v)$.
    \end{enumerate}
        
Propagation stops at the first vertex meeting one of these conditions, so the traversed portion of $Q'$ contains $x$ at most as its last vertex. Hence $P$ is a path unless $\vend(P)=x$, in which case $x$ abandons the construction below.

In the confirmation substage, vertex $\vend(P)$ sends a message $\tau_{\operatorname{reply}}$ along $\operatorname{rev}(P)$ to $x$ whose content is according to the following exhaustive cases (with earlier cases taking precedence over the later ones). 
    \begin{itemize}
        \item If $\vend(P)\in N^2[\bigcup_{i=0}^{j-2}R_i(e,\len)]\cup\{x\}$, $\tau_{\operatorname{reply}}=(e,x,\vend(F),\A,\B,\suc=-1)$. Note that in case $\vend(P)=x$, $\vend(P)$ sends the message to itself.
        \item If $\vend(P)$ is an endpoint of $Q$, $\tau_{\operatorname{reply}}=(e,x,\vend(F),\A,\B,\suc=1)$.
        \item If $\vend(P)$ is not an endpoint of $Q$,  $\tau_{\operatorname{reply}}=(e,x,\vend(F),\A,\B,\suc=0)$.
    \end{itemize}
If $x$ receives $\tau_{\operatorname{reply}}=(e,x,\vend(F),\A,\B,\suc=1)$ or $\tau_{\operatorname{reply}}=(e,x,\vend(F),\A,\B,\suc=0)$, $x$ finishes the construction of $F+P$. If $x$ receives $\tau_{\operatorname{reply}}=(e,x,\vend(F),\A,\B,\suc=-1)$, $x$ abandons the construction.       
        
Now that $x$ has either finished or abandoned the construction in the \nameref{j_phase_construction_case_3}, we make the following observation about $\bichrom(P)$ which follows from the construction.
\begin{observation} \label{obsvn:is_bichromatic_path_phase_j_case_3}
    If $x$ finishes the construction of $F+P$ by the \nameref{j_phase_construction_case_3},
    \begin{itemize}
        \item $\bichrom(P)$ is a $\vend(F)$-maximal $\ABtuple$-bichromatic path of length at most $\len$ under $\chi$, and
        \item $x$ and $\vend(F)$ are not $\bichrom(P)$-related.
\end{itemize}
\end{observation}

This completes the description of the chain construction stage. In each of the three cases, appending $F$ and an initial segment of $P$ to $C$ again gives a self-avoiding chain. We state this below and defer the proof to \Cref{subsec:chain-properties-phase-j}.

\begin{restatable}{lemma}{IsNonOverlapping}\label{lemma:C+F+P_self_avoiding}
     Assume some vertex $x$ finishes the construction of candidate chain $F+P$ under $\Shift(\chi,C)$ where $C=F_1+P_1+\dots F_{j-1}+P_{j-1}$ is the unique $(j-1)$-step Vizing chain guaranteed by \Cref{fact_1}. Let $v\in V(P)\setminus\{x\}$. Then, $C+F+P|d_P(x,v)$ is a self-avoiding $j$-step Vizing chain of step length $\len$ on $e$ under $\chi$.
\end{restatable}

As in construction phase $1$, appending a successful candidate chain to $C$ yields a terminating Vizing chain, so that the construction on $e$ may stop at it. We state this below and defer the proof to \Cref{subsec:chain-properties-phase-j} as well.

\begin{restatable}{lemma}{SuccessfulIsTerminatingPhaseJ}
    \label{lemma:successful_is_terminating_phase_j}
     Assume some vertex $x$ finishes the construction of a successful candidate chain $F+P$ for $e$ in construction phase $j$, and let $C=F_1+P_1+\dots F_{j-1}+P_{j-1}$ be the unique $(j-1)$-step Vizing chain guaranteed by \Cref{fact_1}. Then $C+F+P$ is a terminating self-avoiding $j$-step Vizing chain of step length $\len$ on $e$ under $\chi$.
\end{restatable}

\paragraph*{Set construction stage for $x$.}    
From the chain construction stage, $F+P$ is a successful, hopeful, or unsuccessful candidate chain according as $\tau_{\operatorname{reply}}$ carries $\suc=1$, $\suc=0$, or $\suc=-1$.
Each vertex $v\in V(P)$ can determine which, since it sends or receives $\tau_{\operatorname{reply}}$.
     
     Each vertex $v\in V(P)\setminus\{x\}$, as well as vertex $x$, adds themselves to different sets according to the following cases.
    \begin{itemize}
        
        \item If $F+P$ is a successful candidate chain, then $x$ adds itself to $D_{j-1}(e,\len)$.  If $F+P$ is a hopeful candidate chain, then $x$ adds itself to $Q_{j-1}(e,\len)$.  If $F+P$ is an unsuccessful candidate chain, then $x$ adds itself to $B_{j-1}(e,\len)$.
        
        \item  If $F+P$ is a successful candidate chain or a hopeful candidate chain, $v$ adds itself to $R_{j}(e,\len)$. Note that each $v$ can locally determine whether $v\in V(P)\setminus\{x\}$, since $v\in V(P)\setminus\{x\}$ iff $v\neq x$ and $v$ sent or received $\tau_{\operatorname{reply}}$.
        
        \item 
        Suppose $F+P$ is a hopeful candidate chain (which implies $F+P$ was constructed by the \nameref{j_phase_construction_case_2} or the \nameref{j_phase_construction_case_3}).
        Let $v\in V(\bichrom(P))\setminus N^2[x]$. Let $vw=\eend(P|d_{P}(x,v))$. Note that $vw\in E(\bichrom(P))$ since $v\notin N^2[x]$. 
         Vertex $v$ now  
        
        \begin{itemize}
            \item  adds itself to $S_j(e,\len)$ if $v\notin N^2[\bigcup_{k=0}^{j-1}R_{k}(e,\len)]$ and $$\frac{\len}{10\Delta^4}\cdot(\Delta^3\cdot\mu_{\vend(F)}(x)+r) \leq d_P(v,x) < \frac{\len}{10\Delta^4}\cdot(\Delta^3\cdot\mu_{\vend(F)}(x)+r+1).$$
           
            \item  adds to $M_j(v)$ the tuple $(e,x,\theta,\theta',vw)$, where $\{\theta,\theta'\}$ is the color pair of $\bichrom(P)$ and $\theta=\chi(vw)$, provided $v\in S_j(e,\len)$.

        \end{itemize} 
        We claim that these additions are well defined. 
        Vertex $v$ knows the set $W(\vend(F))$ from the preprocessing step of algorithm $\fA$ (since $\bichrom(P)$ is a $\vend(F)$-maximal bichromatic path of length at most $\len$ by \Cref{obsvn:is_bichromatic_path_phase_j_case_2} and \Cref{obsvn:is_bichromatic_path_phase_j_case_3}). 
        Moreover, $\chi(vw)=\Shift(\chi,C+F)(vw)$ holds if $E(\bichrom(P))\cap E(C+F)=\emptyset$. This is true since $C+F+P$ is a self-avoiding $j$-step Vizing chain by \Cref{lemma:C+F+P_self_avoiding}.
    \end{itemize}
    
 A vertex $v$ in $R_j(e,\len)$ notifies its $2$-hop neighborhood that it is part of $R_j(e,\len)$ by sending a message $(v,e)$. This concludes the set construction stage as well as subphase $(c,r)$ in construction phase $j$.

\subsection{Selecting and Learning a Terminating  Multi-step Vizing Chain}\label{subsec:selecting-learning}

This completes the construction phases. We will show in \Cref{sec:existence-of-MSVC} that algorithm $\mathcal{A}$ has thereby constructed, for every uncolored edge, at least one terminating multi-step Vizing chain (\Cref{lemma:uncol_edges_have_terminating_chains}). We now explain how each uncolored edge chooses one such chain, and how the vertices of that chain learn that they belong to it.

The obstacle is that no vertex knows a multi-step chain it is part of in full. A vertex $x\in S_i(e,\len)$ knows only the candidate chain it constructed for $e$ in construction phase $i+1$, and whether that chain is successful or hopeful. A successful candidate chain already completes a multi-step Vizing chain on $e$. A hopeful one is of use only if it can be continued, through the candidate chains of the vertices that it reaches, to a terminating chain in one of the phases. We therefore proceed in three steps.

First (\Cref{subsubsec:marking-phase-i}), we determine which candidate chains can be continued to a terminating multi-step Vizing chain. For $e\in U$ and $0\leq i<T$, let $\zeta_i(e,\len):=D_i(e,\len)$, so that $\zeta_i(e,\len)$ initially holds those vertices of $S_i(e,\len)$ whose own candidate chain is successful; by \Cref{lemma-Di_non_empty}, we have $\zeta_i(e,\len)\neq\emptyset$ for some $0\leq i<T$. We then run $T-1$ \emph{marking phases}, indexed in reverse order, namely from $T-2$ down to $0$. In marking phase $i$, a vertex $x\in S_i(e,\len)$ joins $\zeta_i(e,\len)$ if its candidate chain reaches a vertex $v\in\zeta_{i+1}(e,\len)$ such that $M_{i+1}(v)$ contains a tuple naming $e$ and $x$. The set $\zeta_{i+1}(e,\len)$ receives vertices only at initialization and during marking phase $i+1$; running the phases in decreasing order of the index therefore means that no vertex joins $\zeta_{i+1}(e,\len)$ after marking phase $i$ has begun.

Second (\Cref{subsubsec:selecting-ce}), each uncolored edge $e$ uses these sets to fix a single chain $C_e$. The first chosen vertex is the unique vertex of $\zeta_0(e,\len)$. If a chosen vertex of $\zeta_i(e,\len)$ has a hopeful candidate chain, the next chosen vertex is the vertex of $\zeta_{i+1}(e,\len)$ that this candidate chain reaches; the process stops at the first chosen vertex whose candidate chain is successful. The candidate chains of the chosen vertices, each truncated at the next chosen vertex, are then concatenated to form $C_e$.

Third (\Cref{subsubsec:learning-ce}), we run $T$ \emph{relay} phases, in which every vertex of $C_e$ learns that it lies on $C_e$ and which segment of $C_e$ it lies in. We emphasize that $C_e$ is stored in a distributed fashion and not at a single vertex.

Two facts used in this subsection are proved later: that some $D_i(e,\len)$ is nonempty (\Cref{lemma-Di_non_empty}, proved in \Cref{sec:existence-of-MSVC}), and that the selection process of \Cref{subsubsec:selecting-ce} is well defined, in that $\zeta_0(e,\len)$ is a singleton and the process stops within $T$ steps (\Cref{claim:stopping-index}, proved in \Cref{sec:analysis-selection}).

\subsubsection{Marking Phase $i$ for $0\leq i \leq T-2$}\label{subsubsec:marking-phase-i}
Each marking phase is split into subphases, analogously to the construction phases. Each subphase is indexed by a pair $(c,r)$ with $c\in[\Delta^2+1]$ and $r\in[2\Delta^2]$. We iterate through all such pairs $(c,r)$ and run subphase $(c,r)$, described next.

\paragraph*{Subphase $(c,r)$ in marking phase $i$.}
Subphase $(c,r)$ is a single global subphase of the algorithm. It is not defined separately for each vertex. Instead, it consists of the local actions performed in parallel by all vertices $x\in S_i(e,\len)$, over all uncolored edges $e\in U$, for which $x$ constructed a candidate chain for $e$ in subphase $(c,r)$ of construction phase $i+1$.

Thus, in the description below, we fix an uncolored edge $e\in U$ and a vertex $x\in S_i(e,\len)$ that constructed a candidate chain $F+P$ for $e$ in subphase $(c,r)$ of construction phase $i+1$. Note that $F+P$ is the unique candidate chain constructed by $x$ for $e$ in construction phase $i+1$, by \Cref{obsvn:S_0_construct_candidates,obsvn:S_j-1_construct_candidates}.

We may restrict attention to the case where $F+P$ is hopeful: if $F+P$ is successful, then $x\in D_i(e,\len)\subseteq \zeta_i(e,\len)$ already, and if $F+P$ is unsuccessful, then $x$ abandoned its construction. So assume $F+P$ is hopeful; in particular $\bichrom(P)\neq\emptyset$.

Each subphase has two stages: the \emph{check} stage and the \emph{status} stage.

In the check stage, vertex $x$ checks whether there exists a vertex $v\in V(P)\setminus\{x\}$ such that $v\in \zeta_{i+1}(e,\len)$ and $(e,x,\theta,\theta',vw)\in M_{i+1}(v)$ for some $w\in N(v)$ and some ordering $(\theta,\theta')$ of the two colors $\A$ and $\B$ of $\bichrom(P)$ under $\chi$. We do not fix the order of the two colors in the tuple, since by the set construction stage in construction phase $i+1$ the first of them is $\chi(vw)$ and therefore depends on $v$. If such a vertex exists, then in the status stage vertex $x$ adds itself to $\zeta_i(e,\len)$.

We now describe these two stages from the perspective of the fixed vertex $x$.

\subparagraph*{Check stage for $x$.}
Vertex $x$ propagates the message $\tau_{\operatorname{check}}=(e,x,\vend(F),\A,\B)$ along $P$, where $\A$ and $\B$ are the two colors of $\bichrom(P)$ under $\chi$. This message is well defined because $x$ already knows $\vend(F)$, $\A$, and $\B$ from subphase $(c,r)$ of construction phase $i+1$. The propagation is also well defined, since each vertex of $P$ knows from the contents of $\tau_{\operatorname{check}}$ to which neighboring vertex the message must be forwarded.

The propagation stops at a vertex $v\in V(P)$ if one of the following happens:
\begin{itemize}
    \item $v\in \zeta_{i+1}(e,\len)$ and $(e,x,\theta,\theta',vw)\in M_{i+1}(v)$ for some $w\in N(v)$ and some $\{\theta,\theta'\}=\ABSet$; or
    \item $v=\vend(P)$.
\end{itemize}
The second condition is well defined because $\vend(P)$ knows, from the notification substage in construction phase $i+1$, that it is the last vertex of the path chain $P$ in the candidate chain $F+P$.

Let $P'$ be the subpath of $P$ from $x$ to the vertex at which the propagation stops.

\subparagraph*{Status stage for $x$.}
If $\vend(P')\in \zeta_{i+1}(e,\len)$ and $(e,x,\theta,\theta',\vend(P')w)\in M_{i+1}(\vend(P'))$ for some $w\in N(\vend(P'))$ and some $\{\theta,\theta'\}=\ABSet$, then $\vend(P')$ propagates the message $\tau_{\operatorname{status}}=(e,x,\vend(F),\A,\B,\vend(P'))$ back to $x$ along $\operatorname{rev}(P')$. Again, this propagation is well defined because each vertex of $\operatorname{rev}(P')$ knows, from the contents of $\tau_{\operatorname{status}}$, where the message has to be forwarded.

If $x$ receives the message $\tau_{\operatorname{status}}=(e,x,\vend(F),\A,\B,\vend(P'))$, then $x$ adds itself to $\zeta_i(e,\len)$. In this case, we say that $x$ is added to $\zeta_i(e,\len)$ because of $\vend(P')$.

This concludes subphase $(c,r)$ in marking phase $i$.

\subsubsection{Selecting the Terminating Multi-step Vizing Chain}\label{subsubsec:selecting-ce}

The marking phases determine, for each edge $e$, the sets $\zeta_i(e,\len)$; we now use them to single out a unique terminating multi-step Vizing chain. Fix an uncolored edge $e$.

We choose a sequence of vertices as follows. Let $x_0$ be the unique vertex in $\zeta_0(e,\len)$. For $i\ge 1$, suppose $x_{i-1}$ has been chosen. If $x_{i-1}\in D_{i-1}(e,\len)$, we stop and set $k_e:=i-1$; otherwise, by the definition of $\zeta_{i-1}(e,\len)$, there is a vertex $x'\in\zeta_i(e,\len)$ such that $x_{i-1}$ was added to $\zeta_{i-1}(e,\len)$ because of $x'$, and we set $x_i:=x'$. We call $k_e$ the \emph{stopping index} of $e$. That this process is well defined --- that $\zeta_0(e,\len)$ is a singleton, and that the process stops at some $k_e\leq T-1$ with $x_{k_e}\in D_{k_e}(e,\len)$ --- is shown in \Cref{subsec:selected-chain-correctness}.

For each $0\le i\le k_e$, let $F_{i+1}'+P_{i+1}'$ be the unique candidate chain constructed by $x_i$ for $e$ in construction phase $i+1$ (uniqueness by \Cref{obsvn:S_0_construct_candidates,obsvn:S_j-1_construct_candidates}). Set $F_{i+1}:=F_{i+1}'$ and $P_{i+1}:=P_{i+1}'|d_{P_{i+1}'}(x_i,x_{i+1})$ for $0\le i\le k_e-1$, and $F_{k_e+1}:=F_{k_e+1}'$, $P_{k_e+1}:=P_{k_e+1}'$. Here $P_{i+1}$ is well defined for $0\le i\le k_e-1$ since $x_i$ was added to $\zeta_i(e,\len)$ because of $x_{i+1}$, so $x_{i+1}\in V(\bichrom(P_{i+1}'))$.

Now let $C_e:=F_1+P_1+\dots+F_{k_e+1}+P_{k_e+1}$. By \Cref{lemma:selected-chain-vizing}, $C_e$ is a terminating self-avoiding $(k_e+1)$-step Vizing chain of step length $\len$ on $e$ under $\chi$.

\subsubsection{Learning the Terminating Multi-step Vizing Chain}\label{subsubsec:learning-ce}
Having selected $C_e$, we describe how each vertex of $C_e$ learns that it is part of $C_e$.
We proceed through $T$ \emph{relay} phases, run in increasing order of the index, from $1$ to $T$.
For any $e\in U$, let $C_e=F_{1}^e+P_{1}^e+\dots+F_{k_e+1}^e+P_{k_e+1}^e$ be the terminating Vizing chain we have chosen, where $k_e$ is the stopping index of $e$. Recall that $k_e\leq T-1$, so $C_e$ has at most $T$ segments. Note that $\cen(F_1^e)$ knows that it is $\cen(F_1^e)$, since $\cen(F_1^e)\in S_0(e,\len)$ and $S_0(e,\len)$ is a singleton. Moreover, for each $1\leq i\leq k_e$, the vertex $\cen(F_i^e)$ knows $\cen(F_{i+1}^e)$ from the marking phases. Similarly, for each $2\leq i\leq k_e+1$, the vertex $\cen(F_i^e)$ knows $\cen(F_{i-1}^e)$ from the marking phases.

Observe that, for each $1\leq i\leq k_e+1$, the vertex $\cen(F_i^e)$ knows $\vend(P_i^e)$: if $i\leq k_e$, then $P_i^e$ is truncated at $\vend(P_i^e)=\cen(F_{i+1}^e)$, which $\cen(F_i^e)$ knows from the marking phases; if $i=k_e+1$, then $P_i^e$ is the full path chain of the candidate chain constructed by $\cen(F_i^e)$, whose endpoint $\cen(F_i^e)$ knows from construction phase $i$.

\paragraph*{Relay phase $i$ for $1\leq i \leq T$.}
We do the following for any edge $e\in U$ with $k_e+1\geq i$.

If $i=1$, vertex $\cen(F_1^e)$ knows that it is $\cen(F_1^e)$ as noted above. If $i\geq 2$, then $\cen(F_i^e)=\vend(P_{i-1}^e)$, and hence $\cen(F_i^e)$ will have known by the end of relay phase $i-1$ that it is $\cen(F_i^e)$.

Vertex $\cen(F_i^e)$ sends a message $(e,\operatorname{chosen}=1)$ to vertices in $F_i^e$.
Vertex $\cen(F_i^e)$ propagates the message $(e,\cen(F_i^e),\vend(P_i^e),\A,\B,\operatorname{chosen}=1)$ across $P_i^e$, where $\A$ and $\B$ are the colors of $\bichrom(P_i^e)$ under $\chi$. This message is well defined since $\cen(F_i^e)$ knows $\vend(P_i^e)$, $\A$, and $\B$, as noted above. The propagation is well defined since each vertex in $P_i^e$ knows where to propagate to, based on the contents of the message. The propagation stops at the vertex $\vend(P_i^e)$ named in the message; note that $\vend(P_i^e)=\cen(F_{i+1}^e)$ whenever $i\leq k_e$.

All the vertices that got the message $(e,\operatorname{chosen}=1)$ or $(e,\cen(F_i^e),\vend(P_{i}^e),\A,\B,\operatorname{chosen}=1)$ know that they are part of $F_i^e+P_i^e$.

\section{Analysis of the Construction Phases}\label{sec:analysis-construction-phases}

\subsection{Properties of the Chain Construction Stage in Construction Phase 1}\label{subsec:properties-phase-1}

We first record properties of the candidate chains built in the chain construction stage of construction phase $1$.

We begin with the self-avoidance lemma used in \Cref{subsec:implementation-phase-1}, which underlies the well-definedness of the set construction stage and recurs throughout the later analysis.

\IsSelfAvodingVizingChainPhaseOne*
 \begin{proof}
For ease of writing, we use the same variables as in the \nameref{Phase_1_construction_case_1} and the \nameref{Phase_1_construction_case_2}. Note that $F+P|d_{P}(x,v)$ is a self-avoiding $1$-step Vizing chain of step length $\len$ on $e$ under $\chi$ if $F+P$ is a self-avoiding $1$-step Vizing chain of step length $\len$ on $e$ under $\chi$: a prefix of a bichromatic path chain is a bichromatic path chain of no greater length, and since $\bichrom(P|d_{P}(x,v))\subseteq \bichrom(P)$, the vertices $x$ and $\vend(F)$ remain not $\bichrom(P|d_{P}(x,v))$-related. Hence, we will prove that $F+P$ is a self-avoiding $1$-step Vizing chain of step length $\len$ on $e$ under $\chi$.

We will start with showing that $F+P$ is a $1$-step Vizing chain of step length $\len$ on $e$ under $\chi$. By construction, $F$ is a fan chain under $\chi$ on $e$. Hence, we only need to show that $P$ is a bichromatic path chain on $(x\vend(F),\vend(F))$ under $\Shift(\chi,F)$. Clearly $x\vend(F)$ is uncolored under $\Shift(\chi,F)$ since $F$ is a fan chain. Assume  $x$ finished the construction of $F+P$ by the \nameref{Phase_1_construction_case_1}. Then $P=(xz_p)$ is a bichromatic path chain of length $1$ on $(x\vend(F),\vend(F))$ under $\Shift(\chi,F)$ and we are done. Hence, assume  $x$ finished the construction of $F+P$ by the \nameref{Phase_1_construction_case_2}. Then we have that
\begin{itemize}
    \item $\A\in A(x,\Shift(\chi,F))$ and $\B\in A(\vend(F),\Shift(\chi,F))$ by \Cref{obsvn:fan_available_color} because $\A\in A(x,\chi)$ and $\B\in A(\vend(F),\chi)$, 

    \item $\bichrom(P)$ is colored the same in both $\chi$ and $\Shift(\chi,F)$. This is because $x$ and $\vend(F)$ are not $\bichrom(P)$-related by \Cref{obsvn:vend(F)_maximal_point}, which implies $E(\bichrom(P))\cap E(F)=\emptyset$ as every edge of $F$ is incident to $x$.
\end{itemize}
Moreover, $P$ is a path in $G$: $\bichrom(P)$ is a subpath of the path $Q'$, and $x\notin V(\bichrom(P))$ since $x$ and $\vend(F)$ are not $\bichrom(P)$-related by \Cref{obsvn:vend(F)_maximal_point}. Also, the first edge of $\bichrom(P)$ is colored $\A$ under $\chi$, since it is incident to $\vend(F)$ and $\B\in A(\vend(F),\chi)$. Together, these conditions imply that $P$ is an $\ABtuple$-bichromatic path chain under $\Shift(\chi,F)$. Hence, we have that $F+P$ is a $1$-step Vizing chain on $e$ under $\chi$ irrespective of whether $x$ built it by the \nameref{Phase_1_construction_case_1} or the \nameref{Phase_1_construction_case_2}.

We now show that $F+P$ is self-avoiding. Since $F+P$ is a $1$-step Vizing chain, $x$ and $\vend(F)$ being not $\bichrom(P)$-related would imply that $F+P$ is self-avoiding. The other conditions for self-avoidance are vacuous for a $1$-step Vizing chain, since there is only one fan and one path chain. If $x$ finished the construction by the \nameref{Phase_1_construction_case_1}, $\bichrom(P)=\emptyset$ and we are done. If  $x$ finished the construction by the \nameref{Phase_1_construction_case_2}, $x$ and $\vend(F)$ are not $\bichrom(P)$-related by \Cref{obsvn:vend(F)_maximal_point}. Hence, $F+P$ is self-avoiding.
\end{proof}

Next we verify that a successful candidate chain is a terminating Vizing chain.

\SuccessfulIsTerminatingPhaseOne*
\begin{proof}
For ease of writing, we use the same variables as in the \nameref{Phase_1_construction_case_1} and the \nameref{Phase_1_construction_case_2}; in addition, for $1\leq i\leq p-1$ we write $\eta_i$ for the representative available color at $z_i$ in $F'$.

We first observe that $\vend(P)\in V(P)\setminus\{x\}$. In the \nameref{Phase_1_construction_case_1} this holds because $\vend(P)=\vend(F)\neq x$. In the \nameref{Phase_1_construction_case_2}, the path $\bichrom(P)$ is nonempty and $\vend(P)\in V(\bichrom(P))$, while $x\notin V(\bichrom(P))$ by \Cref{obsvn:vend(F)_maximal_point}. Taking $v=\vend(P)$ in \Cref{lemma:F+P_is_self-avoiding_Vizing_chain} therefore shows that $F+P$ is a self-avoiding $1$-step Vizing chain of step length $\len$ on $e$ under $\chi$. By \Cref{obsvn:maximal_implies_terminating}, it suffices to prove that $P$ is maximal under $\Shift(\chi,F)$.

If $x$ constructed $F+P$ by the \nameref{Phase_1_construction_case_1}, then $P=(x\vend(F))$ and $\eta\in A(x,\Shift(\chi,F))\cap A(\vend(F),\Shift(\chi,F))$, so $P$ is maximal under $\Shift(\chi,F)$. Assume therefore that $x$ constructed $F+P$ by the \nameref{Phase_1_construction_case_2}.

Since $F+P$ is a successful candidate chain, $\tau_{\operatorname{reply}}$ carries $\suc=1$, so $\vend(P)$ is an endpoint of $Q$. We saw above that $\vend(P)\neq x$. Also $\vend(P)\neq\vend(F)$: since $\A\in A(x,\chi)$ and $A(x,\chi)\cap A(\vend(F),\chi)=\emptyset$, the vertex $\vend(F)$ is incident to an $\A$-colored edge, so $Q'$ has length at least $1$ and $\vend(F)$ is an interior vertex of $Q$. Hence $\vend(P)$ is the endpoint of $Q'$ other than $\vend(F)$, and since $Q'$ is a maximal $\ABtuple$-bichromatic path under $\chi$ there is a color $\theta\in A(\vend(P),\chi)\cap\ABSet$. It suffices to show that $\theta\in A(\vend(P),\Shift(\chi,F))$.

By \Cref{obsvn:fan_available_color}, we only have to show that $\theta\neq\eta_i$ whenever $\vend(P)=z_i$ for some $z_i\in V(F)\setminus\{x,\vend(F)\}$. First, $\eta_i\neq\A$ for every $i$, since $\eta_i\in A(z_i,\chi)$, $\A\in A(x,\chi)$ and $A(x,\chi)\cap A(z_i,\chi)=\emptyset$. Second, the colors $\eta_1,\dots,\eta_{p-1}$ are pairwise distinct and $\eta_q=\B$, so $\eta_i=\B$ forces $i=q$. It therefore remains to rule out $\vend(P)=z_q$ with $z_q\in V(F)\setminus\{x,\vend(F)\}$, that is, with $\vend(F)=z_p$. In that case $x$ received no $\tau_{\operatorname{return}}$, so $\tau_{\operatorname{search}}$ stopped neither at $x$ nor at $z_q$. As $\tau_{\operatorname{search}}$ and $\tau_{\operatorname{notify}}$ traverse the same trail $(xz_p)+K'=(xz_p)+Q'$ and both stop after at most $\len+1$ hops, and $\tau_{\operatorname{notify}}$ reached the endpoint $\vend(P)$ of $Q'$, the message $\tau_{\operatorname{search}}$ reached and stopped at $\vend(P)$ as well. Since it did not stop at $z_q$, we conclude $\vend(P)\neq z_q$.
\end{proof}

The next observation identifies which vertices construct candidate chains in construction phase $1$: exactly the vertices of $S_0(e,\len)$, each constructing one chain per uncolored edge. We use it repeatedly to pass between a candidate chain and the vertex that built it.
\begin{observation} \label{obsvn:S_0_construct_candidates}
    Any candidate chain for an edge $e\in U$ constructed in construction phase $1$ is constructed by some $x$ in $S_0(e,\len)$. Conversely, any $x$ in $S_0(e,\len)$ constructs a unique candidate chain for $e$ in construction phase $1$. 
\end{observation}
\begin{proof}
    By the description of subphase $(c,r)$ in construction phase $1$, a vertex $x$ constructs a candidate chain for $e$ if and only if $\psi(x)=c$, $e\in M_0(x)$, and $\lambda_{x,0}(e)=r$. By definition of $M_0(x)$, $e\in M_0(x)$ iff $x\in S_0(e,\len)$. Since $\lambda_{x,0}$ is injective, $x$ constructs at most one candidate chain for $e$ across all subphases, and constructs exactly one in the unique subphase $(\psi(x),\lambda_{x,0}(e))$.
\end{proof}

A candidate chain is hopeful precisely when its path chain stopped at the step length rather than at an endpoint of $Q$. The next observation makes this quantitative, and it is what forces the sets $R_1(e,\len)$ to be large in the counting arguments of \Cref{sec:existence-of-MSVC}.
   \begin{observation} \label{obsvn:hopeful_is_long_in_phase_1}
    Let $F+P$ be a hopeful candidate chain constructed in construction phase $1$. Then length of $P$ is $\len+1$.
\end{observation}
\begin{proof}
    A hopeful candidate chain receives $\tau_{\operatorname{reply}}$ with $\suc=0$, which occurs only in the \nameref{Phase_1_construction_case_2} when $\vend(P)$ is not an endpoint of $Q$. By the notification substage, this means $\tau_{\operatorname{notify}}$ traversed $\len+1$ edges along $Q$ before stopping. Hence $P=Q|(\len+1)$ has length $\len+1$.
\end{proof}

Finally, we record the shape of the bichromatic part of a candidate chain. This is used whenever we count the chains passing through a given vertex or edge, and when bounding the messages sent in a marking phase.
\begin{observation} \label{claim:Phase_1_bichrom(P)_is_bichromatic}
    Let $F+P$ be a candidate chain constructed in construction phase $1$. Then either $\bichrom(P)=\emptyset$ or $\bichrom(P)$ is a $\vend(F)$-maximal bichromatic path under $\chi$.
\end{observation}
\begin{proof}
If $F+P$ was constructed in the \nameref{Phase_1_construction_case_1}, then $\bichrom(P)=\emptyset$ by construction. If $F+P$ was constructed in the \nameref{Phase_1_construction_case_2}, $\bichrom(P)$ is a $\vend(F)$-maximal bichromatic path under $\chi$ by \Cref{obsvn:vend(F)_maximal_point}.
\end{proof}

\subsection{Properties of the Set Construction Stage in Construction Phase 1}\label{subsec:set-properties-phase-1}

We now record properties of the sets updated in the set construction stage of construction phase $1$. We begin with the fact that every fiber of $M_1(\cdot)$ holds at most one tuple.

\begin{observation} \label{obsvn:M_1_no_duplicates}
    Let $x\in V(G)$. Then $|M_1(x;e)|\leq 1$ for every edge $e\in U$ after
    construction phase $1$.
\end{observation}
\begin{proof}
    By the deduplication step at the end of construction phase $1$, vertex $x$
    replaces every nonempty fiber $M_1(x;e)$ by a singleton subset of it.
\end{proof}

We turn to the sets that the vertices of a candidate chain join. A vertex joins $R_1(e,\len)$ exactly when it lies on the path chain of a finished candidate chain for $e$, other than the vertex that built that chain; a vertex lying only on the fan does not join it.

\begin{observation} \label{obsvn:in_R_1_iff_in_P}
  Let $v\in V(G)$. Then $v$ is added to $R_{1}(e,\len)$ in construction phase $1$ if and only if  some $x\in S_{0}(e,\len)$ finished the construction of a candidate chain $F+P$ for $e$ in construction phase $1$ such that $v\in V(P)\setminus\{x\}$.
\end{observation}
\begin{proof}
    By the set construction stage in construction phase $1$, $v$ is added to $R_1(e,\len)$ if and only if $v\in V(P)\setminus\{x\}$ for some candidate chain $F+P$ for $e$ that was finished by some vertex $x$. By \Cref{obsvn:S_0_construct_candidates}, $x\in S_0(e,\len)$.
\end{proof}

A vertex can moreover join $R_1(\cdot,\len)$ for only $O(\Delta^2)$ uncolored edges within a single subphase, which amounts to bounding the number of different candidate chains that pass through it.

\begin{lemma} \label{lemma:bounded_addition_Phase_1}
Fix $c\in [\Delta^2+1]$ and $r\in [\Delta]$. Let $U^1_v$ be the set of edges $e\in U$ such that $v$ adds itself to $R_1(e,\len)$ in subphase $(c,r)$ of construction phase $1$. Then $|U^1_v|\leq O(\Delta^2)$.
\end{lemma}

\begin{proof}
Note that a vertex $v$ adds itself to $R_1(e,\len)$ only if there exists a vertex $x\in S_0(e,\len)$ that finishes the construction of a candidate chain $F+P$ for $e$ such that $v\in V(P)\setminus\{x\}$. Since $\lambda_{x,0}$ is an injective function, vertex $x$ can construct at most one candidate chain in subphase $(c,r)$. Hence, it suffices to show that there are at most $O(\Delta^2)$ different vertices $x$ that can construct a candidate chain $F+P$ in subphase $(c,r)$ such that $v\in V(P)$.

Suppose first that $v\in N(x)$. Since $\psi$ is a distance-$2$ coloring of $G$, no other vertex in $N[v]$ constructs a candidate chain in subphase $(c,r)$. Hence $x$ is the unique neighbor of $v$ that constructs a chain $F+P$ on some edge $e$ such that $v$ may add itself to $R_1(e,\len)$.

Now suppose that $v\notin N(x)$. Then $v\in V(P)\setminus\{x,\vend(F)\}$, and therefore $P$ has length greater than $1$. In particular, $x$ constructed $F+P$ by the \nameref{Phase_1_construction_case_2}. Then the following hold:
\begin{itemize}
    \item $\vend(F)$ is an endpoint of a maximal bichromatic path under $\chi$ that contains $v$, by \Cref{obsvn:vend(F)_maximal_point}, and
    \item $x$ is the unique neighbor of $\vend(F)$ with $\psi(x)=c$.
\end{itemize}

Since there are at most $O(\Delta^2)$ maximal bichromatic paths under $\chi$ that contain $v$, the vertex $\vend(F)$ must be one of the endpoints of one of these $O(\Delta^2)$ paths. Consequently, $x$ must be the unique neighbor colored $c$ of one of these endpoints. Hence there are at most $O(\Delta^2)$ vertices $x$ that can construct a candidate chain such that $v$ adds itself to $R_1(e,\len)$.
\end{proof}

The vertex that constructs a candidate chain does not join $R_1(e,\len)$. It instead joins $D_0(e,\len)$ or $Q_0(e,\len)$ based on the outcome of its chain, and it does so unambiguously.

  \begin{observation}\label{claim:x_in_D_or_Q}
       Let $x\in S_0(e,\len)$. Then
      $x$ is added to exactly one of $D_0(e,\len)$ or $Q_0(e,\len)$ in construction phase $1$.
   \end{observation} 
\begin{proof}
      Since $x\in S_0(e,\len)$, $e\in M_0(x)$ by definition of $M_0(x)$. Hence $x$ constructs a candidate chain for $e$ in subphase $(\psi(x),\lambda_{x,0}(e))$ in construction phase $1$. Since $M_0(x)$ is a set and contains $e$ only once, $x$ constructs only this candidate chain for $e$ in construction phase $1$. Hence $x$ is added to $D_0(e,\len)$ or $Q_0(e,\len)$ exactly in the sub-phase  $(\psi(x),\lambda_{x,0}(e))$. By the set construction stage for $x$, $x$ is added to either $D_0(e,\len)$ or $Q_0(e,\len)$.
  \end{proof}

The remaining two observations concern the vertices that carry the construction into the next phase. Such a vertex lies in $S_1(e,\len)$ and holds a tuple in $M_1(\cdot)$; these two conditions are in fact equivalent, and the tuple is then unique.

\begin{observation} \label{obsvn:S_0_M_0_relation}
Let $v\in V(G)$. After construction phase $1$,
$v\in S_1(e,\len)$ if and only if $(e,x,\A',\B',vw)\in M_1(v)$ for unique $x\in S_0(e,\len)$, $\A',\B'\in [\Delta+1]$ and $w\in N(v)$. 
\end{observation}
\begin{proof}
Suppose $(e,x,\A',\B',vw)\in M_1(v)$. Then $v\in S_1(e,\len)$ from the set construction stage for $x$. 
We will prove the other direction now. Suppose $v\in S_1(e,\len)$. Note that we just have to show that there exist an element of the form $(e,x,\A',\B',vw)$ that was added to $M_1(v)$. It will be unique since $M_1(v)$ has no duplicates after construction phase $1$. Note that $v$ is added to $S_1(e,\len)$ only if some $x$ constructed a hopeful candidate chain $F+P$ such that $v\in V(P)\setminus N^2[x]$ and $\bichrom(P)$ is some $(\A',\B')$-bichromatic path under $\chi$. Suppose $w$ is the predecessor of $v$ along $P$. Note that $vw\in E(\bichrom(P))$ since $v\in V(P)\setminus N^2[x]$. Also, either $\chi(vw)=\A'$ or $\chi(vw)=\B'$, since $\bichrom(P)$ is $(\A',\B')$-bichromatic under $\chi$. Depending on which of the two holds, $(e,x,\A',\B',vw)$ or $(e,x,\B',\A',vw)$ is added to $M_1(v)$ respectively. Hence done.
\end{proof}

A tuple in $M_1(v)$ in turn determines the chain that produced it, together with the properties of that chain that construction phase $2$ relies on.

\begin{observation}\label{obsvn:prprty_of_1_step_chain}
Let $(e,x,\A,\B,vw)\in M_1(v)$. Then $x\in S_0(e,\len)$ and $x$ finishes the construction of a unique candidate chain $F+P$ for $e$ in construction phase $1$. Moreover, $F+P$ has the following properties:
\begin{itemize}
    \item $F+P$ is a $1$-step self-avoiding Vizing chain of step length $\len$ on $e$ under $\chi$,
    \item $V(\bichrom(P))\subseteq R_1(e,\len)$,
    \item $\bichrom(P)$ is a $\vend(F)$-maximal bichromatic path under $\chi$,
    \item $P$ is either an $(\A,\B)$-bichromatic path chain or a $(\B,\A)$-bichromatic path chain under $\Shift(\chi,F)$,
    \item $vw=\eend(P|d_P(x,v))$, and $\Shift(\chi,F)(vw)=\A$.
\end{itemize}
\end{observation}
   \begin{proof}
       Note that adding $(e,x,\A,\B,vw)$ to $M_1(v)$ in set construction stage is only possible when vertex $x$ finishes the construction of a  candidate chain, say $F+P$, with the properties that
      \begin{itemize}
          \item $v\in V(P)\setminus N^2[x]$ and
          \item $vw = \eend(P|d_P(x,v))$ and $\chi(vw)=\A$.
      \end{itemize}
    $F+P$ is unique since $x$ constructs only one candidate chain for each element in $M_0(x)$ and $M_0(x)$ is a set. 
    
    By \Cref{obsvn:S_0_construct_candidates}, $x\in S_0(e,\len)$.
    Moreover, since $v\in V(P)\setminus N^2[x]$, we have $\bichrom(P)\neq \emptyset$ and hence $F+P$ was constructed in the \nameref{Phase_1_construction_case_2}. 

    Therefore, taking $v=\vend(P)$ in \Cref{lemma:F+P_is_self-avoiding_Vizing_chain}, the chain $F+P$ is a self-avoiding $1$-step Vizing chain of step length $\len$ on $e$ under $\chi$.
    By the definition of the set construction stage, $V(\bichrom(P))\subseteq R_1(e,\len)$. Furthermore, $\bichrom(P)$ is a $\vend(F)$-maximal bichromatic path under $\chi$ by 
    \Cref{obsvn:vend(F)_maximal_point}.
   
    It remains to prove that $P$ is either an $(\A,\B)$-bichromatic path chain or a $(\B,\A)$-bichromatic path chain under $\Shift(\chi,F)$. Since $F+P$ is a $1$-step Vizing chain, $P$ is a path chain under $\Shift(\chi,F)$. 
    Also, $\bichrom(P)$ is bichromatic under $\chi$, and this remains true under $\Shift(\chi,F)$ because $E(F)\cap E(\bichrom(P))=\emptyset$. Indeed, the vertices $x$ and $\vend(F)$ are not $\bichrom(P)$-related by \Cref{obsvn:vend(F)_maximal_point}, which implies $E(F)\cap E(\bichrom(P))=\emptyset$ as every edge of $F$ is incident to $x$. Hence $P$ is either an $\ABtuple$- or a $\BAtuple$-bichromatic path chain under $\Shift(\chi,F)$.

    Finally, $E(F)\cap E(\bichrom(P))=\emptyset$ also gives $\Shift(\chi,F)(vw)=\chi(vw)=\A$, since $vw\in E(\bichrom(P))$ by the set construction stage.

    This proves all the required properties.
\end{proof}

\subsection{Runtime of Construction Phase 1} \label{subsec:runtime-phase-1}

One readily sees that a subphase runs in $O(\len)$ rounds in the LOCAL model, since each message it sends travels at most $O(\len)$ hops. By \Cref{obsvn:LOCAL_to_CONGEST}, the work therefore lies in bounding how many messages cross a single edge. In the chain construction stage this bound comes from the routes the messages take: each message travels along an edge incident to the vertex building the chain, or along a maximal bichromatic path under $\chi$.

\begin{lemma}\label{lemma:runtime_1st_phase_subphase}
Each subphase in construction phase $1$ runs in $O(\Delta^3\len)$ rounds in CONGEST.
\end{lemma}

\begin{proof}
We analyze the runtime of a fixed subphase $(c,r)$ in construction phase $1$.

Consider first the preparation substage. By construction, no messages are sent during the preparation substage; each vertex $x$ with $\psi(x)=c$ constructs the fan chain $F'$ internally using \Cref{lemma:fan-addition-phase-1}. Hence the preparation substage runs in $O(1)$ rounds in CONGEST.

Consider next the search substage. This substage clearly runs in $O(\len)$ rounds in the LOCAL model. Fix an edge $e\in E(G)$. Any message that traverses $e$ during the search substage is of the form
$$
\tau_{\operatorname{search}}=(f,x,z,\A,\B),
$$
where $f\in U$, $x,z\in V(G)$ with $\psi(x)=c$, and $\A,\B\in[\Delta+1]$. Each such message can be encoded using $O(\log n)$ bits, and no such message traverses the same edge more than once. Hence, it suffices to bound the number of such messages that traverse $e$.

A message $\tau_{\operatorname{search}}$ can traverse $e$ only if one of the following holds:
\begin{itemize}
    \item $e$ is incident to $x$, where $\psi(x)=c$, or
    \item $e$ lies on a maximal $\ABtuple$-bichromatic path under $\chi$, and $x$ is adjacent to an endpoint of this path, with $\psi(x)=c$.
\end{itemize}
There are at most $O(\Delta^2)$ choices for such a vertex $x$. Therefore, at most $O(\Delta^2)$ search messages traverse $e$. The substage also sends the messages $\tau_{\operatorname{return}}$, but a message $\tau_{\operatorname{return}}=(f,x)$ traverses $e$ only if $e$ is incident to $x$; since $\psi$ is a proper coloring of $G^2$, at most one endpoint of $e$ is colored $c$, and that vertex serves a single edge $f$ in the subphase, so at most one such message traverses $e$. By \Cref{obsvn:LOCAL_to_CONGEST}, the search substage can be simulated in $O(\Delta^2\len)$ rounds in the CONGEST model.

Consider next the notification substage. This substage also clearly runs in $O(\len)$ rounds in the LOCAL model. Fix an edge $e\in E(G)$. Any message that traverses $e$ during the notification substage is of the form
$$
\tau_{\operatorname{notify}}=(f,x,y,\neg,\A) \quad \text{or} \quad \tau_{\operatorname{notify}}=(f,x,y,\A,\B),
$$
where $f\in U$, $x,y\in V(G)$, $\psi(x)=c$, and $\A,\B\in[\Delta+1]$. Each such message can be encoded using $O(\log n)$ bits, and no such message traverses the same edge more than once. Hence, it suffices to bound the number of notification messages that traverse $e$.

A message $\tau_{\operatorname{notify}}$ can traverse $e$ only if one of the following holds:
\begin{itemize}
    \item $e$ is incident to $x$, where $\psi(x)=c$, or
    \item $e$ lies on a maximal $\ABtuple$-bichromatic path under $\chi$, and $x$ is adjacent to an endpoint of this path, with $\psi(x)=c$.
\end{itemize}
Again, there are at most $O(\Delta^2)$ choices for such a vertex $x$. Therefore, at most $O(\Delta^2)$ notification messages traverse $e$. By \Cref{obsvn:LOCAL_to_CONGEST}, the notification substage can be simulated in $O(\Delta^2\len)$ rounds in the CONGEST model.

Now consider the confirmation substage. Any message that traverses an edge $e$ during this substage is of the form
\begin{equation*}
    \begin{split}
        &\tau_{\operatorname{reply}}=(f,x,y,\neg,\A,\suc=1),\ \\ &\tau_{\operatorname{reply}}= (f,x,y,\A,\B,\suc=0)
\quad\text{or}\\
&\tau_{\operatorname{reply}}=(f,x,y,\A,\B,\suc=1),
    \end{split}
\end{equation*}

where $f\in U$, $x,y\in V(G)$, $\psi(x)=c$, and $\A,\B\in[\Delta+1]$. Each such message can be encoded using $O(\log n)$ bits, and no such message traverses the same edge more than once. Hence, it suffices to bound the number of reply messages that traverse $e$.

A message $\tau_{\operatorname{reply}}$ can traverse $e$ only if one of the following holds:
\begin{itemize}
    \item $e$ is incident to $x$, where $\psi(x)=c$, or
    \item $e$ lies on a maximal $\ABtuple$-bichromatic path under $\chi$, and $x$ is adjacent to an endpoint of this path, with $\psi(x)=c$.
\end{itemize}
As above, there are at most $O(\Delta^2)$ choices for such a vertex $x$. Therefore, at most $O(\Delta^2)$ reply messages traverse $e$. Since the confirmation substage runs in $O(\len)$ rounds in the LOCAL model, \Cref{obsvn:LOCAL_to_CONGEST} implies that it can be simulated in $O(\Delta^2\len)$ rounds in the CONGEST model.

Finally, consider the set construction stage. All additions to the sets $D_0(e,\len)$, $Q_0(e,\len)$, $R_1(e,\len)$, $S_1(e,\len)$, and $M_1(v)$ are decided locally from information received during the chain construction stage; the only messages sent are of the form $(v,e)$ for some $v\in V(G)$ and $e\in U$, sent when a vertex $v$ notifies its $2$-hop neighborhood that it added itself to $R_1(e,\len)$. Hence, if $(v,e)$ passes through an edge $e'$, then $v$ is adjacent to an endpoint of $e'$ and $v\in R_1(e,\len)$. By \Cref{lemma:bounded_addition_Phase_1}, $v$ adds itself to $R_1(e,\len)$ in this subphase for at most $O(\Delta^2)$ edges $e$. Since there are $O(\Delta)$ vertices adjacent to an endpoint of $e'$, at most $O(\Delta^3)$ such messages pass through $e'$. Each such message can be described by $O(\log n)$ bits, and the stage runs in $O(1)$ rounds in the LOCAL model. Hence, by \Cref{obsvn:LOCAL_to_CONGEST}, the set construction stage can be simulated in $O(\Delta^3)$ rounds in the CONGEST model.

Therefore, each subphase in construction phase $1$ runs in $O(\Delta^3+\Delta^2\len)$ rounds in CONGEST, which is $O(\Delta^3\len)$.
\end{proof}
    
Summing over the subphases of construction phase $1$ now gives \Cref{lemma:runtime_1st_phase}.
\RuntimePhaseOne*
\begin{proof}
    Since there are $O(\Delta^3)$ subphases in construction phase $1$ and each subphase runs in $O(\Delta^3\len)$ rounds in CONGEST by \Cref{lemma:runtime_1st_phase_subphase}, construction phase $1$ runs in $O(\Delta^6\len)$ rounds in CONGEST.
\end{proof}

\subsection{Properties of the Chain Construction Stage in Construction Phase \texorpdfstring{$j$}{j} for \texorpdfstring{$2\leq j\leq \numphases$}{2 ≤ j ≤ \numphases}}\label{subsec:chain-properties-phase-j}

We first record properties of the candidate chains built in the chain construction stage of construction phase $j$ for $2\leq j\leq \numphases$, including the proofs deferred from \Cref{subsubsec:implementation-sub-phase-in-phase-j}.
Throughout this subsection we use the variables of \Cref{subsubsec:implementation-sub-phase-in-phase-j}: vertex $x$ holds a tuple $(e,x',\A,\B,xy)\in M_{j-1}(x)$, the chain $C=F_1+P_1+\dots+F_{j-1}+P_{j-1}$ is the $(j-1)$-step self-avoiding Vizing chain on $e$ given by \Cref{fact_1}, and $F+P$ is the candidate chain that $x$ constructs on $xy$ under $\Shift(\chi,C)$.

The first observation identifies which vertices construct candidate chains in construction phase $j$: exactly the vertices of $S_{j-1}(e,\len)$, each constructing one chain per uncolored edge. As in construction phase $1$, we use it to pass between a candidate chain and the vertex that built it.
\begin{observation} \label{obsvn:S_j-1_construct_candidates}
    Any candidate chain for $e$ constructed in construction phase $j$ is constructed by some $x$ in $S_{j-1}(e,\len)$. Conversely, any $x\in S_{j-1}(e,\len)$ constructs a unique candidate chain for $e$ in construction phase $j$.
\end{observation}
\begin{proof}
    By the description of subphase $(c,r)$ in construction phase $j$, a vertex $x$ constructs a candidate chain for $e$ if and only if $\psi(x)=c$ and there is a tuple $(e,x',\A,\B,xy)\in M_{j-1}(x)$ with $\lambda_{x,j-1}((e,x',\A,\B,xy))=r$. By statement $\Pi_3(j-1)$, such a tuple exists (for some $x',\A,\B,y$) if and only if $x\in S_{j-1}(e,\len)$.
    Since $|M_{j-1}(x;e)|\leq 1$ at the start of construction phase $j$ by \Cref{obsvn:M_1_no_duplicates,obsvn:M_j_no_duplicates} and $\lambda_{x,j-1}$ is injective, $x$ constructs at most one candidate chain for $e$ across all subphases, and constructs exactly one when $x\in S_{j-1}(e,\len)$.
\end{proof}

We will also use the following consequence of \Cref{fact_1} repeatedly.
\begin{observation}\label{obsvn:C_in_nbhd_of_R}
Let $(e,x',\A,\B,xy)\in M_{j-1}(x)$ and let $C=F_1+P_1+\dots+F_{j-1}+P_{j-1}$ be the chain guaranteed by \Cref{fact_1}. Then
$$V(F_1+P_1+\dots+F_{j-2}+P_{j-2}+F_{j-1})\subseteq N\left[\bigcup_{k=0}^{j-2}R_k(e,\len)\right].$$
\end{observation}
\begin{proof}
For $1\leq k\leq j-1$, the vertex $\cen(F_k)$ constructed the candidate chain $F_k^*+P_k^*$ for $e$ in construction phase $k$ by \Cref{fact_1}. Hence $\cen(F_k)\in S_{k-1}(e,\len)$ by \Cref{obsvn:S_0_construct_candidates,obsvn:S_j-1_construct_candidates}, and $S_{k-1}(e,\len)\subseteq R_{k-1}(e,\len)$ by the set construction stages. Therefore $V(F_k)\subseteq N[\cen(F_k)]\subseteq N[R_{k-1}(e,\len)]$. Moreover, for $1\leq k\leq j-2$, we have $V(P_k)\subseteq V(F_k)\cup V(\bichrom(P_k))\subseteq V(F_k)\cup R_k(e,\len)$, where the first containment holds since $\estart(P_k)=\eend(F_k)$ and the second by \Cref{fact_1}. Combining these containments yields the claim.
\end{proof}

We next record that the vertices of $F+P$ keep their distance from the sets built in the earlier construction phases.
\begin{observation}\label{obsvn:vertex_disjointness_j_j-2}
If $x$ finished the construction of $F+P$, then
    \begin{itemize}
        \item $V(F)\cap N[\bigcup_{i=0}^{j-2}R_i(e,\len)]=\emptyset$, and
        \item  $V(P)\cap N^2[\bigcup_{i=0}^{j-2}R_i(e,\len)]=\emptyset$.
    \end{itemize}
\end{observation}
\begin{proof}
    By \Cref{fact_2}, $x\notin  N^2[\bigcup_{i=0}^{j-2}R_i(e,\len)]$. Since $x$ is the center of $F$, this gives $V(F)\cap N[\bigcup_{i=0}^{j-2}R_i(e,\len)]=\emptyset$ in all three cases.

    It remains to show that $V(P)\cap N^2[\bigcup_{i=0}^{j-2}R_i(e,\len)]=\emptyset$. Suppose first that $x$ finished the construction of $F+P$ by the \nameref{j_phase_construction_case_1}. Then $z_p\notin N^2[\bigcup_{i=0}^{j-2}R_i(e,\len)]$, since otherwise $\tau_{\operatorname{reply}}$ would have carried $\suc=-1$ and $x$ would have abandoned the construction; as $V(P)=\{x,z_p\}$, the claim follows. Suppose instead that $x$ finished the construction of $F+P$ by the \nameref{j_phase_construction_case_2} or the \nameref{j_phase_construction_case_3}. By construction in the notification substage, $\tau_{\operatorname{notify}}$ stops as soon as it reaches any vertex in $N^2[\bigcup_{i=0}^{j-2}R_i(e,\len)]$, so $V(P)\setminus\{x\}$ is disjoint from $N^2[\bigcup_{i=0}^{j-2}R_i(e,\len)]$. Combined with $x\notin N^2[\bigcup_{i=0}^{j-2}R_i(e,\len)]$, the claim follows.
\end{proof}

We can now prove the lemma deferred from \Cref{subsubsec:implementation-sub-phase-in-phase-j}: shifting $C$ changes the available colors only at $x$ and $y$, so $x$ can compute $A(z,\Shift(\chi,C))$ for every $z\in N[x]$ without communication.
\claimSetOfAvailableColorsIsComputable*

\begin{proof}
    We will first show that $A(z,\chi)=A(z,\Shift(\chi,C))$ unless $z\in\{x,y\}$.
If $z\notin V(C)$, $A(z,\chi)=A(z,\Shift(\chi,C))$ by \Cref{obsvn:available_colors_after_shift}. Hence assume $z\in V(C)$.
Note that since $x\notin N^2\big[\bigcup_{k=0}^{j-2}R_{k}(e,\len)\big]$ (by \Cref{fact_2}) and $ V(F_1+P_1+\dots F_{j-2}+P_{j-2}+F_{j-1})\subseteq N\big[\bigcup_{i=0}^{j-2}R_i(e,\len)\big]$ (by \Cref{obsvn:C_in_nbhd_of_R}), we have that $z\notin V(F_1+P_1+\dots F_{j-2}+P_{j-2}+F_{j-1})$. 
Hence $z\in V(P_{j-1})\setminus V(\estart(P_{j-1}))$, since $\estart(P_{j-1})=\eend(F_{j-1})$ by the definition of a $(j-1)$-step Vizing chain. If $z\in V(P_{j-1})\setminus \big(V(\estart P_{j-1})\cup \{x,y\}\big)$, we would have $A(z,\chi)=A(z,\Shift(\chi,C))$ by \Cref{obsvn:available_colors_after_shift}, and we are done. 

Hence, assume $z\in\{x,y\}$.
Since we have that $\Shift(\chi,F_1+P_1+\dots+F_{j-1})(xy)=\A$ and $P_{j-1}$ is either a $\BAtuple$-bichromatic path chain or an $\ABtuple$-bichromatic path chain of length greater than $2$ with $\eend(P_{j-1})=xy$ (by \Cref{fact_1}), we have that 
\begin{itemize}
    \item $A(y,\Shift(\chi,C))=A(y,\Shift(\chi,F_1+P_1+\dots+F_{j-1}))\cup\{\B\}$, and
    \item  $A(x,\Shift(\chi,C))=A(x,\Shift(\chi,F_1+P_1+\dots+F_{j-1}))\cup\{\A\}$.
\end{itemize}
We are done if we now show that $A(y,\chi)=A(y,\Shift(\chi,F_1+P_1+\dots+F_{j-1}))$ and $A(x,\chi)=A(x,\Shift(\chi,F_1+P_1+\dots+F_{j-1}))$. This is true since $x,y\notin V(F_1+P_1+\dots F_{j-2}+P_{j-2}+F_{j-1})$: indeed $x\notin N^2\big[\bigcup_{i=0}^{j-2}R_i(e,\len)\big]$ by \Cref{fact_2} and $y\in N(x)$, while $V(F_1+P_1+\dots F_{j-2}+P_{j-2}+F_{j-1})\subseteq N\big[\bigcup_{i=0}^{j-2}R_i(e,\len)\big]$ by \Cref{obsvn:C_in_nbhd_of_R}.

\end{proof}

We next show that appending $F$ and $P$ to $C$ again yields a Vizing chain.
\begin{lemma}\label{claim:C+F+P_is_j_step_Vizing_chain}
If $x$ finished the construction of $F+P$, then $F+P$ is a $1$-step Vizing chain of step length $\len$ on $xy$ under $\Shift(\chi,C)$ and $C+F+P$ is a $j$-step Vizing chain of step length $\len$ on $e$ under $\chi$.
\end{lemma}
\begin{proof}
We treat the three cases of the chain construction stage in turn.

\proofsubparagraph*{Case 1.}
Suppose $x$ finished the construction of $F+P$ by the \nameref{j_phase_construction_case_1}. Then $F$ is a fan under $\Shift(\chi,C)$ by construction and $P$ is a path chain of length $1$ under $\Shift(\chi,C+F)$. Hence $F+P$ is a $1$-step Vizing chain of step length $\len$ on $xy$ under $\Shift(\chi,C)$ and $C+F+P$ is a $j$-step Vizing chain of step length $\len$ on $e$ under $\chi$.

\proofsubparagraph*{Case 2.}
Suppose $x$ finished the construction of $F+P$ by the \nameref{j_phase_construction_case_2}, and adopt the variables of that case.

           Note that $F$ is an initial segment of $F'$, and hence a fan chain on $xy$ under $\Shift(\chi,C)$, since $F'$ is a fan chain on $xy$ under $\Shift(\chi,C)$. To show that $F+P$ is a $1$-step Vizing chain under $\Shift(\chi,C)$, we now have to show that $P$ is a $\GDtuple$-bichromatic path chain under $\Shift(\chi,C+F)$.

            We have that $\G\in A(x,\Shift(\chi,C))$ and $\D\in A(\vend(F),\Shift(\chi,C))$.
            By \Cref{obsvn:fan_available_color}, $\G\in A(x,\Shift(\chi,C+F))$ and $\D\in A(\vend(F),\Shift(\chi,C+F))$ as well. Since $F$ is a fan under $\Shift(\chi,C)$, $\Shift(\chi,C+F)(x\vend(F))=\neg$. 
            
            Since $\Shift(\chi,C+F)(\estart(P))=\neg$, $\G\in A(x,\Shift(\chi,C+F))$ and $\D\in A(\vend(F),\Shift(\chi,C+F))$,
             $P$ is a $\GDtuple$-bichromatic path chain under $\Shift(\chi,C+F)$ if $\bichrom(P)$ is $\GDtuple$-bichromatic path starting at $\vend(F)$ in $\Shift(\chi,C+F)$. We already have that $\bichrom(P)$ is a $\GDtuple$-bichromatic path under $\chi$ by \Cref{obsvn:is_bichromatic_path_phase_j_case_2}. Hence, we only have to show that $E(\bichrom(P)) \cap E(C+F)=\emptyset$.
            
          We will first show that $E(\bichrom(P))\cap E(C)=\emptyset$.
          We have that $E(\bichrom(P))\cap E(F_1+P_1+\dots F_{j-1})=\emptyset$ since $V(F+P)\cap N[\bigcup_{i=0}^{j-2}R_i(e,\len)]=\emptyset$ by \Cref{obsvn:vertex_disjointness_j_j-2} and $V(F_1+P_1+\dots+F_{j-1})\subseteq N[\bigcup_{i=0}^{j-2}R_i(e,\len)]$ by \Cref{obsvn:C_in_nbhd_of_R}. In particular, this also implies $E(\bichrom(P))\cap E(\estart(P_{j-1}))=\emptyset$.  
          If we now show that $E(\bichrom(P))\cap E(\bichrom(P_{j-1}))=\emptyset$, we will have that  $E(\bichrom(P))\cap E(C)=\emptyset$. Note that $\bichrom(P)$ is a $\GDtuple$-bichromatic path under $\chi$ by \Cref{obsvn:is_bichromatic_path_phase_j_case_2}. On the other hand $\bichrom(P_{j-1})$ is an $\ABtuple$-bichromatic path under $\chi$ by \Cref{fact_1}.
          Since $\{\G,\D\}\cap \ABSet=\emptyset$, we then have that $E(\bichrom(P))\cap E(\bichrom(P_{j-1}))=\emptyset$. Hence we have $E(\bichrom(P))\cap E(C)=\emptyset$. 
        
        It remains to show that $E(F)\cap E(\bichrom(P))=\emptyset$. To show this, it is sufficient to prove $x$ and $\vend(F)$ are not $\bichrom(P)$-related. By \Cref{obsvn:x_vend(F)_far_along_Q'_phase_j_case_2}, either $x\notin V(Q')$ or $x$ and $\vend(F)$ are at distance greater than $\len$ along $Q'$. Since $\bichrom(P)$ is a subpath of $Q'$ starting at $\vend(F)$ and of length at most $\len$, in either case $x$ and $\vend(F)$ are not $\bichrom(P)$-related. Hence $E(\bichrom(P))\cap E(F)=\emptyset$.

            Hence $F+P$ is a $1$-step Vizing
            chain of step-length $\len$ on $xy$ under $\Shift(\chi,C)$. 
            Since $C$ is
            already a $(j-1)$-step Vizing chain of step length $\len$ 
            on $e$ under $\chi$, $C+F+P$ is a $j$-step Vizing chain of step length
            $\len$ on $e$ under $\chi$.
            
\proofsubparagraph*{Case 3.}
Suppose $x$ finished the construction of $F+P$ by the \nameref{j_phase_construction_case_3}, and adopt the variables of that case.

$F=F'$ is a fan chain on $xy$ under $\Shift(\chi,C)$ since $F'$ is a fan chain on $xy$ under $\Shift(\chi,C)$. To show that $F+P$ is a $1$-step Vizing chain under $\Shift(\chi,C)$, we have to now show that $P$ is an $\ABtuple$-bichromatic path chain under $\Shift(\chi,C+F)$.

            We have that $\A\in A(x,\Shift(\chi,C))$ and $\B\in A(\vend(F),\Shift(\chi,C))$. By \Cref{obsvn:fan_available_color}, $\A\in A(x,\Shift(\chi,C+F))$ and $\B\in A(\vend(F),\Shift(\chi,C+F))$ as well. Also, $\Shift(\chi,C+F)(x\vend(F))=\neg$. Hence, to show that $P$ is an $\ABtuple$-bichromatic path chain under $\Shift(\chi,C+F)$, we just need to show that $\bichrom(P)$ is a $\vend(F)$-maximal $\ABtuple$-bichromatic path under $\Shift(\chi,C+F)$.

            We already have that $\bichrom(P)$ is a $\ABtuple$-bichromatic path under $\chi$ by \Cref{obsvn:is_bichromatic_path_phase_j_case_3}. Hence we only have to show that $E(\bichrom(P)) \cap E(C+F)=\emptyset$ to show that  $\bichrom(P)$ is a $\vend(F)$-maximal $\ABtuple$-bichromatic path under $\Shift(\chi,C+F)$.
            
            We have $E(\bichrom(P))\cap E(F_1+P_1+\dots+F_{j-1})=\emptyset$ by \Cref{obsvn:vertex_disjointness_j_j-2} since $V(F+P)\cap N^2[\bigcup_{i=0}^{j-2}R_i(e,\len)]=\emptyset$. In particular, we also have $E(\bichrom(P))\cap E(\estart(P_{j-1}))=\emptyset$. 

            We proceed to show that $E(\bichrom(P))\cap E(\bichrom(P_{j-1}))=\emptyset$. Note that $\bichrom(P)$ is a $\vend(F)$-maximal $\ABtuple$-bichromatic path starting at $\vend(F)$ under $\chi$ by \Cref{obsvn:is_bichromatic_path_phase_j_case_3}. 
            By \Cref{fact_1}, $\bichrom(P_{j-1})$ is a $\vend(F_{j-1})$-maximal $\ABtuple$-bichromatic path under $\chi$ whose endpoints are $\vend(F_{j-1})$ and $x$. Suppose, for contradiction, that $E(\bichrom(P))\cap E(\bichrom(P_{j-1}))\neq\emptyset$. Since two $\ABtuple$-bichromatic paths under $\chi$ that share an edge are subpaths of a common maximal $\ABtuple$-bichromatic path under $\chi$, the paths $\bichrom(P)$ and $\bichrom(P_{j-1})$ are subpaths of a common maximal $\ABtuple$-bichromatic path $M$ under $\chi$. Moreover, $\vend(F)\neq\vend(F_{j-1})$, since $\vend(F)=\vend(F_{j-1})$ would imply $x\in N^2[\bigcup_{i=0}^{j-2}R_i(e,\len)]$, contradicting \Cref{fact_2}. As $\bichrom(P)$ and $\bichrom(P_{j-1})$ are maximal at the distinct vertices $\vend(F)$ and $\vend(F_{j-1})$ respectively, these two vertices are the two endpoints of $M$. Since $\bichrom(P_{j-1})$ is the subpath of $M$ from $\vend(F_{j-1})$ to $x$, and $\bichrom(P)$ is the subpath of $M$ starting at $\vend(F)$, any edge shared by $\bichrom(P)$ and $\bichrom(P_{j-1})$ forces $\bichrom(P)$ to reach $x$ along $M$, so $x\in V(\bichrom(P))$. This however implies $x$ and $\vend(F)$ are $\bichrom(P)$-related, contradicting \Cref{obsvn:is_bichromatic_path_phase_j_case_3}. (It would also have caused $x$ to abandon the construction of $F+P$, since $\tau_{\operatorname{notify}}$ would reach $x$ again.) Hence $E(\bichrom(P))\cap E(\bichrom(P_{j-1}))=\emptyset$.
            
            Hence we already have $E(\bichrom(P))\cap E(C)=\emptyset$. We will now show that $E(F)\cap E(\bichrom(P))=\emptyset$ or equivalently that $x$ and $\vend(F)$ are not $\bichrom(P)$-related. If $x$ and $\vend(F)$ are $\bichrom(P)$-related, $x$ would abandon the construction of $F+P$ since $\tau_{\operatorname{notify}}$ would reach back $x$. Hence $E(\bichrom(P))\cap E(F)=\emptyset$. 

            Hence $E(\bichrom(P))\cap E(C+F)=\emptyset$. This implies $\bichrom(P)$ is a $\ABtuple$-bichromatic path under $\Shift(\chi,C+F)$ as well since the colors of $\Shift(\chi,C+F)$ and $\chi$ differs only at edges in $C+F$. 
            
            Hence we have that $P=(x\vend(F))+\bichrom(P)$ is a path such that
            \begin{itemize}
                \item $\Shift(\chi,C+F)(x\vend(F))=\neg$,
                \item $\A\in A(x,\Shift(\chi,C+F))$ and $\B\in A(\vend(F),\Shift(\chi,C+F))$, 
                \item $\bichrom(P)$ is a $\ABtuple$ bichromatic path under $\Shift(\chi,C+F)$, and
                \item length of $\bichrom(P)$ is at most $\len$.
            \end{itemize}
            Hence $P$ is a $\ABtuple$-bichromatic path chain of length at most $\len+1$ on $(x\vend(F),\vend(F))$ under
            $\Shift(\chi,C+F)$. Hence $F+P$ is a $1$-step Vizing
            chain of step-length $\len$ on $xy$ under $\Shift(\chi,C)$. Since $C$ is
            already a $(j-1)$-step Vizing chain of step length $\len$ 
            on $e$ under $\chi$, $C+F+P$ is a $j$-step Vizing chain of step length
            $\len$ on $e$ under $\chi$.
            
\end{proof}

The chain moreover remains self-avoiding, which is the lemma deferred from \Cref{subsubsec:implementation-sub-phase-in-phase-j}.
\IsNonOverlapping*

\begin{proof}
 First of all, note that $C+F+P|d_P(x,v)$ is indeed a $j$-step Vizing chain of step length $\len$ on $e$ under $\chi$ since $C+F+P$ is a $j$-step Vizing chain of step length $\len$ on $e$ under $\chi$ by \Cref{claim:C+F+P_is_j_step_Vizing_chain}.
 Moreover, if $C+F+P$ is self-avoiding, $C+F+P|d_P(x,v)$ is also self-avoiding. Hence we only have to show that $C+F+P$ is self-avoiding.

Since $C$ is already a $(j-1)$-step self-avoiding Vizing chain  of step $\len$ on $e$ under $\chi$, we have that
\begin{itemize}
    \item for all $1\leq k\leq j-2$, we have $E(P_k)\cap E(F_{k+1}+P_{k+1})=\{\eend(P_{k})\}$,
    \item for all $m, k$ satisfying $3\leq m+2\leq k\leq j-1$, we have $V(F_m+P_m)\cap V(F_k+P_k)=\emptyset$,
    \item for all $1\leq k\leq j-2$, we have $V(F_k)\cap V(F_{k+1})=\emptyset$, and
    \item for all $1\leq k \leq j-1$, we have that $\cen(F_k)$ and $\vend(F_k)$ are not $\bichrom(P_k)$-related. 
\end{itemize}
 
Hence, to show that $C+F+P$ is self-avoiding, we only need to show the following statements:
\begin{enumerate}[(S1)]
     \item $E(P_{j-1})\cap E(F+P)=\{\eend(P_{j-1})\}$,
     \item For $1\leq k\leq j-2$, $V(F_k+P_k)\cap V(F+P)=\emptyset$,
     \item $V(F)\cap V(F_{j-1})=\emptyset$, and
     \item $\cen(F)=x$ and $\vend(F)$ are not $\bichrom(P)$-related.
\end{enumerate}

We prove these statements below. We prove (S3) before (S1), since the proof of (S1) uses it.
\proofsubparagraph*{Proof of (S2).}
We show here that, for $1\leq k\leq j-2$, $V(F_k+P_k)\cap V(F+P)=\emptyset$. Suppose $w\in V(F+P)\cap V(F_k+P_k)$ for some $1\leq k\leq j-2$.  We have $x\notin N^2[\bigcup_{i=0}^{j-2}R_i(e,\len)]$ by \Cref{fact_2} and $V(F_1+P_1+\dots+F_{j-1})\subseteq N[\bigcup_{i=0}^{j-2}R_i(e,\len)]$ by \Cref{obsvn:C_in_nbhd_of_R}. Hence $w\notin N[x]$ which implies $w\notin V(F)$. 
           Since $x$ finishes the construction of $F+P$, $V(P)\cap N^2[\bigcup_{i=0}^{j-2}R_i(e,\len)]= \emptyset$ by \Cref{obsvn:vertex_disjointness_j_j-2}. Hence $w\notin V(P)$. Hence $w\notin V(F+P)$ a contradiction.
           Hence for $1\leq k\leq j-2$, $V(F+P)\cap V(F_k+P_k)=\emptyset$.

\proofsubparagraph*{Proof of (S3).}
We show here that $V(F)\cap V(F_{j-1})=\emptyset$. If $z\in V(F)\cap V(F_{j-1})$ for some $z$, we have $x\in N^2[\cen(F_{j-1})]$ since $\cen(F)=x$. Note that $\cen(F_{j-1})\in\bigcup_{i=0}^{j-2}R_i(e,\len)$: for $j\geq 3$, the vertex $\cen(F_{j-1})$ constructed the candidate chain $F_{j-1}^*+P_{j-1}^*$ for $e$ in construction phase $j-1$ by \Cref{fact_1}, so $\cen(F_{j-1})\in S_{j-2}(e,\len)\subseteq R_{j-2}(e,\len)$ by \Cref{obsvn:S_j-1_construct_candidates}, and for $j=2$, we have $\cen(F_1)\in S_0(e,\len)=R_0(e,\len)$ by \Cref{obsvn:S_0_construct_candidates}. Hence $x\in N^2[\bigcup_{i=0}^{j-2}R_i(e,\len)]$, contradicting \Cref{fact_2}. Hence $V(F)\cap V(F_{j-1})=\emptyset$.
           
\proofsubparagraph*{Proof of (S1).}
We show here that $E(F+P)\cap E(P_{j-1})=\{\eend(P_{j-1})\}$.  Let $F$ be written as $(xz_1=xy,\dots, xz_s)$. Let $xy=\eend(P_{j-1})$.  
           
           We first show that $E(F)\cap E(P_{j-1})=\{xy\}$. Since $F$ is a fan chain constructed on $xy$ under $\Shift(\chi,C)$, $xy\in E(F)$. Hence $xy\in E(F)\cap E(P_{j-1})$. Suppose the reverse inclusion is not true. This would imply that for some $2\leq k\leq s$, $x z_{k}\in E(P_{j-1})\setminus \{xy\}$. Note that $xz_k\neq\estart(P_{j-1})$, since $\estart(P_{j-1})=\eend(F_{j-1})$ and $V(F)\cap V(F_{j-1})=\emptyset$ as we already saw. Hence $xz_k\in E(P_{j-1})\setminus \{\estart(P_{j-1}),xy\}$. Since $P_{j-1}=\estart(P_{j-1})+\bichrom(P_{j-1})$, this implies $xz_k\in E(\bichrom(P_{j-1}))\setminus \{xy\}$. Hence $xz_k\in E(\bichrom(P_{j-1}))$ and $xy\in E(\bichrom(P_{j-1}))$ with $z_k\neq y$. This implies that $x$ is not an endpoint of $\bichrom(P_{j-1})$ because $\bichrom(P_{j-1})$ is a bichromatic path under $\chi$ by \Cref{fact_1}. This contradicts the fact that $x$ is  indeed an endpoint of $\bichrom(P_{j-1})$. Hence we have that $E(F)\cap E(P_{j-1})=\{xy\}$.
           
            It remains to show that show $E(P)\cap E(P_{j-1})\subseteq\{xy\}$. Consider the following exhaustive cases.
        \begin{itemize}
        
        \item Assume $x$ finished the construction of $F+P$  by the \nameref{j_phase_construction_case_1}. Then, length of $P$ is $1$ and hence $E(P)=\{xz_p\}\subseteq E(F)$. Since we saw just above that $E(F)\cap E(P_{j-1})=\{xy\}$, we have that  $E(P)\cap E(P_{j-1})\subseteq\{xy\}$ as well.
        
        \item Assume $x$ finished the construction of $F+P$  by the \nameref{j_phase_construction_case_2}. Then $P$ is a $\GDtuple$-bichromatic path chain for $\{\G,\D\}\cap\ABSet=\emptyset$. We have $V(P)\cap N^2[\bigcup_{i=0}^{j-2}R_i(e,\len)]=\emptyset$ by \Cref{obsvn:vertex_disjointness_j_j-2}. Since $V(F_1+P_1+\dots+F_{j-1})\subseteq N[\bigcup_{i=0}^{j-2}R_i(e,\len)]$ by \Cref{obsvn:C_in_nbhd_of_R}, we have $V(P)\cap V(F_1+P_1+\dots+F_{j-1})=\emptyset$. This in particular implies $E(P)\cap E(\estart(P_{j-1}))=\emptyset$ since $\estart(P_{j-1})=\eend(F_{j-1})$.
                 
        Note that $\bichrom(P)$ is a $\GDtuple$-bichromatic path under $\chi$ by \Cref{obsvn:is_bichromatic_path_phase_j_case_2}. Since $\bichrom(P_{j-1})$ is an $\ABtuple$-bichromatic path under $\chi$ by \Cref{fact_1}, $E(\bichrom(P))\cap E(\bichrom(P_{j-1}))=\emptyset$.  
                 
        Combining $E(P)\cap E(\estart(P_{j-1}))=\emptyset$ and $E(\bichrom(P))\cap E(\bichrom(P_{j-1}))=\emptyset$ gives $E(\bichrom(P))\cap E(P_{j-1})=\emptyset$. Since $E(P)\subseteq E(F)\cup E(\bichrom(P))$, we have that $E(P)\cap E(P_{j-1}) \subseteq  E(F)\cap E(P_{j-1}) =\{\eend(P_{j-1})\}=\{xy\}$. 
                
      \item Assume $x$ finished the construction of $F+P$  by the \nameref{j_phase_construction_case_3}. 
      
      We have that $\bichrom(P_{j-1})$ is a $\vend(F_{j-1})$-maximal $\ABtuple$-bichromatic path under $\chi$ by \Cref{fact_1}, and its endpoints are $\vend(F_{j-1})$ and $x$. We also have that $\bichrom(P)$ is a $\vend(F)$-maximal $\ABtuple$-bichromatic path under $\chi$ by \Cref{obsvn:is_bichromatic_path_phase_j_case_3}. 
      Suppose, for contradiction, that $E(\bichrom(P))\cap E(\bichrom(P_{j-1}))\neq\emptyset$. Since two $\ABtuple$-bichromatic paths under $\chi$ that share an edge are subpaths of a common maximal $\ABtuple$-bichromatic path under $\chi$, the paths $\bichrom(P)$ and $\bichrom(P_{j-1})$ are subpaths of a common maximal $\ABtuple$-bichromatic path $M$ under $\chi$. Moreover, $\vend(F)\neq \vend(F_{j-1})$ since $\vend(F)=\vend(F_{j-1})$ would imply  $x\in N^2[\bigcup_{i=0}^{j-2}R_i(e,\len)]$, contradicting \Cref{fact_2}. As $\bichrom(P)$ and $\bichrom(P_{j-1})$ are maximal at the distinct vertices $\vend(F)$ and $\vend(F_{j-1})$ respectively, these two vertices are the two endpoints of $M$. Since $\bichrom(P_{j-1})$ is the subpath of $M$ from $\vend(F_{j-1})$ to $x$, and $\bichrom(P)$ is the subpath of $M$ starting at $\vend(F)$, any edge shared by $\bichrom(P)$ and $\bichrom(P_{j-1})$ forces $\bichrom(P)$ to reach $x$ along $M$, so $x\in V(\bichrom(P))$. This contradicts the fact that $x$ and $\vend(F)$ are not $\bichrom(P)$-related by \Cref{obsvn:is_bichromatic_path_phase_j_case_3}. Hence $E(\bichrom(P))\cap E(\bichrom(P_{j-1}))=\emptyset$.

      Since $x$ finished the construction of $F+P$ by the \nameref{j_phase_construction_case_3}, we have $V(P)\cap N^2[\bigcup_{i=0}^{j-2}R_i(e,\len)]=\emptyset$ by \Cref{obsvn:vertex_disjointness_j_j-2}. Since $V(F_1+P_1+\dots+F_{j-1})\subseteq N[\bigcup_{i=0}^{j-2}R_i(e,\len)]$ by \Cref{obsvn:C_in_nbhd_of_R}, we have $E(P)\cap E(\estart(P_{j-1}))=\emptyset$ since $\estart(P_{j-1})=\eend(F_{j-1})$. Combining this with $E(\bichrom(P))\cap E(\bichrom(P_{j-1}))=\emptyset$ gives $E(\bichrom(P))\cap E(P_{j-1})=\emptyset$. Since $E(P)\subseteq E(F)\cup E(\bichrom(P))$, we have that $E(P)\cap E(P_{j-1})\subseteq E(F)\cap E(P_{j-1})=\{\eend(P_{j-1})\}=\{xy\}$.  
    \end{itemize}
    Hence, we have that $E(F+P)\cap E(P_{j-1})=\{xy\}$.
           
\proofsubparagraph*{Proof of (S4).}
We show here that $\cen(F)=x$ and $\vend(F)$ are not $\bichrom(P)$-related.
    If $x$ finished the construction of $F+P$  by the \nameref{j_phase_construction_case_1}, $P$ has length $1$. Then $x$ and $\vend(F)$ are not $\bichrom(P)$-related since $\bichrom(P)=\emptyset$. If $x$ finished the construction of $F+P$  by the \nameref{j_phase_construction_case_2}, $\vend(F)$ and $x$ are not $\bichrom(P)$-related by \Cref{obsvn:is_bichromatic_path_phase_j_case_2}.  If $x$ finished the construction of $F+P$  by the \nameref{j_phase_construction_case_3}, $\vend(F)$ and $x$ are not $\bichrom(P)$-related by \Cref{obsvn:is_bichromatic_path_phase_j_case_3}.
This proves all the required properties for $C+F+P$.
\end{proof}

As in construction phase $1$, a candidate chain is hopeful precisely when its path chain stopped at the step length rather than at an endpoint of $Q$ or at a vertex near the sets of the earlier phases. The next observation makes this quantitative.
\begin{observation} \label{obsvn:hopeful_is_long_in_phase_j}
    Let $F+P$ be a hopeful candidate chain constructed in construction phase $j$. Then length of $P$ is $\len+1$.
\end{observation}
\begin{proof}
    A hopeful candidate chain receives $\tau_{\operatorname{reply}}$ with $\suc=0$, which (by the confirmation substages of the \nameref{j_phase_construction_case_2} and the \nameref{j_phase_construction_case_3}) occurs only when $\vend(P)$ is neither an endpoint of $Q$ nor in $N^2[\bigcup_{i=0}^{j-2}R_i(e,\len)]$. By the notification substage, this means $\tau_{\operatorname{notify}}$ stopped because it traversed $\len+1$ edges along $Q$. Hence $P=Q|(\len+1)$ has length $\len+1$.
\end{proof}

We again record the shape of the bichromatic part of a candidate chain.
\begin{observation}
 \label{claim:j_phase_bichrom(P)_is_bichromatic}
    Let $F+P$ be a candidate chain constructed in construction phase $j$. Then either $\bichrom(P)=\emptyset$ or $\bichrom(P)$ is a $\vend(F)$-maximal bichromatic path under $\chi$.   
\end{observation} 
\begin{proof}
    Let $F+P$ be a candidate chain constructed for $e$ in construction phase $j$.
    By \Cref{obsvn:S_j-1_construct_candidates}, there exists an $x\in S_{j-1}(e,\len)$ that constructs $F+P$.

    Suppose $x$ constructed $F+P$ by the \nameref{j_phase_construction_case_1}. Then $\bichrom(P)=\emptyset$ by construction.

    Suppose $x$ constructed $F+P$ by the \nameref{j_phase_construction_case_2} or the \nameref{j_phase_construction_case_3}. Then $\bichrom(P)$ is a $\vend(F)$-maximal bichromatic path under $\chi$ by \Cref{obsvn:is_bichromatic_path_phase_j_case_2} and \Cref{obsvn:is_bichromatic_path_phase_j_case_3}. 
\end{proof}

Finally, we isolate the single case in which $\cen(F)$ and $\vend(F)$ can be $\bichrom(P)$-related.
\begin{observation} \label{obsvn:x_vend(F)_connected_in_case_3}
    Let $F+P$ be any candidate chain constructed in construction phase $j$. Then $\cen(F)$ and $\vend(F)$ are $\bichrom(P)$-related only if $F+P$ was constructed by the \nameref{j_phase_construction_case_3}.
\end{observation}
\begin{proof}
    If $F+P$ was constructed by the \nameref{j_phase_construction_case_1}, then $\bichrom(P)=\emptyset$, so $\cen(F)$ and $\vend(F)$ are not $\bichrom(P)$-related. If $F+P$ was constructed by the \nameref{j_phase_construction_case_2}, then $\cen(F)$ and $\vend(F)$ are not $\bichrom(P)$-related by \Cref{obsvn:is_bichromatic_path_phase_j_case_2}. Hence $\cen(F)$ and $\vend(F)$ can be $\bichrom(P)$-related only if $F+P$ was constructed by the \nameref{j_phase_construction_case_3}.
\end{proof}

We close this subsection with the phase-$j$ counterpart of \Cref{lemma:successful_is_terminating_phase_1}: appending a successful candidate chain to $C$ yields a terminating Vizing chain.

\SuccessfulIsTerminatingPhaseJ*
\begin{proof}
For ease of writing, we use the same variables as in the \nameref{j_phase_construction_case_1}, the \nameref{j_phase_construction_case_2} and the \nameref{j_phase_construction_case_3}; in addition, for $1\leq i\leq p-1$ we write $\eta_i$ for the representative available color at $z_i$ in $F'$, and we let $\{\sigma,\sigma'\}$ denote $\GDtuple$ in the \nameref{j_phase_construction_case_2} and $\ABSet$ in the \nameref{j_phase_construction_case_3}, so that in either case $Q'$ is a maximal $\{\sigma,\sigma'\}$-bichromatic path under $\chi$ with $\vend(F)$ as an endpoint.

We first observe that $\vend(P)\in V(P)\setminus\{x\}$. In the \nameref{j_phase_construction_case_1} this holds because $\vend(P)=z_p\neq x$. In the \nameref{j_phase_construction_case_2} it holds by \Cref{obsvn:x_vend(F)_far_along_Q'_phase_j_case_2}, and in the \nameref{j_phase_construction_case_3} the propagation of $\tau_{\operatorname{notify}}$ stops at $x$, so $\vend(P)=x$ would give $\suc=-1$, contradicting that $F+P$ is successful. Taking $v=\vend(P)$ in \Cref{lemma:C+F+P_self_avoiding} therefore shows that $C+F+P$ is a self-avoiding $j$-step Vizing chain of step length $\len$ on $e$ under $\chi$. By \Cref{obsvn:maximal_implies_terminating}, it suffices to prove that $P$ is maximal under $\Shift(\chi,C+F)$.

If $x$ constructed $F+P$ by the \nameref{j_phase_construction_case_1}, then $P=(xz_p)$ and $\eta\in A(x,\Shift(\chi,C+F))\cap A(z_p,\Shift(\chi,C+F))$, so $P$ is maximal under $\Shift(\chi,C+F)$. Assume therefore that $x$ constructed $F+P$ by the \nameref{j_phase_construction_case_2} or the \nameref{j_phase_construction_case_3}.

Since $F+P$ is a successful candidate chain, $\tau_{\operatorname{reply}}$ carries $\suc=1$, so $\vend(P)$ is an endpoint of $Q$ and $\vend(P)\notin N^2[\bigcup_{i=0}^{j-2}R_i(e,\len)]$. Also $\vend(P)\neq\vend(F)$: exactly one of $\sigma$ and $\sigma'$ is available at $\vend(F)$ under $\Shift(\chi,C)$, so $\vend(F)$ is incident to an edge colored with the other one, and hence $Q'$ has length at least $1$. Thus $\vend(P)$ is the endpoint of $Q'$ other than $\vend(F)$, and there is a color $\zeta\in A(\vend(P),\chi)\cap\{\sigma,\sigma'\}$.

We first show $\zeta\in A(\vend(P),\Shift(\chi,C))$. For $1\leq k\leq j-1$ we have $\cen(F_k)\in S_{k-1}(e,\len)\subseteq R_{k-1}(e,\len)$ and $V(F_k)\subseteq N[\cen(F_k)]$, so $\bigcup_{k=1}^{j-1}V(F_k)\subseteq N^2[\bigcup_{i=0}^{j-2}R_i(e,\len)]$ and therefore $\vend(P)\notin\bigcup_{k=1}^{j-1}V(F_k)$. Moreover $V(\eend(P_{j-1}))=\{x,y\}$ and $\vend(P)\neq x$. If $\vend(P)\neq y$, then $A(\vend(P),\chi)=A(\vend(P),\Shift(\chi,C))$ by \Cref{obsvn:available_colors_after_shift}; and if $\vend(P)=y$, then $A(y,\chi)\subseteq A(y,\Shift(\chi,C))$ by \Cref{lemma:last_node_nbrhd_available_colors}. In both cases $\zeta\in A(\vend(P),\Shift(\chi,C))$.

It remains to show $\zeta\in A(\vend(P),\Shift(\chi,C+F))$. By \Cref{obsvn:fan_available_color}, applied to $F$ under $\Shift(\chi,C)$, we only have to show that $\zeta\neq\eta_i$ whenever $\vend(P)=z_i$ for some $z_i\in V(F)\setminus\{x,\vend(F)\}$. By the construction of $F'$ in \Cref{lemma:fan-addition-phase-j} we have $\eta_i\notin\ABSet$ for every $i$, which settles the \nameref{j_phase_construction_case_3}. In the \nameref{j_phase_construction_case_2} we have $\eta_i\neq\G$ for every $i$, since $\eta_i\in A(z_i,\Shift(\chi,C))$, $\G\in A(x,\Shift(\chi,C))$ and $A(x,\Shift(\chi,C))\cap A(z_i,\Shift(\chi,C))=\emptyset$; and the colors $\eta_1,\dots,\eta_{p-1}$ are pairwise distinct with $\eta_q=\D$, so $\eta_i=\D$ forces $i=q$. It therefore remains to rule out $\vend(P)=z_q$ with $\vend(F)=z_p$. In that case $x$ received no $\tau_{\operatorname{return}}$, so $\tau_{\operatorname{search}}$ stopped neither at $x$ nor at $z_q$; as $\tau_{\operatorname{search}}$ and $\tau_{\operatorname{notify}}$ traverse the same trail $(xz_p)+K'=(xz_p)+Q'$ and both stop after at most $\len+1$ hops, and $\tau_{\operatorname{notify}}$ reached the endpoint $\vend(P)$ of $Q'$, so did $\tau_{\operatorname{search}}$, and it stopped there. Since it did not stop at $z_q$, we conclude $\vend(P)\neq z_q$.
\end{proof}

\subsection{Properties of the Set Construction Stage in Construction Phase \texorpdfstring{$j$}{j} for \texorpdfstring{$2\leq j\leq \numphases$}{2 ≤ j ≤ \numphases}}\label{subsec:set-properties-phase-j}

We now record properties of the sets updated in the set construction stage of construction phase $j$ for $2\leq j\leq \numphases$. We begin with the phase-$j$ counterpart of \Cref{obsvn:M_1_no_duplicates}: every fiber of $M_j(\cdot)$ holds at most one tuple.

\begin{observation} \label{obsvn:M_j_no_duplicates}
    Let $2\leq j\leq \numphases$ and $x\in V(G)$. Then $|M_j(x;e)|\leq 1$ for
    every edge $e\in U$ after construction phase $j$.
\end{observation}
\begin{proof}
    By the deduplication step at the end of construction phase $j$, vertex $x$
    replaces every nonempty fiber $M_j(x;e)$ by a singleton subset of it.
\end{proof}

As in construction phase $1$, a vertex joins $R_j(e,\len)$ exactly when it lies on the path chain of a finished candidate chain for $e$, other than the vertex that built that chain.
\begin{observation} \label{obsvn:in_R_j_iff_in_P}
   Let $v\in V(G)$. Then
   $v$ is added to $R_{j}(e,\len)$ in construction phase $j$ if and only if some $x\in S_{j-1}(e,\len)$ finished the construction of a candidate chain $F+P$ for $e$ in construction phase $j$ such that $v\in V(P)\setminus\{x\}$.
\end{observation}
\begin{proof}
    By the set construction stage in construction phase $j$, $v$ is added to $R_j(e,\len)$ if and only if $v\in V(P)\setminus\{x\}$ for some candidate chain $F+P$ for $e$ that was finished by some vertex $x$. By \Cref{obsvn:S_j-1_construct_candidates}, $x\in S_{j-1}(e,\len)$.
\end{proof}

A vertex can moreover join $R_j(\cdot,\len)$ for only $O(\Delta^2)$ uncolored edges within a single subphase, exactly as in construction phase $1$.
 \begin{lemma} \label{lemma:bounded_addition_Phase_j}
Fix $c\in [\Delta^2+1]$ and $r\in [2\Delta^2]$. Let $U^j_v$ be the set of edges $e\in U$ such that $v$ adds itself to $R_j(e,\len)$ in subphase $(c,r)$ of construction phase $j$. Then $|U^j_v|\leq O(\Delta^2)$.
\end{lemma}
\begin{proof}
Note that a vertex $v$ adds itself
to $R_j(e,\len)$ for an edge $e\in U$ only if there exists a vertex $x\in S_{j-1}(e,\len)$ that finishes the
construction of a candidate chain $F+P$ for $e$ such that $v\in V(P)\setminus\{x\}$ by \Cref{obsvn:S_j-1_construct_candidates}. Since $\lambda_{x,j-1}$ is an injective function for any vertex $x$, $x$ can construct at most one candidate chain in subphase $(c,r)$. Hence, we only need to show that there are at most $O(\Delta^2)$ different vertices that can construct some candidate chain $F+P$ in subphase $(c,r)$ so that $v\in V(P)$.  
        
We start by counting how many vertices in the neighborhood of $v$ can construct $F+P$ such that $v\in V(F+P)$. Suppose $v\in N(x)$. Since $\psi$ is a proper coloring of $G^2$, at most one vertex in $N[v]$ constructs a candidate chain $F+P$ in subphase $(c,r)$ so that $v$ may potentially add itself to $R_j(e,\len)$ for some $e$.  

We now count how many vertices outside the neighborhood of $v$ can construct $F+P$ such that $v\in V(F+P)$. If $v\notin N(x)$, then $v\in V(P)\setminus\{x,\vend(F)\}$ and hence $P$ is of length greater than $1$. This implies $x$ finished constructing $F+P$ by the \nameref{j_phase_construction_case_2} or the \nameref{j_phase_construction_case_3}. Then
\begin{itemize}
    \item $\vend(F)$ is an endpoint of a maximal bichromatic path under $\chi$ that contains $v$ (by \Cref{obsvn:is_bichromatic_path_phase_j_case_2} and \Cref{obsvn:is_bichromatic_path_phase_j_case_3}), and
    \item $x$ is the unique neighbor of $\vend(F)$ with $\psi(x)=c$.
\end{itemize}
Since there are $O(\Delta^2)$ different maximal bichromatic paths under $\chi$ that contain $v$, $\vend(F)$ has to be one among one of the endpoints of these $O(\Delta^2)$ maximal bichromatic paths and hence $x$ has to be the unique neighbor colored $c$ of one of these endpoints. Hence there are at most $O(\Delta^2)$ vertices that construct a candidate chain so that $v$ adds itself to $R_j(e,\len)$. 
\end{proof}

The vertex that constructs a candidate chain again records the outcome of its chain instead of joining $R_j(e,\len)$. There are now three possible outcomes rather than two, since a candidate chain can also be unsuccessful, and it joins exactly one of the corresponding sets.
\begin{observation} \label{obsvn:x_in_D_or_Q_or_B}
   Let $x\in S_{j-1}(e,\len)$. Then $x$ is added to exactly one of $D_{j-1}(e,\len)$, $Q_{j-1}(e,\len)$ or $B_{j-1}(e,\len)$ in construction phase $j$.
\end{observation}
\begin{proof}
     Since $x\in S_{j-1}(e,\len)$, $x$ constructs a candidate chain for $e$ in construction phase $j$ by \Cref{obsvn:S_j-1_construct_candidates}, and by the same observation it constructs only this one candidate chain for $e$.
      By the set construction stage for $x$ in construction phase $j$, $x$ is added to exactly one of $D_{j-1}(e,\len)$, $Q_{j-1}(e,\len)$ or $B_{j-1}(e,\len)$, according to whether this candidate chain is successful, hopeful or unsuccessful, respectively.
\end{proof}

The remaining two observations concern the vertices that carry the construction into the next phase. Such a vertex lies in $S_j(e,\len)$ and holds a tuple in $M_j(\cdot)$; these two conditions are again equivalent, and the tuple is then unique.
\begin{observation} \label{obsvn:S_j_M_j_relation}
    Suppose construction phase $j$ is over. Consider any vertex $v\in V(G)$. Then 
    $$v\in S_j(e,\len) \iff (e,x,\A',\B',vw)\in M_j(v) $$
    for unique $x\in S_{j-1}(e,\len)$, $\A',\B'\in [\Delta+1]$ and $w\in N(v)$.
\end{observation}
\begin{proof}
    Suppose $(e,x,\A',\B',vw)\in M_j(v)$. Then $v\in S_j(e,\len)$ from the description of the set construction stage.

    We will prove the other direction now. Suppose $v\in S_j(e,\len)$. We just have to show that there exists an element of the form $(e,x,\A',\B',vw)$ that was added to $M_j(v)$; it will be unique since $M_j(v)$ has no duplicates of $e$ after construction phase $j$. 
    
    Note that $v$ is added to $S_j(e,\len)$ only if some $x$ constructed a hopeful candidate chain $F+P$ such that $v\in V(P)\setminus N^2[x]$ and $\bichrom(P)$ is some $(\A',\B')$-bichromatic path under $\chi$. Let $w$ be the predecessor of $v$ along $P$. Note that $vw\in E(\bichrom(P))$ since $v\in V(P)\setminus N^2[x]$. Also, either $\chi(vw)=\A'$ or $\chi(vw)=\B'$, since $\bichrom(P)$ is $(\A',\B')$-bichromatic under $\chi$. Depending on which of the two holds, the tuple $(e,x,\A',\B',vw)$ or $(e,x,\B',\A',vw)$ is added to $M_j(v)$ respectively. Hence done.
 \end{proof}

A tuple in $M_j(v)$ in turn determines the chain that produced it, together with the properties of that chain that the later phases rely on.
\begin{observation} \label{obsvn:prprty_of_j_step_chain}
       Let $(e,x,\A',\B',vw)$ be added to $M_j(v)$ in subphase $(c,r)$ in construction phase $j$. Then $x\in S_{j-1}(e,\len)$ finished the construction of a unique candidate chain $F+P$ for $e$ in construction phase $j$, where $C=F_1+P_1+\dots+F_{j-1}+P_{j-1}$ is the $(j-1)$-step Vizing chain on $e$ under $\chi$ that \Cref{fact_1} attaches to the tuple in $M_{j-1}(x)$ served by $x$, such that 
       \begin{enumerate}
           \item $C+F+P$ is a self-avoiding $j$-step Vizing chain of step length $\len$ on $e$ under $\chi$, \label[instance]{prprty:j_step_self_avoiding}
           
           \item  for $1\leq k \leq j-1$, $F_k=F_k'$ and $P_k$ is a subpath of $P_k'$ where $F_k'+P_k'$ is a candidate chain constructed by $\cen(F_k)$ for $e$ during construction phase $k$, \label[instance]{prprty:j_step_subpath}

           \item $V(\bichrom(P))\subseteq R_j(e,\len)$ and $V(\bichrom(P_k))\subseteq R_k(e,\len)$ for $1\leq k\leq j-1$, \label[instance]{prprty:j_step_R_containment}

           \item $\bichrom(P)$ is a $\vend(F)$-maximal bichromatic path under $\chi$, \label[instance]{prprty:j_step_maximal}

            \item $P$ is either an $(\A',\B')$-bichromatic path chain or a $(\B',\A')$-bichromatic path chain under $\Shift(\chi,C+F)$, and \label[instance]{prprty:j_step_bichromatic_chain}
          
           \item $vw = \eend(P|d_P(x,v))$, $\Shift(\chi,C+F)(vw)=\A'$. \label[instance]{prprty:j_step_eend_color}
       \end{enumerate}
   \end{observation}
   \begin{proof}
     By the description of the set construction stage in construction phase $j$, the tuple $(e,x,\A',\B',vw)$ is added to $M_j(v)$ only when vertex $x$ finishes the construction of a candidate chain, say $F+P$, under $\Shift(\chi,C)$, where $C=F_1+P_1+\dots+F_{j-1}+P_{j-1}$ is the $(j-1)$-step Vizing chain of step length $\len$ on $e$ under $\chi$ given by \Cref{fact_1} for the tuple in $M_{j-1}(x)$ served by $x$. The set construction stage moreover gives that
      \begin{itemize}
          \item $v\in V(P)\setminus N^2[x]$,
          \item $vw = \eend(P|d_P(x,v))$ and $\chi(vw)=\A'$, and
          \item $V(\bichrom(P))\subseteq R_j(e,\len)$,
      \end{itemize}
      while \Cref{fact_1} gives that
      \begin{itemize}
          \item for $1\leq k \leq j-1$, $F_k=F_k'$ and $P_k$ is a subpath of $P_k'$ where $F_k'+P_k'$ is a candidate chain constructed by $\cen(F_k)$ for $e$ during construction phase $k$, and
          \item $V(\bichrom(P_k))\subseteq R_k(e,\len)$ for $1\leq k\leq j-1$.
      \end{itemize}
    The candidate chain $F+P$ is unique since $x$ constructs only one candidate chain for $e$ in construction phase $j$ (as $M_{j-1}(x)$ has no duplicates of $e$). By \Cref{obsvn:S_j-1_construct_candidates}, $x\in S_{j-1}(e,\len)$. This establishes \Cref{prprty:j_step_subpath,prprty:j_step_R_containment} and the first equality of \Cref{prprty:j_step_eend_color}.
    
    Since $v\in V(P)\setminus N^2[x]$, we have $\bichrom(P)\neq \emptyset$, and hence $x$ constructed $F+P$ by the \nameref{j_phase_construction_case_2} or the \nameref{j_phase_construction_case_3} in construction phase $j$. Then $C+F+P$ is a self-avoiding $j$-step Vizing chain of step length $\len$ on $e$ under $\chi$, by \Cref{lemma:C+F+P_self_avoiding} applied to the vertex $\vend(P)$, establishing \Cref{prprty:j_step_self_avoiding}. By \Cref{obsvn:is_bichromatic_path_phase_j_case_2} and \Cref{obsvn:is_bichromatic_path_phase_j_case_3}, $\bichrom(P)$ is a $\vend(F)$-maximal bichromatic path under $\chi$, establishing \Cref{prprty:j_step_maximal}. 
    
    It remains to establish \Cref{prprty:j_step_bichromatic_chain}. Since $F+P$ is a $1$-step Vizing chain, $P$ is a path chain under $\Shift(\chi,C+F)$. Also, $\bichrom(P)$ is $(\A',\B')$-bichromatic under $\chi$, and this remains true under $\Shift(\chi,C+F)$ because $E(C+F)\cap E(\bichrom(P))=\emptyset$ (which holds since $C+F+P$ is a self-avoiding $j$-step Vizing chain). Hence $\bichrom(P)$ is an $(\A',\B')$-bichromatic path under $\Shift(\chi,C+F)$, and therefore $P$ is either an $(\A',\B')$- or a $(\B',\A')$-bichromatic path chain under $\Shift(\chi,C+F)$.

Finally, $E(C+F)\cap E(\bichrom(P))=\emptyset$ also gives $\Shift(\chi,C+F)(vw)=\chi(vw)=\A'$, since $vw\in E(\bichrom(P))$ by the set construction stage. This establishes the second equality of \Cref{prprty:j_step_eend_color}.

Hence, all the required properties are true and we are done.
\end{proof}

\subsection{Runtime of Construction Phase \texorpdfstring{$j$}{j}}\label{sec:runtime-phase-j}

We are now ready to deduce the runtime of a subphase in construction phase $j$ for $2\leq j \leq T$.

\begin{lemma} \label{lemma:runtime_phase_j_subphase}
   For $j$ with $2\leq j \leq \numphases$, each subphase in construction phase $j$ runs in $O(\Delta^3\len)$ rounds in CONGEST.
\end{lemma}
\begin{proof}
       Consider a subphase $(c,r)$ in construction phase $j$.
       
       Consider the preparation substage in the subphase. 
       Since no messages are passed between vertices in the preparation substage by construction, vertices can run this substage internally. Hence the preparation substage runs in $O(1)$ rounds in CONGEST.
       
      Consider the search substage in the subphase. By construction in the \nameref{j_phase_construction_case_1}, the \nameref{j_phase_construction_case_2}, and the \nameref{j_phase_construction_case_3}, any message passing through an edge $e'$ is of the form $\tau_{\operatorname{search}}=(f,x,y,\A,\B)$ for some $f\in U$, $x,y\in V(G)$ with $\psi(x)=c$, and $\A,\B\in[\Delta+1]$. Such a message can be described by $O(\log n)$ bits, and no such message passes across the same edge more than once. Again by construction, $\tau_{\operatorname{search}}$ passes across $e'$ only if either 1) $e'$ is incident to $x$, where $\psi(x)=c$, or 2) $e'$ is an edge in a maximal $\ABtuple$-bichromatic path, say $K'$, under $\chi$ such that $x$ is a neighbor of an endpoint of $K'$. Under $\chi$, there are at most $O(\Delta^2)$ maximal bichromatic paths that $e'$ is an edge of, and for each such path, at most one vertex in the neighborhood of each of its endpoints is colored $c$ under $\psi$. Since, moreover, at most one endpoint of $e'$ is colored $c$, there are at most $O(\Delta^2)$ choices for such a vertex $x$; as each such vertex serves a single tuple in the subphase, at most one message of this form with second entry $x$ is sent in the substage, so there are at most $O(\Delta^2)$ messages that pass through $e'$ in the search substage. The substage also sends the messages $\tau_{\operatorname{return}}$, but a message $\tau_{\operatorname{return}}=(f,x)$ passes across $e'$ only if $e'$ is incident to $x$, and the vertex colored $c$ incident to $e'$, if any, serves a single tuple in the subphase; hence at most one such message passes across $e'$. Hence algorithm $\mathcal{A}$ has a message size of $O(\Delta^2 \log n)$ bits in the search substage. Clearly, the search substage runs in $O(\len)$ rounds in the LOCAL model. Hence the search substage can be simulated in $O(\Delta^2 \len)$ rounds in the CONGEST model by \Cref{obsvn:LOCAL_to_CONGEST}.

        Consider the notification substage in the subphase. By construction in the \nameref{j_phase_construction_case_1}, the \nameref{j_phase_construction_case_2}, and the \nameref{j_phase_construction_case_3}, note that any message passing through an edge $e'$ is of the form $\tau_{\operatorname{notify}}=(f,x,y,\neg,\A)$ or $\tau_{\operatorname{notify}}=(f,x,y,\A,\B)$ for some $f\in U$, $x,y\in V(G)$ with $\psi(x)=c$, and $\A,\B\in[\Delta+1]$. Such a message can be described by $O(\log n)$ bits, and no such message passes across the same edge more than once. We also have that $\tau_{\operatorname{notify}}$ passes across $e'$ only if either 1) $e'$ is incident to $x$, where $\psi(x)=c$, or 2) $e'$ is an edge in a maximal bichromatic path, say $Q'$, under $\chi$ such that $x$ is a neighbor of an endpoint of $Q'$. Since there are at most $O(\Delta^2)$ maximal bichromatic paths that $e'$ is an edge of under $\chi$, for each such path at most one vertex in the neighborhood of each of its endpoints is colored $c$ under $\psi$, and at most one endpoint of $e'$ is colored $c$, there are at most $O(\Delta^2)$ choices for such a vertex $x$; as each such vertex serves a single tuple in the subphase, at most one message of this form with second entry $x$ is sent in the substage, so there are at most $O(\Delta^2)$ messages that pass through $e'$ in the notification substage.  Hence algorithm $\mathcal{A}$ has a message size of $O(\Delta^2 \log n)$ bits in the notification substage. Clearly, the notification substage runs in $O(\len)$ rounds in the LOCAL model. Hence the notification substage can be simulated in $O(\Delta^2 \len)$ rounds in the CONGEST model by \Cref{obsvn:LOCAL_to_CONGEST}.

        Consider the confirmation substage in the subphase. By construction in the \nameref{j_phase_construction_case_1}, the \nameref{j_phase_construction_case_2}, and the \nameref{j_phase_construction_case_3}, note that any message passing through an edge $e'$ is of the form 
         $\tau_{\operatorname{reply}}=(f,x,y,\neg,\A,\suc)$ or
         $\tau_{\operatorname{reply}}=(f,x,y,\A,\B,\suc)$,  for some $f\in U$, $x,y\in V(G)$ with $\psi(x)=c$, $\A,\B\in[\Delta+1]$, and $\suc\in\{-1,0,1\}$. Such a message can be described by $O(\log n)$ bits, and no such message passes across the same edge more than once. Note that $\tau_{\operatorname{reply}}$ passes through $e'$ only if either 1) $e'$ is incident to $x$, where $\psi(x)=c$, or 2) $e'$ is an edge in a maximal $\ABtuple$-bichromatic path, say $Q'$, under $\chi$ such that $x$ is a neighbor of an endpoint of $Q'$. Since there are at most $O(\Delta^2)$ maximal bichromatic paths that $e'$ is an edge of under $\chi$, for each such path at most one vertex in the neighborhood of each of its endpoints is colored $c$ under $\psi$, and at most one endpoint of $e'$ is colored $c$, there are at most $O(\Delta^2)$ choices for such a vertex $x$; as each such vertex serves a single tuple in the subphase, at most one message of this form with second entry $x$ is sent in the substage, so there are at most $O(\Delta^2)$ messages that pass through $e'$ in the confirmation substage.  Hence algorithm $\mathcal{A}$ has a message size of $O(\Delta^2 \cdot \log n)$ bits in the confirmation substage. Clearly, the confirmation substage runs in $O(\len)$ rounds in the LOCAL model.   Hence the confirmation substage can be simulated in $O(\Delta^2 \cdot \len)$ rounds in the CONGEST model by \Cref{obsvn:LOCAL_to_CONGEST}.

        Consider the set construction stage in the subphase. 
        All additions to the sets $D_{j-1}(e,\len)$, $Q_{j-1}(e,\len)$, $B_{j-1}(e,\len)$, $R_j(e,\len)$, $S_j(e,\len)$ and $M_j(v)$ are decided locally from information received during the chain construction stage, so any message sent across any edge $e'$ is of the form $(v,e)$ for some $v\in V(G)$ and $e\in U$. 
        Any message $(v,e)$ is only sent between 
        vertices in a $2$-hop neighborhood of $v$ when $v$ 
        wants to notify vertices in its $2$-hop neighborhood 
        that $v$ has added itself to $R_j(e,\len)$. Hence if
        $(v,e)$ passes through $e'$, $v$ is adjacent to one 
        of the endpoints of $e'$ and $v\in R_j(e,\len)$. By \Cref{lemma:bounded_addition_Phase_j}, vertex $v$ adds itself to $R_j(e,\len)$ in this subphase for at most $O(\Delta^2)$ edges $e$. Since there are $O(\Delta)$ vertices that are adjacent to either endpoint of $e'$, we  then have that at most $O(\Delta^3)$ messages of the form $(v,e)$ for some $v\in V(G)$ and $e\in E(G)$ pass through $e'$. Since each such message can be
        described by $O(\log n)$ bits, algorithm $\mathcal{A}$ has a
        message size of $O(\Delta^3\log n)$ in the set construction
        stage. Clearly, the set construction stage runs in $O(1)$ rounds
        in the LOCAL model. Hence the set construction stage can be
        simulated in $O(\Delta^3)$ rounds in the CONGEST model by
        \Cref{obsvn:LOCAL_to_CONGEST}.

         Hence, each subphase in construction phase $j$ runs in $O(\Delta^3+\Delta^2\len)=O(\Delta^3\len)$ rounds in CONGEST.
\end{proof}
We now bound the runtime of construction phase $j$ for $2\leq j\leq T$.
\RuntimePhaseJ*
\begin{proof}
    Since there are $O(\Delta^4)$ subphases in construction phase $j$ and each subphase runs in $O(\Delta^3\len)$ rounds in CONGEST by \Cref{lemma:runtime_phase_j_subphase}, construction phase $j$ runs in $O(\Delta^7\len)$ rounds in CONGEST.
\end{proof}
\subsection{Establishing the Statements \texorpdfstring{$\Pi_1$, $\Pi_2$ and $\Pi_3$}{Pi 1, Pi 2 and Pi 3}}\label{subsec:statement_Pi1_Pi2_Pi3_proofs}
    
For $i\geq 2$, we assumed statements $\Pi_1(i-1)$, $\Pi_2(i-1)$ and $\Pi_3(i-1)$ to hold at the start of construction phase $i$. Here we prove that those assumptions are valid. 

We rewrite the statements once more.

\begin{statement}[$\Pi_1(i)$ for $1\leq i \leq T$]
   Assume construction phase $i$ is over. Consider any vertex $x\in V(G)$ and any $(e,z,\A,\B,xy)\in M_{i}(x)$. Then there is a unique $i$-step self-avoiding Vizing chain $C=F_1+P_1+\dots+F_{i}+P_{i}$ of step length $\len$ on $e$ under $\chi$ such that
        \begin{itemize}
            \item for $1\leq k \leq i$, $F_k=F_k^*$ and $P_k$ is a subpath of $P_k^*$ where $F_k^*+P_k^*$ is the unique candidate chain constructed for $e$ by $\cen(F_k)$ during construction phase $k$,

            \item for $1\leq k \leq i$, the set $M_k(\vend(P_k))$ contains a tuple with first entry $e$ and second entry $\cen(F_k)$, and

            \item $xy=\eend(P_{i})$ and $x=\vend(P_{i})$. 
        \end{itemize}
        Moreover, this chain $C$ has the additional properties that
        \begin{itemize}
        \item $\cen(F_i)=z$,
         
        \item for $1\leq k \leq i$, $V(\bichrom(P_k))\subseteq R_k(e,\len)$, 

        \item $\bichrom(P_i)$ is a $\vend(F_i)$-maximal $\ABtuple$-bichromatic path under $\chi$,

        \item $P_{i}$ is either a $\BAtuple$-bichromatic path chain or an $\ABtuple$-bichromatic path chain under $\Shift(\chi,F_1+P_1+\dots+F_{i-1}+P_{i-1}+F_{i})$ and length of $P_{i}$ is greater than $2$, 

        \item $\A\in A(x,\Shift(\chi,C))$, $\B\in A(y,\Shift(\chi,C))$ and $\Shift(\chi,F_1+P_1+\dots+F_i)(xy)=\A$.

        \end{itemize}       
\end{statement}

\begin{statement}[$\Pi_2(i)$]
      Assume construction phase $i$ is over. Then, $|M_{i}(x)|\leq 2\Delta^2$ for $x\in V(G)$.
\end{statement}

\begin{statement}[$\Pi_3(i)$]
    Assume construction phase $i$ is finished. Consider any $x\in V(G)$, $e\in U$.
    \begin{itemize}
        \item  If $x\in S_{i}(e,\len)$, then $x\notin N^2[\bigcup_{k=0}^{i-1}R_k(e,\len)]$.
        \item $x\in S_{i}(e,\len) \iff (e,z,\A,\B,xy)\in M_{i}(x)$ for unique $z\in S_{i-1}(e,\len)$, $\A,\B\in[\Delta+1]$ and $y\in N(x)$.
        \end{itemize}
\end{statement}

We prove these statements by induction on the phase index, across the following two lemmas.
\begin{lemma}\label{lemma:Pi1_Pi2_Pi3_1}
    Statements $\Pi_1(1)$, $\Pi_2(1)$ and $\Pi_3(1)$ are true.
\end{lemma}      
\begin{proof}
 We will prove $\Pi_3(1)$, $\Pi_1(1)$ and $\Pi_2(1)$ in that order.  
     \subparagraph*{Proof of $\Pi_3(1)$.}
        Assume construction phase $1$ is finished. Consider a vertex $x$ being added to $S_1(e,\len)$ in some subphase in construction phase $1$. By the description of the set construction stage in construction phase $1$, $x\notin N^2[x']$, where $x'$ is the vertex of $S_0(e,\len)$ that constructed the candidate chain on account of which $x$ was added. Since $R_0(e,\len)=S_0(e,\len)=\{x'\}$ by the preprocessing steps, $x\notin N^2[R_0(e,\len)]$. By \Cref{obsvn:S_0_M_0_relation}, 
     \begin{align*}
         x\in S_{1}(e,\len) \iff &(e,z,\A,\B,xy)\in M_{1}(x)\\ &\text{ for unique $z\in S_{0}(e,\len)$, $\A,\B\in[\Delta+1]$ and $y\in N(x)$}.
     \end{align*}
     Hence $\Pi_3(1)$ is true.
 
    \subparagraph*{Proof of $\Pi_1(1)$.}
    Assume that construction phase $1$ has finished. Let $x\in V(G)$ and let $(e,z,\A,\B,xy)\in M_1(x)$. By \Cref{obsvn:prprty_of_1_step_chain}, the vertex $z\in S_0(e,\len)$ finishes the construction of a unique candidate chain $F+P$ for $e$ in construction phase $1$ such that
\begin{itemize}
    \item $F+P$ is a self-avoiding $1$-step Vizing chain of step length $\len$ on $e$ under $\chi$,
    \item $V(\bichrom(P))\subseteq R_1(e,\len)$,
    \item $\cen(F)=z$,
    \item $\bichrom(P)$ is a $\vend(F)$-maximal bichromatic path under $\chi$,
    \item $P$ is either an $(\A,\B)$-bichromatic path chain or a $(\B,\A)$-bichromatic path chain under $\Shift(\chi,F)$, and
    \item $xy=\eend(P|d_P(z,x))$ and $\Shift(\chi,F)(xy)=\A$.
\end{itemize}

Now let $F_1:=F$ and let $P_1:=P|d_P(z,x)$. We show that $F_1+P_1$ has the required properties.

First, $F_1+P_1$ is a self-avoiding $1$-step Vizing chain of step length $\len$ on $e$ under $\chi$, since it is an initial segment of the self-avoiding $1$-step Vizing chain $F+P$. Also, by \Cref{obsvn:S_0_construct_candidates}, the chain $F+P$ is exactly the candidate chain constructed for $e$ by $\cen(F)=z$ during construction phase $1$.

Thus, $F_1+P_1$ is a self-avoiding $1$-step Vizing chain of step length $\len$ on $e$ under $\chi$ such that
\begin{itemize}
    \item $F_1=F$ and $P_1$ is a subpath of $P$, where $F+P$ is the unique candidate chain constructed for $e$ by $\cen(F)$ during construction phase $1$,
    \item $M_1(\vend(P_1))$ contains a tuple with first entry $e$ and second entry $\cen(F_1)$, namely $(e,z,\A,\B,xy)$, and
    \item $xy=\eend(P_1)$ and $x=\vend(P_1)$.
\end{itemize}
It is the only such chain. Indeed, suppose $F_1'+P_1'$ also satisfies these conditions. Since $\vend(P_1')=x$ and $|M_1(x;e)|\leq 1$ by \Cref{obsvn:M_1_no_duplicates}, the second condition forces $\cen(F_1')=z$. Hence $F_1'+P_1'$ is an initial segment of the unique candidate chain $F+P$ that $z$ constructs for $e$ in construction phase $1$, and $\vend(P_1')=x$ then gives $F_1'+P_1'=F_1+P_1$.

It remains to verify the additional properties.

\begin{itemize}
    \item We have $V(\bichrom(P_1))\subseteq R_1(e,\len)$, since $P_1$ is a subpath of $P$ and $V(\bichrom(P))\subseteq R_1(e,\len)$.

    \item The path $\bichrom(P_1)$ is a $\vend(F_1)$-maximal $\ABtuple$-bichromatic path under $\chi$. Indeed, $\bichrom(P)$ is a $\vend(F)$-maximal bichromatic path under $\chi$, its color pair is $\ABSet$ since $P$ is an $\ABtuple$- or a $\BAtuple$-bichromatic path chain under $\Shift(\chi,F)$, and $P_1$ is an initial segment of $P$ ending at $x$.

    \item The path $P_1$ is either an $\ABtuple$-bichromatic path chain or a $\BAtuple$-bichromatic path chain under $\Shift(\chi,F_1)$, since the same is true for $P$ under $\Shift(\chi,F)$.

    \item The length of $P_1$ is greater than $2$. Indeed, $x\in S_1(e,\len)$ by \Cref{obsvn:S_0_M_0_relation}, the set $S_1(e,\len)$ is disjoint from $N^2[R_0(e,\len)]$ by $\Pi_3(1)$, and $z\in S_0(e,\len)=R_0(e,\len)$.

    \item By the definition of $P_1$, we have $\Shift(\chi,F_1)(\eend(P_1))=\A$.

    \item It follows that $\A\in A(x,\Shift(\chi,F_1+P_1))$.

    \item We also have $\B\in A(y,\Shift(\chi,F_1+P_1))$. Let $f$ be the edge preceding $xy$ in $P_1$. Since the length of $P_1$ is greater than $2$, the edge $f$ lies in $\bichrom(P_1)$. Moreover, because $P_1$ is either an $\ABtuple$-bichromatic path chain or a $\BAtuple$-bichromatic path chain under $\Shift(\chi,F_1)$, the edge $f$ is colored $\B$ under $\Shift(\chi,F_1)$. Hence, after shifting along $P_1$, the color $\B$ becomes available at $y$.
\end{itemize}

Therefore, $F_1+P_1$ is the unique chain with all the required properties.

\subparagraph*{Proof of $\Pi_2(1)$.}

Assume, if possible, that two distinct elements $(e_a,x_a,\A,\B,xy)$ and
$(e_b,x_b,\A,\B,xy)$ are elements of $M_1(x)$ after construction phase $1$. There is no
loss of generality in giving the two tuples the same last entry: if
$(e,x',\A,\B,xy)\in M_1(x)$ then $\chi(xy)=\A$ by the set construction stage, so
$xy$ is the unique $\A$-colored edge at $x$, and is therefore determined by
$\A$. By
\Cref{obsvn:M_1_no_duplicates}, the set $M_1(x)$ contains at most one
tuple with first entry $e$ for every $e\in U$. As the two tuples are distinct,
we conclude $e_a\neq e_b$.

By the description of the set construction stage in construction phase $1$,
the tuple $(e_a,x_a,\A,\B,xy)$ is added to $M_1(x)$ only because $x_a$
finished the construction of a candidate chain $F_a+P_a$ for $e_a$ in
construction phase $1$ with $x\in V(\bichrom(P_a))\setminus N^2[x_a]$ and
$xy=\eend(P_a|d_{P_a}(x_a,x))$. By the description of subphase $(c,r)$ in
construction phase $1$, this construction happened in subphase
$(\psi(x_a),\lambda_{x_a,0}(e_a))$, since $e_a\in M_0(x_a)$ is the edge that
forced the construction. Analogously, $x_b$ finished the construction of a
candidate chain $F_b+P_b$ for $e_b$ in subphase
$(\psi(x_b),\lambda_{x_b,0}(e_b))$.

Note that the lengths of $P_a|d_{P_a}(x_a,x)$ and $P_b|d_{P_b}(x_b,x)$ are
greater than $2$ by statement $\Pi_1(1)$. \Cref{obsvn:prprty_of_1_step_chain}
also implies that $\bichrom(P_a)$ and $\bichrom(P_b)$ are both
$\ABtuple$-bichromatic paths under $\chi$, and that $\vend(F_a)$ and
$\vend(F_b)$ are endpoints of the maximal $\ABtuple$-bichromatic path under
$\chi$ containing $x$.

We now claim that $\vend(F_a)=\vend(F_b)$. Suppose not. Since $x$ is at
distance at least $3$ from $x_a$ and $x_b$ along $P_a$ and $P_b$
respectively, $x$ is at distance at least $2$ from $\vend(F_a)$ and
$\vend(F_b)$ along $\bichrom(P_a)$ and $\bichrom(P_b)$ respectively. As
$\vend(F_a)$ and $\vend(F_b)$ are distinct endpoints of the same maximal
$\ABtuple$-bichromatic path containing $x$, the paths $P_a$ and $P_b$ then
enter $x$ through different edges of that maximal path. This contradicts
$\eend(P_a|d_{P_a}(x_a,x))=xy=\eend(P_b|d_{P_b}(x_b,x))$. Hence
$\vend(F_a)=\vend(F_b)$.

Since $\bichrom(P_a)$ and $\bichrom(P_b)$ are subpaths of the same maximal
$\ABtuple$-bichromatic path under $\chi$, both starting at
$\vend(F_a)=\vend(F_b)$ and both ending at $x$, they traverse the same
segment of that path from $\vend(F_a)$ to $x$. This implies
$d_{P_a}(x_a,x)=d_{P_b}(x_b,x)$.

On the other hand, by the description of the set construction stage in
construction phase $1$, the tuples being added to $M_1(x)$ implies
\begin{itemize}
    \item $\frac{\len}{10\Delta^4}\cdot(\Delta^3\cdot\mu_{\vend(F_a)}(x_a)+
    \lambda_{x_a,0}(e_a))\leq d_{P_a}(x_a,x) <
    \frac{\len}{10\Delta^4}\cdot(\Delta^3\cdot\mu_{\vend(F_a)}(x_a)+
    \lambda_{x_a,0}(e_a)+1)$, and
    \item $\frac{\len}{10\Delta^4}\cdot(\Delta^3\cdot\mu_{\vend(F_b)}(x_b)+
    \lambda_{x_b,0}(e_b))\leq d_{P_b}(x_b,x) <
    \frac{\len}{10\Delta^4}\cdot(\Delta^3\cdot\mu_{\vend(F_b)}(x_b)+
    \lambda_{x_b,0}(e_b)+1)$.
\end{itemize}
Since $d_{P_a}(x_a,x)=d_{P_b}(x_b,x)$ and these half-open intervals are
pairwise disjoint for distinct values of the index
$\Delta^3\cdot\mu+\lambda$, we must have
\begin{equation} \label{eqn:Pi2_index_equal_phase_1}
\Delta^3\cdot\mu_{\vend(F_a)}(x_a)+\lambda_{x_a,0}(e_a)
=
\Delta^3\cdot\mu_{\vend(F_b)}(x_b)+\lambda_{x_b,0}(e_b).
\end{equation}
We derive a contradiction in each of the following exhaustive cases.
\begin{itemize}
    \item Suppose $x_a=x_b$. Then $e_a,e_b\in M_0(x_a)$ and $e_a\neq e_b$.
    As $\lambda_{x_a,0}$ is injective,
    $\lambda_{x_a,0}(e_a)\neq\lambda_{x_a,0}(e_b)$. Since additionally
    $\mu_{\vend(F_a)}(x_a)=\mu_{\vend(F_b)}(x_b)$ (as $x_a=x_b$ and
    $\vend(F_a)=\vend(F_b)$), the two sides of
    \Cref{eqn:Pi2_index_equal_phase_1} differ, a contradiction.

    \item Suppose $x_a\neq x_b$. Since $\vend(F_a)=\vend(F_b)$ and
    $\mu_{\vend(F_a)}$ is injective, we have
    $\mu_{\vend(F_a)}(x_a)\neq\mu_{\vend(F_b)}(x_b)$, and hence the two
    multiples of $\Delta^3$ in \Cref{eqn:Pi2_index_equal_phase_1} differ by
    at least $\Delta^3$. As $\lambda_{x_a,0}(e_a),\lambda_{x_b,0}(e_b)\in
    [\Delta]$, their difference is at most $\Delta-1<\Delta^3$. Hence
    \Cref{eqn:Pi2_index_equal_phase_1} cannot hold, a contradiction.
\end{itemize}

Hence there do not exist distinct elements
$(e_a,x_a,\A,\B,xy),(e_b,x_b,\A,\B,xy)$ in $M_1(x)$. That is, for every
ordered pair $(\A,\B)$ of distinct colors in $[\Delta+1]$, there exists at
most one element of the form $(e,x',\A,\B,xy)$ in $M_1(x)$. Since
there are at most $(\Delta+1)\Delta\leq 2\Delta^2$ ordered pairs of distinct
colors, we conclude $|M_1(x)|\leq 2\Delta^2$. Hence $\Pi_2(1)$ is true.
\end{proof}

The second lemma is the inductive step: it establishes the three statements for construction phase $i$, assuming they hold for all earlier phases.
\begin{lemma}\label{lemma:Pi1_Pi2_Pi3_i}
    For $2\leq i\leq T $, statements $\Pi_1(i)$, $\Pi_2(i)$ and $\Pi_3(i)$ hold.
\end{lemma}
\begin{proof}
     Let $2\leq i \leq T$.
    We will inductively assume that $\Pi_1(k)$, $\Pi_2(k)$ and $\Pi_3(k)$ are true for $1\leq k\leq i-1$. Then $\Cref{fact_1}$ and $\Cref{fact_2}$ hold at the start of construction phase $i$ and hence construction phase $i$ is well defined.

   \subparagraph*{Proof for $\Pi_3(i)$.} Assume construction phase $i$ is finished. Consider a vertex $x$ being added to $S_i(e,\len)$ in some subphase in construction phase $i$. By the description of the set construction stage in construction phase $i$, $x\notin N^2[\bigcup_{k=0}^{i-1}R_k(e,\len)]$. 

    Then, by \Cref{obsvn:S_j_M_j_relation}, we have $x\in S_{i}(e,\len) \iff (e,z,\A,\B,xy)\in M_{i}(x)$ for unique $z\in S_{i-1}(e,\len)$, $\A,\B\in[\Delta+1]$ and $y\in N(x)$.

     Hence $\Pi_3(i)$ is true.
    
\subparagraph*{Proof for $\Pi_1(i)$.}
Let $x\in V(G)$ and let $(e,z,\A,\B,xy)\in M_i(x)$. We apply \Cref{obsvn:prprty_of_j_step_chain} to this tuple, i.e., with the roles of the observation's variables $v$, $x$, $\A'$, $\B'$, $vw$ played by $x$, $z$, $\A$, $\B$, $xy$ respectively. It follows that $z\in S_{i-1}(e,\len)$ finishes the construction of a candidate chain $F_i'+P_i'$ under $\Shift(\chi,C)$, where
$$
C=F_1+P_1+\dots+F_{i-1}+P_{i-1}
$$
is an $(i-1)$-step Vizing chain on $e$ under $\chi$. Note that throughout this proof the letter $C$ denotes this $(i-1)$-step prefix, whereas in the statement of $\Pi_1(i)$ it stands for the full $i$-step chain. Moreover, this chain satisfies the following properties:
\begin{itemize}
    \item $C+F_i'+P_i'$ is a self-avoiding $i$-step Vizing chain of step length $\len$ on $e$ under $\chi$ (by \Cref{prprty:j_step_self_avoiding}),
    \item for each $1\leq k\leq i-1$, we have $F_k=F_k'$ and $P_k$ is a subpath of $P_k'$, where $F_k'+P_k'$ is the unique candidate chain constructed for $e$ by $\cen(F_k)$ during construction phase $k$ (by \Cref{prprty:j_step_subpath}),
    \item for each $1\leq k\leq i-1$, the set $M_k(\vend(P_k))$ contains a tuple with first entry $e$ and second entry $\cen(F_k)$ (since $C$ is the chain that statement $\Pi_1(i-1)$ gives for the tuple of $M_{i-1}(z)$ served by $z$),
    \item $V(\bichrom(P_i'))\subseteq R_i(e,\len)$ and $V(\bichrom(P_k))\subseteq R_k(e,\len)$ for each $1\leq k\leq i-1$ (by \Cref{prprty:j_step_R_containment}),
    \item $\cen(F_i')=z$,
    \item $\bichrom(P_i')$ is a $\vend(F_i')$-maximal bichromatic path under $\chi$ (by \Cref{prprty:j_step_maximal}),
    \item $P_i'$ is either an $(\A,\B)$-bichromatic path chain or a $(\B,\A)$-bichromatic path chain under $\Shift(\chi,C+F_i')$ (by \Cref{prprty:j_step_bichromatic_chain}),
    \item $xy=\eend(P_i'|d_{P_i'}(z,x))$ and $\Shift(\chi,C+F_i')(xy)=\A$ (by \Cref{prprty:j_step_eend_color}).
\end{itemize}

Also, by \Cref{obsvn:S_j_M_j_relation}, we have $x\in S_i(e,\len)$.

Now let $F_i:=F_i'$ and let $P_i$ be the initial segment of $P_i'$ ending at $x$. Then, by construction,
$$
xy=\eend(P_i), \qquad x=\vend(P_i),
$$
and
$$
\Shift(\chi,C+F_i)(xy)=\A.
$$
We claim that $C+F_i+P_i$ satisfies all the remaining required properties.

Since $P_i$ is an initial segment of $P_i'$, we immediately get
$$
V(\bichrom(P_i))\subseteq R_i(e,\len).
$$
Similarly, because $\bichrom(P_i')$ is a $\vend(F_i')$-maximal $\ABtuple$-bichromatic path under $\chi$, it follows that $\bichrom(P_i)$ is a $\vend(F_i)$-maximal $\ABtuple$-bichromatic path under $\chi$.

Moreover, since $P_i'$ is either an $\ABtuple$-bichromatic path chain or a $\BAtuple$-bichromatic path chain under $\Shift(\chi,C+F_i')$, the same is true for $P_i$ under $\Shift(\chi,C+F_i)$. The length of $P_i$ is greater than $2$. Indeed, we have $x\in S_i(e,\len)$, the set $S_i(e,\len)$ is disjoint from
$$
N^2\!\left[\bigcup_{k=0}^{i-1}R_k(e,\len)\right]
$$
by $\Pi_3(i)$, and $z\in R_{i-1}(e,\len)$.

Next we verify the availability conditions. We already know that
$$
\Shift(\chi,C+F_i)(xy)=\A.
$$
Let $f$ be the edge preceding $xy$ in $P_i$. Since $P_i$ has length greater than $2$ and is either an $\ABtuple$- or a $\BAtuple$-bichromatic path chain under $\Shift(\chi,C+F_i)$, we have
$$
\Shift(\chi,C+F_i)(f)=\B.
$$
Therefore, after shifting along $P_i$, the color $\A$ is available at $x$ and the color $\B$ is available at $y$. That is,
$$
\A\in A(x,\Shift(\chi,C+F_i+P_i))
\quad\text{and}\quad
\B\in A(y,\Shift(\chi,C+F_i+P_i)).
$$

Furthermore, $C+F_i+P_i$ is a self-avoiding $i$-step Vizing chain of step length $\len$ on $e$ under $\chi$, because it is an initial segment of the self-avoiding $i$-step Vizing chain $C+F_i'+P_i'$.

Finally, the chain
$$
F_1+P_1+\dots+F_i+P_i
$$
satisfies the conditions of $\Pi_1(i)$, namely
\begin{itemize}
    \item for each $1\leq k\leq i$, we have $F_k=F_k'$ and $P_k$ is a subpath of $P_k'$, where $F_k'+P_k'$ is the unique candidate chain constructed by $\cen(F_k)$ in construction phase $k$,
    \item for each $1\leq k\leq i$, the set $M_k(\vend(P_k))$ contains a tuple with first entry $e$ and second entry $\cen(F_k)$, and
    \item $xy=\eend(P_i)$ and $x=\vend(P_i)$.
\end{itemize}
The second condition holds for $k\leq i-1$ as recorded above, and for $k=i$ because $\vend(P_i)=x$ and $(e,z,\A,\B,xy)\in M_i(x)$ with $\cen(F_i)=z$.

No other chain satisfies these conditions. Indeed, let $\tilde F_1+\tilde P_1+\dots+\tilde F_i+\tilde P_i$ be a chain that does. Fix $1\leq k\leq i$ and suppose $\vend(\tilde P_k)=\vend(P_k)$, which holds for $k=i$ since both are equal to $x$. Writing $v:=\vend(\tilde P_k)$, the fiber $M_k(v;e)$ contains at most one tuple by \Cref{obsvn:M_1_no_duplicates,obsvn:M_j_no_duplicates}, and the second condition places in it a tuple with second entry $\cen(\tilde F_k)$ and a tuple with second entry $\cen(F_k)$; hence $\cen(\tilde F_k)=\cen(F_k)$. By \Cref{obsvn:S_0_construct_candidates,obsvn:S_j-1_construct_candidates}, this vertex constructs exactly one candidate chain for $e$ in construction phase $k$, so the first condition makes $\tilde F_k+\tilde P_k$ an initial segment of that chain, and $\vend(\tilde P_k)=v$ then gives $\tilde F_k+\tilde P_k=F_k+P_k$. In particular $\vend(\tilde P_{k-1})=\cen(\tilde F_k)=\cen(F_k)=\vend(P_{k-1})$, so the same argument applies at index $k-1$. Descending from $k=i$ to $k=1$ therefore gives $\tilde F_k+\tilde P_k=F_k+P_k$ for every $k$.

In particular, this chain satisfies the following additional properties:
\begin{itemize}
    \item $\cen(F_i)=z$,
    \item for each $1\leq k\leq i$, we have $V(\bichrom(P_k))\subseteq R_k(e,\len)$,
    \item $\bichrom(P_i)$ is a $\vend(F_i)$-maximal $\ABtuple$-bichromatic path under $\chi$,
    \item $P_i$ is either a $\BAtuple$-bichromatic path chain or an $\ABtuple$-bichromatic path chain under $\Shift(\chi,F_1+P_1+\dots+F_{i-1}+P_{i-1}+F_i)$, and the length of $P_i$ is greater than $2$,
    \item $\A\in A(x,\Shift(\chi,C+F_i+P_i))$, $\B\in A(y,\Shift(\chi,C+F_i+P_i))$, and
    $$
    \Shift(\chi,F_1+P_1+\dots+F_i)(xy)=\A.
    $$
\end{itemize}

Hence $F_1+P_1+\dots+F_i+P_i$ is the unique $i$-step Vizing chain with the required properties. Therefore, $\Pi_1(i)$ holds.

\subparagraph*{Proof for $\Pi_2(i)$.}

Assume, if possible, that two distinct elements $(e_a,x_a,\A,\B,xy)$ and
$(e_b,x_b,\A,\B,xy)$ are elements of $M_i(x)$ after construction phase $i$. There is no
loss of generality in giving the two tuples the same last entry: if
$(e,x',\A,\B,xy)\in M_i(x)$ then $\chi(xy)=\A$ by the set construction stage, so
$xy$ is the unique $\A$-colored edge at $x$, and is therefore determined by
$\A$. By
\Cref{obsvn:M_j_no_duplicates} (applied at index $i$), we have
$|M_i(x;e)|\leq 1$ for every $e\in U$. As the two tuples are distinct, we
conclude $e_a\neq e_b$.

By the description of the set construction stage in construction phase $i$,
the tuple $(e_a,x_a,\A,\B,xy)$ is added to $M_i(x)$ only because $x_a$
finished the construction of a candidate chain $F_a+P_a$ for $e_a$ in
construction phase $i$ with $x\in V(\bichrom(P_a))\setminus N^2[x_a]$ and
$xy=\eend(P_a|d_{P_a}(x_a,x))$. By the description of subphase $(c,r)$ in
construction phase $i$, this construction happened in subphase
$(\psi(x_a),\lambda_{x_a,i-1}(t_a))$, where $t_a\in M_{i-1}(x_a)$ is the
unique tuple with first entry $e_a$ that forced the construction;
uniqueness holds by \Cref{obsvn:M_1_no_duplicates} when $i-1=1$, and
otherwise by \Cref{obsvn:M_j_no_duplicates} applied at index $i-1$.
Analogously, $x_b$ finished the construction of a candidate chain $F_b+P_b$
for $e_b$ in subphase $(\psi(x_b),\lambda_{x_b,i-1}(t_b))$, where $t_b\in
M_{i-1}(x_b)$ is the unique tuple with first entry $e_b$.

Note that the lengths of $P_a|d_{P_a}(x_a,x)$ and $P_b|d_{P_b}(x_b,x)$ are
greater than $2$ by statement $\Pi_1(i)$.
\Cref{prprty:j_step_bichromatic_chain,prprty:j_step_maximal} of \Cref{obsvn:prprty_of_j_step_chain} also imply that
$\bichrom(P_a)$ and $\bichrom(P_b)$ are both $\ABtuple$-bichromatic paths
under $\chi$, and that $\vend(F_a)$ and $\vend(F_b)$ are endpoints of the
maximal $\ABtuple$-bichromatic path under $\chi$ containing $x$.

We now claim that $\vend(F_a)=\vend(F_b)$. Suppose not. Since $x$ is at
distance at least $3$ from $x_a$ and $x_b$ along $P_a$ and $P_b$
respectively, $x$ is at distance at least $2$ from $\vend(F_a)$ and
$\vend(F_b)$ along $\bichrom(P_a)$ and $\bichrom(P_b)$ respectively. As
$\vend(F_a)$ and $\vend(F_b)$ are distinct endpoints of the same maximal
$\ABtuple$-bichromatic path containing $x$, the paths $P_a$ and $P_b$ then
enter $x$ through different edges of that maximal path. This contradicts
$\eend(P_a|d_{P_a}(x_a,x))=xy=\eend(P_b|d_{P_b}(x_b,x))$. Hence
$\vend(F_a)=\vend(F_b)$.

Since $\bichrom(P_a)$ and $\bichrom(P_b)$ are subpaths of the same maximal
$\ABtuple$-bichromatic path under $\chi$, both starting at
$\vend(F_a)=\vend(F_b)$ and both ending at $x$, they traverse the same
segment of that path from $\vend(F_a)$ to $x$. This implies
$d_{P_a}(x_a,x)=d_{P_b}(x_b,x)$.

On the other hand, by the description of the set construction stage in
construction phase $i$, the tuples being added to $M_i(x)$ implies
\begin{itemize}
    \item $\frac{\len}{10\Delta^4}\cdot(\Delta^3\cdot\mu_{\vend(F_a)}(x_a)+
    \lambda_{x_a,i-1}(t_a))\leq d_{P_a}(x_a,x)< 
    \frac{\len}{10\Delta^4}\cdot(\Delta^3\cdot\mu_{\vend(F_a)}(x_a)+
    \lambda_{x_a,i-1}(t_a)+1)$, and
    \item $\frac{\len}{10\Delta^4}\cdot(\Delta^3\cdot\mu_{\vend(F_b)}(x_b)+
    \lambda_{x_b,i-1}(t_b))\leq d_{P_b}(x_b,x) <
    \frac{\len}{10\Delta^4}\cdot(\Delta^3\cdot\mu_{\vend(F_b)}(x_b)+
    \lambda_{x_b,i-1}(t_b)+1)$.
\end{itemize}
Since $d_{P_a}(x_a,x)=d_{P_b}(x_b,x)$ and these half-open intervals are
pairwise disjoint for distinct values of the index
$\Delta^3\cdot\mu+\lambda$, we must have
\begin{equation} \label{eqn:Pi2_index_equal}
\Delta^3\cdot\mu_{\vend(F_a)}(x_a)+\lambda_{x_a,i-1}(t_a)
=
\Delta^3\cdot\mu_{\vend(F_b)}(x_b)+\lambda_{x_b,i-1}(t_b).
\end{equation}
We derive a contradiction in each of the following exhaustive cases.
\begin{itemize}
    \item Suppose $x_a=x_b$. Then $t_a,t_b\in M_{i-1}(x_a)$ and $t_a\neq
    t_b$, since their first entries $e_a$ and $e_b$ differ. As
    $\lambda_{x_a,i-1}$ is injective,
    $\lambda_{x_a,i-1}(t_a)\neq\lambda_{x_a,i-1}(t_b)$. Since additionally
    $\mu_{\vend(F_a)}(x_a)=\mu_{\vend(F_b)}(x_b)$ (as $x_a=x_b$ and
    $\vend(F_a)=\vend(F_b)$), the two sides of
    \Cref{eqn:Pi2_index_equal} differ, a contradiction.

    \item Suppose $x_a\neq x_b$. Since $\vend(F_a)=\vend(F_b)$ and
    $\mu_{\vend(F_a)}$ is injective, we have
    $\mu_{\vend(F_a)}(x_a)\neq\mu_{\vend(F_b)}(x_b)$, and hence the two
    multiples of $\Delta^3$ in \Cref{eqn:Pi2_index_equal} differ by at least
    $\Delta^3$. As $\lambda_{x_a,i-1}(t_a),\lambda_{x_b,i-1}(t_b)\in[2\Delta^2]$, their
    difference is at most $2\Delta^2-1<\Delta^3$ for $\Delta\geq 2$. Hence
    \Cref{eqn:Pi2_index_equal} cannot hold, a contradiction.
\end{itemize}

Hence there do not exist distinct elements
$(e_a,x_a,\A,\B,xy),(e_b,x_b,\A,\B,xy)$ in $M_i(x)$. That is, for every
ordered pair $(\A,\B)$ of distinct colors in $[\Delta+1]$, there exists at
most one element of the form $(e,x',\A,\B,xy)$ in $M_i(x)$. Since
there are at most $(\Delta+1)\Delta\leq 2\Delta^2$ ordered pairs of distinct
colors, we conclude $|M_i(x)|\leq 2\Delta^2$. Hence $\Pi_2(i)$ is true.

\end{proof}
\section{Existence of a Multi-Step Vizing Chain}\label{sec:existence-of-MSVC}

We will see that algorithm $\mathcal{A}$ gives a terminating self-avoiding multi-step Vizing chain for any uncolored edge $e$. The proof follows along similarly as in \cite{christiansen2023power}.

We observe that the set $S_i(e,\len)$ can be partitioned into $D_i(e,\len)$, $Q_i(e,\len)$, and $B_i(e,\len)$ for any $0\leq i < \numphases$.
\begin{observation} \label{claim:D_Q_S_disjoint}
    Let $e\in U$ and $0\leq i< \numphases $. Then $S_{0}(e,\len)=D_0(e,\len)\sqcup Q_0(e,\len)$ and $S_i(e,\len)=D_i(e,\len)\sqcup Q_i(e,\len) \sqcup B_i(e,\len)$.
\end{observation}
\begin{proof}
    We will consider cases $i=0$ and $i>0$ separately.
    \begin{itemize}
        \item Suppose $i=0$. Consider any $x\in S_0(e,\len)$. By \Cref{claim:x_in_D_or_Q}, $x\in D_0(e,\len)\sqcup Q_0(e,\len)$.
        \item Suppose $i>0$.  Consider any $x\in S_i(e,\len)$. By \Cref{obsvn:x_in_D_or_Q_or_B}, $x\in D_i(e,\len)\sqcup Q_{i}(e,\len)\sqcup B_i(e,\len)$.
    \end{itemize}
\end{proof}
We now determine the conditions under which $R_i(e,\len)$ and $S_j(e,\len)$ can intersect.
\begin{observation} \label{claim:R_i_and_S_j_disjointness} 
    Let $e\in U$ and $0\leq i,j \leq \numphases-1$ with $i<j$. Then $R_i(e,\len)\cap S_{j}(e,\len)=\emptyset$. In particular $S_i(e,\len)\cap S_{j}(e,\len)=\emptyset$
\end{observation}
\begin{proof}
    This follows from \Cref{lemma:Pi1_Pi2_Pi3_1,lemma:Pi1_Pi2_Pi3_i}.
\end{proof}

We state some results that help us to analyze the construction.

\begin{observation} \label{obsvn:density}
  Let $G$ be a graph of maximum degree $\Delta$. Then, we have that for any $H\subseteq G$,
  $$\frac{|E(H)|}{|V(H)|}\leq \frac{\Delta}{2}.$$
\end{observation}

For any edge $e$, we count the number of candidate chains for $e$ that pass through a vertex or an edge in a construction phase $i$. This will later help us to bound $B_i(e,\len)$, the set of bad vertices.
\begin{lemma} \label{lemma:chains_through_node_or_edge}
    Let $v\in V(G)$, $f\in E(G)$, $e\in U$ and $1\leq i\leq \numphases$. Then,
    \begin{itemize}
        \item at most $(\Delta+1)^3$ candidate chains for $e$ are constructed in construction phase $i$ such that $v$ is contained in the bichromatic part of the candidate chain, 
        \item at most $2\Delta^2$ candidate chains for $e$ are constructed in construction phase $i$ such that $f$ is contained in the bichromatic part of the candidate chain, and
        \item at most $\Delta+1$ candidate chains for $e$ are constructed in construction phase $i$ such that $v$ is contained in the fan chain of the candidate chain.
    \end{itemize}
\end{lemma}

\begin{proof}
    Note that for any candidate chain $F+P$ constructed for $e$ in construction phase $i$, $\bichrom(P)$ is a bichromatic path under $\chi$ by \Cref{claim:Phase_1_bichrom(P)_is_bichromatic} and   \Cref{claim:j_phase_bichrom(P)_is_bichromatic}. 

    Now, let $x$ be a vertex that constructs a candidate chain $F+P$ such that $v\in V(\bichrom(P))$ in construction phase $i$. Then $x$ is the neighbor of $\vend(F)$ where $\vend(F)$ is the endpoint of a maximal bichromatic path under $\chi$ in which $v$ is contained.  There are at most $\binom{\Delta+1}{2}=\frac{\Delta(\Delta+1)}{2}$ maximal bichromatic paths under $\chi$ that contain $v$, since such a path is determined by its pair of colors, and each of them has two endpoints; hence $\vend(F)$ must come from a set of $\Delta(\Delta+1)$ vertices in $G$ and consequently, $x$ must come from a set of $\Delta^2(\Delta+1)\leq(\Delta+1)^3$ vertices. Moreover, $x$ constructs at most one candidate chain for $e$ in construction phase $i$ since $M_{i-1}(x)$ does not have duplicates at the end of construction phase $i-1$. Hence at most $(\Delta+1)^3$ candidate chains for $e$ are constructed in construction phase $i$ such that $v$ is contained in the bichromatic part of the candidate chain.

    By an analogous argument, we also have that at most $2\Delta^2$ candidate chains for $e$ are constructed in construction phase $i$ such that $f$ is contained in the bichromatic part of the candidate chain.

    By construction, $v$ is contained in the fan of a candidate chain $F+P$ if it is constructed by a vertex in $N[v]$. Note that any vertex $v'\in N[x]$ can construct at most one candidate chain on $e$ in construction phase $i$ since $M_{i-1}(x)$ does not have duplicates at the start of construction phase $i$. Hence at most $\Delta+1$ candidate chains for $e$ are constructed in construction phase $i$ such that $v$ is contained in the fan chain of the candidate chain.
\end{proof}

We show that the algorithm ``reaches'' a large number of previously ``unreached'' vertices from an uncolored edge $e$ if a terminating multi-step Vizing chain from $e$ is already not found.
\begin{lemma} \label{lemma-num_new_nodes}
     Let $e\in U$ and $1\leq i < \numphases$. 
     Then $$|R_i(e,\len)|\geq \frac{\len}{\Delta^3}\cdot |Q_{i-1}(e,\len)|.$$
\end{lemma}
\begin{proof}
    
     For each $x\in Q_{i-1}(e,\len)$, let $F_x+P_x$ be the unique candidate chain that $x$ constructed for $e$ in construction phase $i$, and let
    $$
        \mathcal{P}_e = \parens*{\bichrom(P_x)}_{x\in Q_{i-1}(e,\len)}
    $$
    be the family of paths indexed by these vertices, so that $|\mathcal{P}_e|=|Q_{i-1}(e,\len)|$ even if two indices yield the same path.
     Note that the members of $\mathcal{P}_e$ are bichromatic paths under $\chi$.
     Let $H$ be the graph induced by their edges.

     Since $x\in Q_{i-1}(e,\len)$, the path $\bichrom(P_x)$ is not maximal under $\chi$ and hence has length $\len$. So the pairs $(x,f)$ with $f\in E(\bichrom(P_x))$ number exactly $\len\cdot|Q_{i-1}(e,\len)|$. As distinct indices give distinct candidate chains, \Cref{lemma:chains_through_node_or_edge} lets each edge $f$ occur in at most $2\Delta^2$ of these pairs.

     Hence there are at least $\frac{\len\cdot |Q_{i-1}(e,\len)|}{2\Delta^2}$ edges in $H$. The number of vertices is at most $|R_{i}(e,\len)|$ since $V(P)\subseteq R_i(e,\len)$ if $F+P$ was finished constructing in construction phase $i$ by the description of set construction stage. Hence by \Cref{obsvn:density},
     $$\frac{\len\cdot |Q_{i-1}(e,\len)|}{2\Delta^2\cdot |R_{i}(e,\len)|} \leq \frac{\Delta}{2}$$

     We get the required inequality by reordering the terms.
\end{proof}

Having bounded how many new vertices a construction phase reaches, we now bound from below how many of them carry the construction into the next phase, provided no terminating chain has appeared so far.
\begin{lemma} \label{lemma:Si_poly(Delta)_frac_of_Ri}
    Let $e\in U$ and $1\leq i\leq \numphases$. Suppose $\bigcup_{k=0}^{i-1} D_{k}(e,\len)=\emptyset$, and if $i\geq 2$, suppose additionally that $|R_{i-1}(e,\len)|\leq 22(\Delta+1)^7\cdot |Q_{i-1}(e,\len)|$. Then $$|S_i(e,\len)|\geq \frac{|R_i(e,\len)|}{11(\Delta+1)^7}.$$
\end{lemma}
\begin{proof}
    
    Let $v\in R_i(e,\len)$. By \Cref{obsvn:in_R_1_iff_in_P,obsvn:in_R_j_iff_in_P}, a vertex $v$ is added to $R_i(e,\len)$ if and only if some vertex $x\in S_{i-1}(e,\len)$ finishes the construction of a candidate chain $F+P$ for $e$ in construction phase $i$ such that $v\in V(P)\setminus\{x\}$. Moreover, as we argue below, every candidate chain for $e$ whose construction is finished in construction phase $i$ is hopeful, so $P$ has length $\len+1$ and $V(P)\setminus\{x\}=V(\bichrom(P))$. Hence,
    $$
    R_i(e,\len)=\bigcup_{P\in \mathcal{P}_e}V(P),
    $$
    where
    $$
    \mathcal{P}_e=\parens*{\bichrom(P_x)}_{x}
    $$
    is the family indexed by the vertices $x$ that finished the construction of a candidate chain $F_x+P_x$ on $e$ in construction phase $i$.

    Consider a candidate chain $F+P$ on $e$ constructed by some vertex $x$ in construction phase $i$. By \Cref{obsvn:S_0_construct_candidates,obsvn:S_j-1_construct_candidates}, we have $x\in S_{i-1}(e,\len)$. Since $x$ finished the construction of $F+P$, the chain $F+P$ is not unsuccessful. Hence $x\notin B_{i-1}(e,\len)$. By assumption, we also have $x\notin D_{i-1}(e,\len)$.

    Moreover, by \Cref{claim:x_in_D_or_Q,obsvn:x_in_D_or_Q_or_B}, every vertex in $S_{i-1}(e,\len)$ belongs to $Q_{i-1}(e,\len)\sqcup D_{i-1}(e,\len)\sqcup B_{i-1}(e,\len)$. Therefore, $x\in Q_{i-1}(e,\len)$. By the set construction stage for $x$, this implies that $F+P$ is a hopeful candidate chain. Hence, by \Cref{obsvn:hopeful_is_long_in_phase_1,obsvn:hopeful_is_long_in_phase_j}, the path $P$ has length $\len+1$.

    Let $S_i'(e,\len)$ denote the set of vertices that would be added to $S_i(e,\len)$ if the condition
    $$
    v\notin N^2\!\left[\bigcup_{k=0}^{i-1}R_k(e,\len)\right]
    $$
    were removed from the set construction stage. By the set construction stage for $x$, the vertices added to $S_i'(e,\len)$ are those $v\in V(\bichrom(P))\setminus N^2[x]$ whose distance $d_P(v,x)$ lies in a half-open interval of length $\frac{\len}{10\Delta^4}$; this interval is contained in $[0,\len]$, since $\Delta^3\mu_{\vend(F)}(x)+r+1\leq \Delta^4+\Delta^3+1\leq 10\Delta^4$. It therefore contains at least $\lfloor\frac{\len}{10\Delta^4}\rfloor$ vertices of $\bichrom(P)$, of which at most $|N^2[x]|\leq \Delta^2+1$ are discarded. Since $\len=1584(\Delta+1)^{22}$ and $\Delta\geq 2$, each member of $\mathcal{P}_e$ thus contributes at least
    $$
    \left\lfloor\frac{\len}{10\Delta^4}\right\rfloor-(\Delta^2+1)\geq \frac{\len+1}{10(\Delta+1)^4}
    $$
    vertices to $S_i'(e,\len)$, that is, at least a $\frac{1}{10(\Delta+1)^4}$-fraction of its $\len+1$ vertices.

    On the other hand, distinct indices give distinct candidate chains, so by \Cref{lemma:chains_through_node_or_edge} every vertex $v\in S_i'(e,\len)$ is contained in at most $(\Delta+1)^3$ members of $\mathcal{P}_e$. Therefore, at least a $\frac{1}{10(\Delta+1)^7}$-fraction of the vertices in $R_i(e,\len)$ belong to $S_i'(e,\len)$, which gives
    $$
    |S_i'(e,\len)|\ge \frac{|R_i(e,\len)|}{10(\Delta+1)^7}.
    $$

    We now account for the condition
    $$
    v\notin N^2\!\left[\bigcup_{k=0}^{i-1}R_k(e,\len)\right].
    $$
    By \Cref{obsvn:vertex_disjointness_j_j-2}, every candidate chain $F+P$ whose construction is finished by $x$ in construction phase $i$ satisfies
    $$
    V(P)\cap N^2\!\left[\bigcup_{k=0}^{i-2}R_k(e,\len)\right]=\emptyset
    $$
    (vacuously if $i=1$). Hence the only vertices of $S_i'(e,\len)$ that may fail to belong to $S_i(e,\len)$ are those in $N^2[R_{i-1}(e,\len)]$. Therefore,
    $$
    S_i(e,\len)\supseteq S_i'(e,\len)\setminus N^2[R_{i-1}(e,\len)].
    $$

    First consider the case $i=1$. Then $|R_0(e,\len)|=1$, and so
    $$
    |N^2[R_0(e,\len)]|\le 2\Delta^2.
    $$
Since $|R_1(e,\len)|\ge \len$ (as $S_0(e,\len)$ is a singleton and $D_0(e,\len)=\emptyset$, the unique vertex of $S_0(e,\len)$ lies in $Q_0(e,\len)$, so its candidate chain is hopeful and $R_1(e,\len)=V(\bichrom(P))$ with $|\bichrom(P)|=\len$),
    $$
    |S_1'(e,\len)|\ge \frac{|R_1(e,\len)|}{10(\Delta+1)^7}\ge \frac{\len}{10(\Delta+1)^7}= 158.4(\Delta+1)^{15},
    $$
    it follows that
    $$
    |S_1(e,\len)|
    \ge |S_1'(e,\len)|-2\Delta^2
    \ge \frac{|R_1(e,\len)|}{11(\Delta+1)^7}.
    $$

    Now consider the case $i\ge 2$. Then
    $$
    |N^2[R_{i-1}(e,\len)]|\le 2\Delta^2|R_{i-1}(e,\len)|.
    $$
    By the additional hypothesis,
    $$
    |R_{i-1}(e,\len)|\le 22(\Delta+1)^7|Q_{i-1}(e,\len)|,
    $$
    and hence
    $$
    |N^2[R_{i-1}(e,\len)]|
    \le 44(\Delta+1)^9|Q_{i-1}(e,\len)|.
    $$
    By \Cref{lemma-num_new_nodes},
    $$
    |Q_{i-1}(e,\len)|\le \frac{\Delta^3}{\len}\,|R_i(e,\len)|.
    $$
    Therefore,
    $$
    |S_i(e,\len)|
    \ge \frac{|R_i(e,\len)|}{10(\Delta+1)^7}
      - \frac{44\Delta^3(\Delta+1)^9}{\len}\,|R_i(e,\len)|.
    $$
    Since $\len=1584(\Delta+1)^{22}$, we obtain
    $$
    |S_i(e,\len)|
    \ge |R_i(e,\len)|\left(\frac{1}{10(\Delta+1)^7}-\frac{1}{30(\Delta+1)^{10}}\right)
    \ge \frac{|R_i(e,\len)|}{11(\Delta+1)^7}.
    $$
\end{proof}

We now show a bound for the growth rate of ``bad vertices'' $B_i(e,\len)$. 
\begin{lemma} \label{lemma-num_conflicts}
      Let $e\in U$ and $1\leq i < \numphases$. Then
   $$|B_i(e,\len)|\leq  3(\Delta+1)^5\cdot\sum_{j=0}^{i-1}|R_j(e,\len)|.$$
   
\end{lemma}

\begin{proof}
    Consider a vertex $x\in S_i(e,\len)$. By \Cref{obsvn:S_0_M_0_relation,obsvn:S_j_M_j_relation}, there exists a tuple $(e,x',\A,\B,xy)\in M_i(x)$ for some $\A,\B\in[\Delta+1]$. By construction, in construction phase $i+1$, the vertex $x$ constructs a $1$-step Vizing chain $F+P$ under $\Shift(\chi,C)$, where
    $$
    C=F_1+P_1+\dots+F_i+P_i
    $$
    is an $i$-step self-avoiding Vizing chain on $e$ with $\vend(P_i)=x$ and $\eend(P_i)=xy$.

    By the description of the set construction stage in construction phase $i+1$, the vertex $x$ belongs to $B_i(e,\len)$ only if $F+P$ is an unsuccessful candidate chain. By definition, this happens only if at least one of the following holds:
    \begin{itemize}
        \item
        $$
        V(P)\cap N^2\!\left[\bigcup_{j=0}^{i-1}R_j(e,\len)\right]\neq\emptyset,
        $$
        \item $x$ and $\vend(F)$ are $\bichrom(P)$-related.
    \end{itemize}

    We first bound the number of vertices $x\in S_i(e,\len)$ for which
    $$
    V(P)\cap N^2\!\left[\bigcup_{j=0}^{i-1}R_j(e,\len)\right]\neq\emptyset.
    $$
    A vertex in
    $$
    \bigcup_{j=0}^{i-1}N^2[R_j(e,\len)]
    $$
    can be reached through path chains of candidate chains starting from at most $(\Delta+1)^5$ different vertices in $S_i(e,\len)$ by \Cref{lemma:chains_through_node_or_edge}. Similarly, a vertex in
    $$
    \bigcup_{j=0}^{i-1}N^2[R_j(e,\len)]
    $$
    can be reached through fans of candidate chains starting from at most $(\Delta+1)^3$ different vertices in $S_i(e,\len)$ by \Cref{lemma:chains_through_node_or_edge}. Hence the number of vertices $x\in S_i(e,\len)$ for which
    $$
    V(P)\cap N^2\!\left[\bigcup_{j=0}^{i-1}R_j(e,\len)\right]\neq\emptyset
    $$
    holds is at most
    $$
    (\Delta+1)^5\sum_{j=0}^{i-1}|R_j(e,\len)|
    +
    (\Delta+1)^3\sum_{j=0}^{i-1}|R_j(e,\len)|.
    $$

We now bound the number of vertices $x\in S_i(e,\len)$ for which $x$ and $\vend(F)$ are $\bichrom(P)$-related. Since $i\geq 1$, the chain $F+P$ is constructed in construction phase $i+1\geq 2$. Hence $x$ constructs $F+P$ by the \nameref{j_phase_construction_case_3} by \Cref{obsvn:x_vend(F)_connected_in_case_3}. Since
    \begin{itemize}
        \item $\bichrom(P_i)$ and $\bichrom(P)$ are respectively $\vend(F_i)$-maximal and $\vend(F)$-maximal bichromatic paths under $\chi$, by \Cref{claim:Phase_1_bichrom(P)_is_bichromatic,claim:j_phase_bichrom(P)_is_bichromatic},
        \item $\bichrom(P_i)$ and $\bichrom(P)$ use the same pair of colors, by construction in the \nameref{j_phase_construction_case_3},
        \item
        $$
        x\in V(\bichrom(P_i))\cap V(\bichrom(P)),
        $$
    \end{itemize}
    it follows that $\bichrom(P_i)$ and $\bichrom(P)$ are subpaths of the same maximal bichromatic path $M$ in $\chi$, with $\vend(F_i)$ and $\vend(F)$ as the two endpoints of $M$. We now count the number of such vertices $x$. First, $\vend(F_i)$ is one of the endpoints of $M$; since
    $$
    \vend(F_i)\in N(\cen(F_i))\subseteq N(S_{i-1}(e,\len)),
    $$
    there are at most $\Delta\cdot |S_{i-1}(e,\len)|$ choices for $\vend(F_i)$. The pair of colors of $M$ can be chosen in at most $\Delta^2$ ways, and these choices determine $M$ and hence its other endpoint $\vend(F)$. Finally, $x\in N(\vend(F))$, giving at most $\Delta$ further choices. Since $x$ constructs at most one candidate chain for $e$ in construction phase $i+1$, there are at most
    $$
    \Delta\cdot\Delta^2\cdot\Delta\cdot |S_{i-1}(e,\len)|=\Delta^4\cdot |S_{i-1}(e,\len)|
    $$
    vertices $x\in S_i(e,\len)$ for which $x$ and $\vend(F)$ are $\bichrom(P)$-related.

    Summing the two contributions, we obtain
    $$
    |B_i(e,\len)|
    \leq
    (\Delta+1)^5\sum_{j=0}^{i-1}|R_j(e,\len)|
    +
    (\Delta+1)^3\sum_{j=0}^{i-1}|R_j(e,\len)|
    +
    \Delta^4|S_{i-1}(e,\len)|.
    $$
    Since
    $$
    S_{i-1}(e,\len)\subseteq R_{i-1}(e,\len)\subseteq \bigcup_{j=0}^{i-1}R_j(e,\len),
    $$
    we have
    $$
    \Delta^4|S_{i-1}(e,\len)|
    \leq
    (\Delta+1)^5\sum_{j=0}^{i-1}|R_j(e,\len)|.
    $$
    Therefore,
    $$
    |B_i(e,\len)|
    \leq
    3(\Delta+1)^5\sum_{j=0}^{i-1}|R_j(e,\len)|.
    $$
\end{proof}
    
The last two lemmas bound $|S_i(e,\len)|$ from below and $|B_i(e,\len)|$ from above. Since $S_i(e,\len)$ is the disjoint union of $D_i(e,\len)$, $Q_i(e,\len)$ and $B_i(e,\len)$, together they bound $|Q_i(e,\len)|$ from below. This is the growth statement we now prove.
    \begin{lemma} \label{lemma-growth_rate}
         Statement $\Pi(i)$ is true where $\Pi(i)$ is stated as follows: Let $1\leq i< \numphases$. Suppose $\bigcup_{j=0}^{i}D_j(e,\len)=\emptyset$. Let $\len=1584(\Delta+1)^{22}$. Then the following properties hold at the end of algorithm $\fA$. 
         \begin{enumerate}
             \item $|Q_{i}(e,\len)|\geq 66(\Delta+1)^{12}\sum_{k=0}^{i-1}|Q_k(e,\len)|$\label[instance]{instance:Q_geq_sum-Q}
             \item $|Q_{i}(e,\len)|\geq \frac{|R_i(e,\len)|}{22(\Delta+1)^7}$\label[instance]{instance:Q_geq_R}
         \end{enumerate}
         \end{lemma}

         \begin{proof}
             By \Cref{claim:D_Q_S_disjoint} and the given fact that $\bigcup_{k=0}^{i}D_k(e,\len)=\emptyset$, we have that 
             $S_i(e,\len)=Q_i(e,\len)\sqcup B_i(e,\len)$.
             
             We will proceed by strong induction on $i$.
             
             \textbf{Case $i=1$.}
By \Cref{claim:x_in_D_or_Q}, and the given fact that $\bigcup_{k=0}^{1}D_k(e,\len)=\emptyset$, we have that
$$
S_0(e,\len)=Q_0(e,\len).
$$

Let $F+P$ be a candidate chain constructed for $e$ in construction phase $1$. By \Cref{obsvn:S_0_construct_candidates}, there exists an $x\in S_0(e,\len)$ that constructs this candidate chain. Since $S_0(e,\len)$ is a singleton by definition, $x$ is the unique vertex that constructs a candidate chain for $e$ in construction phase $1$. Since $M_0(x)$ does not have any duplicates, $F+P$ is the only candidate chain $x$ constructs for $e$ in construction phase $1$. Hence
$$
R_1(e,\len)=V(P)\setminus\{x\}
$$
by \Cref{obsvn:in_R_1_iff_in_P}.

Since $S_0(e,\len)=Q_0(e,\len)$, we have $x\in Q_0(e,\len)$. By the set construction stage for $x$ in construction phase $1$, this implies that $F+P$ is hopeful. Hence the length of $\bichrom(P)$ is $\len$ by \Cref{obsvn:hopeful_is_long_in_phase_1}, and $R_1(e,\len)=V(P)\setminus\{x\}=V(\bichrom(P))$. In particular,
$$
|R_1(e,\len)|\ge \len.
$$

By \Cref{lemma:Si_poly(Delta)_frac_of_Ri},
$$
|S_1(e,\len)|\ge \frac{|R_1(e,\len)|}{11(\Delta+1)^7}.
$$
By \Cref{lemma-num_conflicts},
$$
|B_1(e,\len)|
\le 3(\Delta+1)^5\sum_{j=0}^{0}|R_j(e,\len)|
=3(\Delta+1)^5|R_0(e,\len)|.
$$
Since $|R_0(e,\len)|=1$, we obtain
$$
|B_1(e,\len)|\le 3(\Delta+1)^5.
$$

Therefore,
$$
|Q_1(e,\len)|
=
|S_1(e,\len)|-|B_1(e,\len)|
\ge
\frac{|R_1(e,\len)|}{11(\Delta+1)^7}-3(\Delta+1)^5.
$$
Since $|R_1(e,\len)|\ge \len=1584(\Delta+1)^{22}$, we have
$$
\frac{|R_1(e,\len)|}{22(\Delta+1)^7}
\ge
72(\Delta+1)^{15}
\ge
3(\Delta+1)^5.
$$
Hence
$$
|Q_1(e,\len)|
\ge
\frac{|R_1(e,\len)|}{22(\Delta+1)^7}.
$$
Thus \Cref{instance:Q_geq_R} in $\Pi(1)$ is true.

Also,
$$
|Q_1(e,\len)|
\ge
\frac{\len}{11(\Delta+1)^7}-3(\Delta+1)^5
=
144(\Delta+1)^{15}-3(\Delta+1)^5
\ge
66(\Delta+1)^{12}|Q_0(e,\len)|,
$$
since $|Q_0(e,\len)|=1$. Hence \Cref{instance:Q_geq_sum-Q} in $\Pi(1)$ is true.

Hence $\Pi(1)$ is true.

             \textbf{Case $2\leq i <\numphases$:} Assume inductively that $\Pi(k)$ is true for $1\leq k \leq i-1$. By the inductive hypothesis, $|R_{i-1}(e,\len)|\leq 22(\Delta+1)^7\cdot |Q_{i-1}(e,\len)|$, satisfying the hypothesis of \Cref{lemma:Si_poly(Delta)_frac_of_Ri}.
                
             By \Cref{lemma:Si_poly(Delta)_frac_of_Ri}, 
             $$|S_i(e,\len)|\geq \frac{|R_i(e,\len)|}{11(\Delta+1)^7}.$$
             
By \Cref{lemma-num_conflicts},
$$
|B_i(e,\len)|
\le
3(\Delta+1)^5\sum_{j=0}^{i-1}|R_j(e,\len)|.
$$
Now $|R_0(e,\len)|=1=|Q_0(e,\len)|$, and for every $1\le j\le i-1$, by \Cref{instance:Q_geq_R} in $\Pi(j)$,
$$
|R_j(e,\len)|\le 22(\Delta+1)^7|Q_j(e,\len)|.
$$
Therefore,
\begin{equation} \label{eqn:Bi_upper_bound}
\begin{split}
|B_i(e,\len)|
&\le 3(\Delta+1)^5\left(|R_0(e,\len)|+\sum_{j=1}^{i-1}|R_j(e,\len)|\right)\\
&\le 3(\Delta+1)^5\left(|Q_0(e,\len)|+22(\Delta+1)^7\sum_{j=1}^{i-1}|Q_j(e,\len)|\right)\\
&\le 66(\Delta+1)^{12}\sum_{j=0}^{i-1}|Q_j(e,\len)|.
\end{split}
\end{equation}

        We have by \Cref{instance:Q_geq_sum-Q} that 
        \begin{equation} \label{eqn:Qi_greater_than_sum_Qj}
            \sum_{k=0}^{i-1}|Q_{k}(e,\len)|\leq (1+\frac{1}{66(\Delta+1)^{12}})|Q_{i-1}(e,\len)|\leq \frac{12}{11}|Q_{i-1}(e,\len)|        
        \end{equation}

        Combining \Cref{lemma-num_new_nodes} and \Cref{eqn:Qi_greater_than_sum_Qj}, we get
\begin{equation} \label{eqn:Si_lower_bound}
        \begin{split}
|S_i(e,\len)|
\ge \frac{|R_i(e,\len)|}{11(\Delta+1)^7} 
&\ge \frac{|Q_{i-1}(e,\len)|}{11(\Delta+1)^7}\cdot\frac{\len}{\Delta^3}\\
&\ge \frac{\len}{12(\Delta+1)^{10}}\sum_{k=0}^{i-1}|Q_k(e,\len)|\\
&\geq 132 (\Delta+1)^{12}\sum_{k=0}^{i-1}|Q_{k}(e,\len)|
\end{split}
\end{equation}

        Hence from \Cref{eqn:Bi_upper_bound} and \Cref{eqn:Si_lower_bound}, we have that 
        \begin{equation} \label{eqn:Qi_lower_bound}
            |Q_{i}(e,\len)|= |S_{i}(e,\len)|-|B_{i}(e,\len)| \geq 66 (\Delta+1)^{12}\sum_{k=0}^{i-1}|Q_{k}(e,\len)| 
        \end{equation}
        Hence, \Cref{instance:Q_geq_sum-Q} is true in $\Pi(i)$.
        
        From \Cref{eqn:Qi_lower_bound} and \Cref{eqn:Bi_upper_bound}, $|Q_i(e,\len)|\geq |B_i(e,\len)|$. Hence 
        $$
            2|Q_i(e,\len)|\geq |Q_i(e,\len)|+|B_i(e,\len)|= |S_i(e,\len)|.
        $$
        Since $|S_i(e,\len)|\geq \frac{|R_i(e,\len)|}{11(\Delta+1)^7}$ by \Cref{lemma:Si_poly(Delta)_frac_of_Ri}, we have
        $$
            |Q_i(e,\len)|\ge \frac{|S_i(e,\len)|}{2}\ge \frac{|R_i(e,\len)|}{22(\Delta+1)^7}.
        $$
        Hence \Cref{instance:Q_geq_R} is true in $\Pi(i)$.

        Hence $\Pi(i)$ is true.
\end{proof}

This growth cannot be sustained for $\numphases$ phases, since the sets $Q_i(e,\len)$ would outgrow the graph. Hence some construction phase must produce a terminating chain.
    \begin{lemma}\label{lemma-Di_non_empty}
    Let $e\in U$. Then $\bigcup_{j=0}^{\numphases-1} D_{j}(e,\len)\neq \emptyset$ at the end of algorithm $\fA$.
    \end{lemma}

    \begin{proof}
   Suppose $\bigcup_{j=0}^{\numphases-1} D_{j}(e,\len)= \emptyset$.
    Then, \Cref{lemma-growth_rate} along with the fact that $|Q_0(e,\len)|=1$ implies that 
    $|Q_{\numphases-1}(e,\len)|\geq (66(\Delta+1)^{12})^{\numphases-1}>n$, using $66(\Delta+1)^{12}\geq 2^{25}$ and $\numphases-1\geq\log_2 n-1>\frac{\log_2 n}{25}$ (as $\Delta\geq 2$ forces $n\geq 3$). This is not possible since there are at most $n$ vertices. Hence, we are done.
    \end{proof}

Hence we have the following result at the end of $\fA$.

\begin{lemma} \label{lemma:uncol_edges_have_terminating_chains}
    Algorithm $\mathcal{A}$ has constructed a self-avoiding terminating multi-step Vizing chain of step length $\len$ on every uncolored edge $e$ under $\chi$.
\end{lemma}
\begin{proof}
    We have that there exists an $i$ with $0\leq i < \numphases$ such that $D_i(e,\len)\neq \emptyset$ for any uncolored edge $e$ by \Cref{lemma-Di_non_empty}. 
    
    If $i=0$, then let $x\in D_0(e,\len)$. By the set construction stage in construction phase $1$, this implies that $x$ finished the construction of a successful candidate chain $F+P$ for $e$. By \Cref{lemma:successful_is_terminating_phase_1}, $F+P$ is a terminating self-avoiding $1$-step Vizing chain of step length $\len$ on $e$ under $\chi$.

   If $1\leq i < \numphases$, then let $x\in D_i(e,\len)$. By the set construction stage in construction phase $i+1$, this implies that $x$ finished the construction of a successful candidate chain $F+P$ for $e$ in construction phase $i+1$. By \Cref{lemma:successful_is_terminating_phase_j}, $C+F+P$ is a terminating self-avoiding $(i+1)$-step Vizing chain of step length $\len$ on $e$ under $\chi$, where $C$ is the $i$-step Vizing chain of \Cref{fact_1}.
   
   Hence done.
    
\end{proof}

\section{Analysis of Selection}\label{sec:analysis-selection}

We now analyze \Cref{subsec:selecting-learning}. We bound the runtime of a marking phase, show that the selection of $C_e$ is well defined and that $C_e$ is a terminating self-avoiding multi-step Vizing chain, bound the runtime of a relay phase, and finally bound the number of chains that a single chain can intersect.

\subsection{Properties of Marking Phase $i$ for $0\leq i\leq T-2$}\label{subsec:marking-phase-properties}
We begin by analyzing the runtime of a marking phase.
\RuntimeMarkingPhase*
\begin{proof}
    Consider a marking phase $i$ and a subphase $(c,r)$ in marking phase $i$.

    Consider the check stage. By construction, any message passing through an edge $e'$ is of the form $\tau_{\operatorname{check}}=(f,x,y,\A,\B)$ for some $f\in E(G)$, $x,y\in V(G)$ and $\A,\B\in [\Delta+1]$. Note that $\tau_{\operatorname{check}}$ passes through $e'$ only if either 1) $e'$ is incident to $x$ with $\psi(x)=c$, or 2) $e'$ is an edge in a maximal $\ABtuple$-bichromatic path, say $Q$, under $\chi$ such that $x$ is a neighbor of an endpoint of $Q$. Indeed, by the description of the check stage, the message $\tau_{\operatorname{check}}=(f,x,y,\A,\B)$ is propagated by the vertex $x$, with $\psi(x)=c$, along the candidate chain $F+P$ that $x$ constructed for $f$ in subphase $(c,r)$ of construction phase $i+1$, where $y=\vend(F)$: the first hop of the propagation uses an edge incident to $x$, and every subsequent hop uses an edge of $\bichrom(P)$. By \Cref{claim:Phase_1_bichrom(P)_is_bichromatic,claim:j_phase_bichrom(P)_is_bichromatic}, $\bichrom(P)$ is a $\vend(F)$-maximal $\ABtuple$-bichromatic path under $\chi$, so every edge of $\bichrom(P)$ lies on the maximal $\ABtuple$-bichromatic path $Q$ under $\chi$ that has $y=\vend(F)\in N(x)$ as an endpoint.
    Since there are at most $O(\Delta^2)$ vertices that are colored $c$ and are a neighbor of an endpoint of a maximal bichromatic path under $\chi$ in which $e'$ is contained, at most $O(\Delta^2)$ messages pass through $e'$ in the check stage. Hence, there are at most $O(\Delta^2\log n)$ bits passing through an edge $e'$ in the check stage. Clearly, the check stage runs in $O(\len)$ rounds in the LOCAL model. Hence the check stage can be simulated in CONGEST in $O(\Delta^2\len)$ rounds by \Cref{obsvn:LOCAL_to_CONGEST}.

    An analogous argument works for the status stage and gives that the status stage can be simulated in CONGEST in $O(\Delta^2\len)$ rounds.

    Since there are at most $O(\Delta^4)$ subphases in marking phase $i$, marking phase $i$ then runs in $O(\Delta^6\len)$ rounds.
\end{proof}

We turn to the sets $\zeta_i(e,\len)$ that the marking phases produce. The next observation collects three facts about them: they consist of vertices of $S_i(e,\len)$, they are pairwise disjoint, and they are nonempty for every index up to some $k$. We use it in \Cref{subsec:selected-chain-correctness} to show that the selection of \Cref{subsubsec:selecting-ce} is well defined; in particular $\zeta_0(e,\len)\neq\emptyset$, so the selection has a vertex to start from.
\begin{observation} \label{obsvn:zeta_in_Si_and_non_empty}
Let $e\in U$. Then, there exists a $k$ such that $0\leq k\leq T-1$ and $\zeta_i(e,\len)\neq \emptyset$ for any index $i$ with $0\leq i\leq k$. Moreover, for any indices $i,j$ with $0\leq i,j\leq T-1$ and $i\neq j$,
\begin{itemize}
    \item $\zeta_i(e,\len)\subseteq S_i(e,\len)$, and
    \item  $\zeta_i(e,\len)\cap \zeta_j(e,\len)=\emptyset$.
\end{itemize}

\end{observation}
\begin{proof}
    We will first show that $\zeta_i(e,\len)\subseteq S_i(e,\len)$ for $0\leq i\leq T-1$. Consider $x\in \zeta_i(e,\len)$. If $x$ was in $\zeta_i(e,\len)$ before the marking phases, $x\in D_i(e,\len)$ and hence $x\in S_i(e,\len)$ by \Cref{claim:D_Q_S_disjoint}. If not, $x$ was added to $\zeta_i(e,\len)$ during marking phase $i$. Any vertex that is added to  $\zeta_i(e,\len)$ is in $ S_i(e,\len)$ by description of marking phase $i$. Hence $\zeta_i(e,\len)\subseteq S_i(e,\len)$ for  any index $i$ with $0\leq i\leq T-1$. 
    
    We now proceed to show that $\zeta_i(e,\len)\cap \zeta_j(e,\len)=\emptyset$ for any $i\neq j$. By \Cref{claim:R_i_and_S_j_disjointness}, $S_i(e,\len)\cap S_j(e,\len)=\emptyset$ for any $i\neq j$. Since $\zeta_i(e,\len)\subseteq S_i(e,\len)$ and $\zeta_j(e,\len)\subseteq S_j(e,\len)$, it follows that $\zeta_i(e,\len)\cap \zeta_j(e,\len)=\emptyset$ for any $i\neq j$. 
    
    We now proceed to show that there exists a $k$ with $0\leq k\leq T-1$ such that $\zeta_i(e,\len)\neq \emptyset$ for $0\leq i \leq k$. 
    
    By \Cref{lemma-Di_non_empty}, there exists an index with $D_i(e,\len)\neq\emptyset$; let $k$ be the largest index with $0\leq k\leq T-1$ and $D_k(e,\len)\neq \emptyset$. Since $\zeta_k(e,\len)$ is initialized to $D_k(e,\len)$ and vertices are only ever added to it during the marking phases, we have $\zeta_k(e,\len)\supseteq D_k(e,\len)\neq\emptyset$.

    We now show by downward induction on $i$ that $\zeta_i(e,\len)\neq\emptyset$ for all $0\leq i\leq k$. The base case $i=k$ holds since $\zeta_k(e,\len)=D_k(e,\len)\neq\emptyset$. For the inductive step, fix $0\leq i\leq k-1$ and assume $\zeta_{i+1}(e,\len)\neq\emptyset$. We show $\zeta_i(e,\len)\neq\emptyset$.
    
    Let $v\in \zeta_{i+1}(e,\len)$. Then $v\in S_{i+1}(e,\len)$. By \Cref{obsvn:S_j_M_j_relation}, there exists a unique element $(e,x,\A,\B,vw)$ in $M_{i+1}(v)$ for some $x\in S_i(e,\len)$. By the set construction stage in construction phase $i+1$, $(e,x,\A,\B,vw)\in M_{i+1}(v)$ implies that $x$ constructs a unique candidate chain $F+P$ for $e$ in construction phase $i+1$ such that $v\in V(P)\setminus N^2[x]$. 
    
    Hence $x$ constructs a unique candidate chain $F+P$ for $e$ in construction phase $i+1$ with $v\in V(\bichrom(P))$, $v\in \zeta_{i+1}(e,\len)$ and $(e,x,\A,\B,vw)\in M_{i+1}(v)$. In particular, $v$ meets the first stopping condition of the check stage of marking phase $i$. Since $v\in V(P)$ and the propagation of $\tau_{\operatorname{check}}$ stops at the first vertex of $P$ meeting one of the two stopping conditions, it stops at $v$ or at an earlier vertex of $P$ that also meets the first condition; in either case $\vend(P')$ meets the first condition. By the status stage for $x$ in marking phase $i$, the vertex $\vend(P')$ therefore sends $\tau_{\operatorname{status}}$ to $x$, so $x$ gets added to $\zeta_{i}(e,\len)$ and hence $\zeta_i(e,\len)\neq\emptyset$. This completes the induction.
\end{proof}

\subsection{Correctness of the Selected Chain}\label{subsec:selected-chain-correctness}

Recall from \Cref{subsubsec:selecting-ce} the vertex sequence $x_0,\dots,x_{k_e}$, the stopping index $k_e$, and the chains $F_{i+1},P_{i+1}$. We first check that the process defining this sequence is well defined: it has a unique vertex to start from, and it stops at an index at most $T-1$, at a vertex whose own candidate chain is successful.

\begin{claim}\label{claim:stopping-index}
    $\zeta_0(e,\len)$ is a singleton, $k_e\leq T-1$, and $x_{k_e}\in D_{k_e}(e,\len)$.
\end{claim}

\begin{proof}
The set $\zeta_0(e,\len)$ is nonempty by \Cref{obsvn:zeta_in_Si_and_non_empty}, and it contains at most one vertex since $\zeta_0(e,\len)\subseteq S_0(e,\len)$ and $S_0(e,\len)$ is a singleton. Hence $\zeta_0(e,\len)$ is a singleton and $x_0$ is well defined.

Each $x_i$ lies in $\zeta_i(e,\len)$, so the process can run for at most $T$ steps. Moreover, if the process reaches $x_{T-1}$, then $x_{T-1}\in\zeta_{T-1}(e,\len)=D_{T-1}(e,\len)$, since no marking phase has index $T-1$ and hence $\zeta_{T-1}(e,\len)$ retains its initial value. Therefore the process stops at some $k_e\le T-1$ with $x_{k_e}\in D_{k_e}(e,\len)$.
\end{proof}

The candidate chain of each chosen vertex is truncated at the next chosen vertex before the chains are concatenated, so it is not immediate that the result is a Vizing chain at all. We show that every prefix of the concatenation is one, which is what the induction needs, and that the full concatenation $C_e$ is moreover terminating.
\begin{lemma}\label{lemma:selected-chain-vizing}
For each $0\leq i\leq k_e$, the chain $F_1+P_1+\dots+F_{i+1}+P_{i+1}$ is a self-avoiding $(i+1)$-step Vizing chain of step length $\len$ on $e$ under $\chi$. Moreover, $C_e=F_1+P_1+\dots+F_{k_e+1}+P_{k_e+1}$ is a terminating self-avoiding $(k_e+1)$-step Vizing chain of step length $\len$ on $e$ under $\chi$.
\end{lemma}

\begin{proof}
We proceed by induction on $i$.

\textbf{Case $i=0$.}
By \Cref{lemma:F+P_is_self-avoiding_Vizing_chain}, $F_1'+P_1'$ is a self-avoiding $1$-step Vizing chain of step length $\len$ on $e$ under $\chi$. Since $F_1=F_1'$ and $P_1$ is a subpath of $P_1'$, it follows that $F_1+P_1$ is also a self-avoiding $1$-step Vizing chain of step length $\len$ on $e$ under $\chi$.

\textbf{Case $1\leq i\leq k_e$.}
Assume inductively that $F_1+P_1+\dots+F_i+P_i$ is a self-avoiding $i$-step Vizing chain of step length $\len$ on $e$ under $\chi$. By the definitions above, this chain satisfies
\begin{itemize}
    \item for each $1\leq k\leq i$, we have $F_k=F_k'$ and $P_k$ is a subpath of $P_k'$, where $F_k'+P_k'$ is the unique candidate chain constructed by $\cen(F_k)$ for $e$ during construction phase $k$,
    \item for each $1\leq k\leq i$, the set $M_k(\vend(P_k))$ contains a tuple with first entry $e$ and second entry $\cen(F_k)$, and
    \item $\vend(P_i)=x_i$, where $x_i=\cen(F_{i+1})$, and $\eend(P_i)$ is the edge by which $P_i'$ reaches $x_i$, since $P_i=P_i'|d_{P_i'}(x_{i-1},x_i)$.
\end{itemize}

The second of these conditions holds for every $1\leq k\leq i$: the vertex $x_{k-1}$ was added to $\zeta_{k-1}(e,\len)$ because of $x_k$ in marking phase $k-1$, so by the check stage of that phase, $M_k(x_k)$ contains a tuple with first entry $e$ and second entry $x_{k-1}=\cen(F_k)$, and $\vend(P_k)=x_k$ since $P_k=P_k'|d_{P_k'}(x_{k-1},x_k)$.

Since $x_i\in\zeta_i(e,\len)\subseteq S_i(e,\len)$ by \Cref{obsvn:zeta_in_Si_and_non_empty}, statement $\Pi_3(i)$ from \Cref{lemma:Pi1_Pi2_Pi3_1,lemma:Pi1_Pi2_Pi3_i} yields a unique tuple $t\in M_i(x_i)$ with first entry $e$; by the previous paragraph its second entry is $x_{i-1}$. Hence $t$ was added to $M_i(x_i)$ on account of the candidate chain $F_i'+P_i'$ constructed by $x_{i-1}$ in construction phase $i$, and so, by the set construction stage of that phase, the last entry of $t$ is $\eend\bigl(P_i'|d_{P_i'}(x_{i-1},x_i)\bigr)=\eend(P_i)$.

Applying statement $\Pi_1(i)$ from \Cref{lemma:Pi1_Pi2_Pi3_1,lemma:Pi1_Pi2_Pi3_i} to $t$, the chain $F_1+P_1+\dots+F_i+P_i$ is the unique $i$-step Vizing chain with these properties. Moreover, $t$ is the tuple that $x_i$ serves when it constructs a candidate chain for $e$ in construction phase $i+1$, since it is the only tuple in $M_i(x_i)$ with first entry $e$.

Hence $F_1+P_1+\dots+F_i+P_i$ is exactly the chain $C$ such that $x_i$ constructed the candidate chain $F_{i+1}'+P_{i+1}'$ under $\Shift(\chi,C)$, by \Cref{fact_1}. Therefore, by \Cref{lemma:C+F+P_self_avoiding}, the chain $F_1+P_1+\dots+F_i+P_i+F_{i+1}'+P_{i+1}'$ is a self-avoiding $(i+1)$-step Vizing chain of step length $\len$ on $e$ under $\chi$. Since $F_{i+1}=F_{i+1}'$ and $P_{i+1}$ is a subpath of $P_{i+1}'$, it follows that $F_1+P_1+\dots+F_i+P_i+F_{i+1}+P_{i+1}$ is also a self-avoiding $(i+1)$-step Vizing chain of step length $\len$ on $e$ under $\chi$.
For the moreover part, recall that $F_{k_e+1}+P_{k_e+1}$ is the candidate chain that $x_{k_e}$ constructed for $e$ in construction phase $k_e+1$, untruncated, and that it is successful since $x_{k_e}\in D_{k_e}(e,\len)$. If $k_e=0$, then $C_e=F_1+P_1$ is a terminating self-avoiding $1$-step Vizing chain of step length $\len$ on $e$ under $\chi$ by \Cref{lemma:successful_is_terminating_phase_1}. If $k_e\geq 1$, then the case $i=k_e$ above identifies $F_1+P_1+\dots+F_{k_e}+P_{k_e}$ with the chain $C$ of \Cref{fact_1} for the tuple that $x_{k_e}$ serves, so \Cref{lemma:successful_is_terminating_phase_j} applied to $x_{k_e}$ shows that $C_e=C+F_{k_e+1}+P_{k_e+1}$ is a terminating self-avoiding $(k_e+1)$-step Vizing chain of step length $\len$ on $e$ under $\chi$.
\end{proof}

\subsection{Runtime for Learning the Terminating Chain}\label{subsec:runtime-learning}

By \Cref{obsvn:LOCAL_to_CONGEST}, the runtime of a relay phase again comes down to how many messages cross a single edge. In relay phase $i$, a message travels only within the $i$-th segment of a selected chain, so it suffices to bound how many selected chains have their $i$-th segment passing through a fixed vertex.

\begin{lemma} \label{lemma:terminating_chains_through_node}
    Let $x\in V(G)$ and $i\in[\numphases]$. Let \[E^i_x:=\{e\in U\mid x\in V(F_{i}^e+P_{i}^e) \text{ where $C_e=F_{1}^e+P_{1}^e+\dots+F_{k_e+1}^e+P_{k_e+1}^e$} \}.\]  Then $$|E^i_x|\leq O(\Delta^6)$$ for any $1\leq i \leq \numphases$.
\end{lemma}
\begin{proof}
    Note that for any $e\in U$, we have $x\in V(F_i^e+P_i^e)$ only if $x\in V(F+P)$ where $F+P$ is the candidate chain constructed for $e$ such that $F_i^e= F$ and $P_i^e$ is a subpath of $P$. 
    
    By the set construction stage in construction phase $i$, $x\in V(\bichrom(P_i^e))$ implies that $x$ adds itself to $R_i(e,\len)$. By \Cref{lemma:bounded_addition_Phase_1} (if $i=1$) and \Cref{lemma:bounded_addition_Phase_j} (if $i\geq 2$), and the fact that there are at most $O(\Delta^4)$ subphases in construction phase $i$, any edge $e$ such that $x$ adds itself to $R_i(e,\len)$ must come from a set that is at most $O(\Delta^6)$ large. Hence any edge $e$ such that $x\in V(\bichrom(P_i^e))$ must come from a set that has $O(\Delta^6)$ elements.

    If $x\in V(F_i^e)\setminus V(\bichrom(P_i^e))$, there is at most one vertex in $N[x]$ that can potentially construct $F$ in a fixed subphase $(c,r)$ of construction phase $i$ (since $\psi$ is a proper coloring of $G^2$ and only vertices colored $c$ construct a candidate chain in subphase $(c,r)$). Moreover, each vertex constructs at most one candidate chain per subphase. Hence there is only one candidate chain constructed in subphase $(c,r)$ whose fan $x$ is contained in. Since there are $O(\Delta^4)$ subphases in construction phase $i$, any edge such that $x\in V(F_i^e)$  must come from a set that is at most $O(\Delta^4)$ large.
    
    Hence, any edge $e$ such that $x\in V(F_i^e+P_i^e)$ must come from a set that has $O(\Delta^6)$ elements. Hence we have that $$\card{E_x^i}\leq O(\Delta^6).$$
\end{proof}

The runtime of a relay phase now follows.
\RuntimeRelayPhase*

\begin{proof}
    We analyze the runtime of a fixed relay phase $i$.
    Clearly, relay phase $i$ runs in $O(\len)$ rounds in LOCAL. Fix an edge $f\in E(G)$. Any message that traverses $f=vw$ during relay phase $i$ is of the form
    $(e,\operatorname{chosen}=1)$ or $(e,x,y,\A,\B,\operatorname{chosen}=1)$ where $e\in U$, $x,y\in V(G)$, $\A,\B\in [\Delta+1]$. Each such message can be encoded by $O(\log n)$ bits. Moreover, no such message passes through an edge twice. Hence it suffices to bound the number of messages of these forms that pass through $f$. Note a message of the form  $(e,\operatorname{chosen}=1)$ or $(e,x,y,\A,\B,\operatorname{chosen}=1)$ passes through $f$ only if $f\in E(F_{i}^e+P_{i}^e)$. This implies 
    $e\in E_v^i$ and $e\in E_w^i$. Hence by \Cref{lemma:terminating_chains_through_node}, there are $O(\Delta^6)$ messages that pass through $f$ in relay phase $i$. Hence relay phase $i$ runs in $O(\Delta^6\len)$ rounds in the CONGEST model.
    
\end{proof}

\subsection{The Conflict Graph and Its Maximum Degree}

Let $\chainG$ be the multi-graph where $V(\chainG)=\{C_e\mid e\in U\}$. Two nodes $C_e$ and $C_f$ in $\chainG$ share an edge in $\chainG$ for each vertex of $G$ at which $C_e$ and $C_f$ intersect; that is, if $C_e$ and $C_f$ intersect at $k$ distinct vertices, then there are $k$ parallel edges between $C_e$ and $C_f$ in $\chainG$. We may refer to $\chainG$ as the conflict graph of the terminating Vizing chains.

The maximum degree of $\chainG$ gives an upper bound on how many selected chains a single chain can conflict with, and it is the quantity that the independent set computation of \Cref{sec:comp-indep-chains} depends on. It follows from \Cref{lemma:terminating_chains_through_node} together with the bound on the number of vertices of a selected chain.
\begin{lemma} \label{lemma:degree_of_terminating_chain_graph}
    The maximum degree $\chainDelta$ of $\chainG$ is at most $O(\Delta^6 \len \cdot \log^2 n )$.
\end{lemma}
\begin{proof}
    Fix $e\in U$. By the definition of $\chainG$, the edges of $\chainG$ incident to $C_e$ are in bijection with the pairs $(v,C_f)$ such that $v\in V(C_e)$, $C_f\neq C_e$ and $v\in V(C_f)$. Counting these pairs by their first coordinate, the degree of $C_e$ in $\chainG$, taken with multiplicity, is
    \[
        \deg_{\chainG}(C_e)=\sum_{v\in V(C_e)}\card{\{C_f\neq C_e\mid v\in V(C_f)\}}.
    \]
    Consider any vertex $v\in V(C_e)$. By \Cref{lemma:terminating_chains_through_node} applied to each of the $T=O(\log n)$ relay indices, at most $O(\Delta^6\log n)$ terminating multi-step Vizing chains pass through $v$, so each summand above is $O(\Delta^6\log n)$.
    Since the size of $V(C_e)$ is $O(\len\log n)$, the sum has $O(\len\log n)$ terms, and hence $\deg_{\chainG}(C_e)$ is at most $O(\Delta^6 \len \cdot \log^2 n )$.

\end{proof}

\section{Computing Independent Chains}\label{sec:comp-indep-chains}

While we constructed a set of terminating self-avoiding chains in previous sections, one for each uncolored edge $e \in U$, these chains are only self-avoiding in the sense that any given chain does not intersect itself.
In contrast, chains associated with distinct uncolored edges might overlap. As shown in \Cref{lemma:degree_of_terminating_chain_graph}, each chain can overlap with up to $\chainDelta = O(\Delta^6 \len \cdot \log^2 n )$ other chains.

While two chains that pass through the same vertices can sometimes be shifted simultaneously, this is often not the case.
Consider, for instance two chains $C$ and $C'$ whose first bichromatic path chains respectively use edges colored $(\alpha,\beta)$ and $(\alpha,\delta)$. If the bichromatic path chains overlap on an edge $e$ colored $\alpha$ somewhere in the middle of both path chains, then shifting chain $C$ involves recoloring $e$ with color $\beta$, while shifting chain $C'$ involves recoloring $e$ with color $\delta$.
This means that not all chains in the set we computed can be shifted simultaneously. 
Our goal in this section is to extract a large set of chains that can be shifted simultaneously, which corresponds to finding a large independent set in the graph $\chainG$. We show the following lemma.

\LemLargeISChainGraph*

The size of the independent set we compute is a direct consequence of the maximum degree $\chainDelta \in O(\Delta^6 \len \cdot \log^2 n )$ of the graph $\chainG$.
A graph of $\card{U}$ nodes with maximum degree $\chainDelta$ necessarily contains an independent set of size at least $\card{U} / (1+\chainDelta)$. We aim for a slightly smaller set of size $\Omega(\card{U} / \chainDelta)$

\paragraph*{Notation} Recall that
\begin{itemize}
\item $U$ is the set of uncolored edges in the graph,
\item for each $e \in U$, $C_e$ is the self-avoiding terminating multi-step Vizing chain,
\item $\chainG = (V_\chainG,E_\chainG)$ is the \emph{chains' conflict (multi-)graph}. Its set of nodes consists of the chains of uncolored edges, $V_\chainG = \set{C_e \mid e \in U}$. For each vertex $w \in V$ such that $w \in V(C_e) \cap V(C_f)$, there exists an edge in $\chainG$ between $C_e$ and $C_f$, 
\item two chains $C_e,C_f$ are independent iff they do not    overlap on any vertex, i.e., $V(C_e) \cap V(C_f) = \emptyset$, meaning they are not connected in $\chainG$. A set of chains is independent if any two pairs in the set are independent
\end{itemize}

Throughout the section $\chainDelta = \max_{e \in E} \deg_{\chainG}(C_e)$ is an upper bound on the maximum degree of $\chainG$, with edges counted with multiplicity. That is, when two chains $C_e$ and $C_f$ overlap on $k$ distinct vertices, there are $k$ edges between $C_e$ and $C_f$ in $\chainG$.
In \Cref{ssec:communication-over-chain-graph}, we explain how adjacent chains in $\chainG$ can communicate with one another. More precisely, we give and analyze several communication primitives for the chains to execute. These include chains broadcasting information within themselves (\Cref{lem:broadcast-chain-graph}, aggregating values within themselves (\Cref{lem:aggregation-chain-graph}, and simulating message passing over the conflict graph (\Cref{lem:message-passing-chain-graph}).
In \Cref{ssec:chain-graph-simulation}, we put these primitives to use to break symmetry over the conflict graph $\chainG$, and prove 
\Cref{lem:large-is-chain-graph}.

\subsection{Communicating over the Conflict Graph}
\label{ssec:communication-over-chain-graph}

A difficulty in our setting is that edges of $\chainG$ are not direct
communication links unlike in CONGEST. Nodes in $\chainG$ cannot send
a $O(\log n)$ message to each of their neighbors in a single round. In fact, nodes do not even
have easy access to their incident edges.

Our setting corresponds to that of \emph{embedded virtual graphs} as defined in \cite[Definitions 3 and 4]{FHN_disc24_virtual_graphs}. An embedded virtual graph $H = (V_H,E_H)$ on a communication graph $G = (V_G,E_G)$ is a multi-graph s.t.
\begin{enumerate}
\item each virtual node $v \in V_H$ is mapped to a set of vertices $V(v) \subseteq V_G$ that is the \emph{support} of $v$,
\item each vertex $w \in V_G$ of the communication graph knows the supports $\set{v \in V_H \mid w \in V(v)}$ it is a part of,
\item for each virtual node $v \in V_H$, its support $V(v) \subseteq V_G$ contains a spanning tree $T(v)$, called \emph{support tree}, such that each vertex in the support $V(v)$ knows which of its incident edges in $G$ belong to $T(v)$,
\item for each virtual edge $e=uv \in E_H$, there is exactly one vertex $w \in V(u) \cap V(v)$ in the intersection of the two endpoints' handling the edge $e$. That is, the virtual edges incident to a virtual node $v\in V_H$ are known distributedly by a subset of the vertices in $v$'s support $V(v)$.
\end{enumerate}

Our graph $\chainG = (V_\chainG,E_\chainG)$ is an embedded virtual
graph over the original graph $G=(V_G,E_G)$. For each chain $C \in V_\chainG$, its support $V(C) \subseteq V_G$ is simply the set of vertices in the multi-step Vizing chain $C$. The multi-step Vizing chains are made of paths and fans, and are self-avoiding: this makes the union of all paths and fans of a chain $C$ a natural support tree $T(C)$ for $C$. We additionally root this tree at one of the endpoints of the edge $e$, which we denote $r_e$. Two chains $C_e$ and $C_f$ are adjacent in $\chainG$ if they overlap on some vertex: this vertex is part of the chains' supports, and is in charge of handling that virtual edge in $\chainG$.

From \Cref{lemma:terminating_chains_through_node} we know that every vertex in the graph belongs to the $i$th step of at most $O(\Delta^6)$ distinct terminating chains. This implies the same bound for edges, that is, for any given edge of the graph $G$ at most $O(\Delta^6)$ distinct terminating chains go through that edge. This bounds the congestion in the algorithms we present.

The remainder of this section is devoted to describing and analyzing communication primitives for communicating over $\chainG$. These communication primitives are used in the next section to compute a large independent set of chains.

\begin{lemma}[Broadcast within a chain]
  \label{lem:broadcast-chain-graph}
  Let $B\geq \log n$ be a positive integer. For each edge $e \in U$, let the root $r_e$ of its support tree have a $B$-bit message $M_e$. There is an
  algorithm of complexity $O(\Delta^6 \len B)$ that delivers each message
  $M_e$ to all vertices in the chain $C_e$.
\end{lemma}
\begin{proof}
  The message delivery follows a similar pattern as the construction of
  the multi-step Vizing chains. Our chains are $\numphases$-step Vizing chains of length $\len$ with $\numphases=O(\log n)$ and $\len \in O(\Delta^{22})$.
  We perform $T$ phases such that during phase $i$, the vertices in the $i$th fan $F_i$ and path $P_i$ of a chain $C_e$ learn the message $M_e$. Each phase runs in $O(\Delta^6\len B / \log n)$ rounds, so the overall process takes $O(T \cdot \Delta^6 \len B/ \log n) = O(\Delta^6 \len B)$ rounds.

  First, each message $M_e$ is tagged with a unique $O(\log n)$-bit
  identifier: that of its edge $e$. This allows to easily distinguish
  the messages and know how to forward them in the graph.  For each $i
  \in [T]$ and chain $C_e$, let $T_i(C_e) \subseteq T(C_e)$ be the
  subtree of $C_e$'s support tree that consists of the $i$th fan $F_i$ and path $P_i$ of $C_e$. During phase $i$, each vertex in $G$ forwards each
  message $M_e$ it knows about to its uninformed neighbors in the
  subtrees $T_i(C_e)$ in an arbitrary greedy fashion. As a vertex $w \in V_G$
  belongs to at most $O(\Delta^6)$ distinct $i$th-step subtrees (\Cref{lemma:terminating_chains_through_node}), it never has more than this many messages to send over an incident edge $ww' \in E_G$ during a phase. As such, a message never takes more than $O(\Delta^6
  B /\log n)$ rounds to cross an edge, and as each message has $\len$ edges to cross, the runtime of each phase is $O(\Delta^6 \len B / \log n)$.

  As each phase of broadcast is completed in $O(\Delta^6 \len B / \log n)$ rounds, performing all $T = O(\log n)$ phases takes no more than $O(\Delta^6 \len B T / \log n) = O(\Delta^6 \len B)$ rounds.
\end{proof}

\begin{lemma}[Aggregation within a chain]
  \label{lem:aggregation-chain-graph}
  Let $B\geq \log n$ be a positive integer and $\mathcal{X}$ a set equipped with an associative and commutative operation $\bigoplus$. For each edge $e \in U$,
  suppose that each vertex $w \in V(C_e)$ in its chains' support holds
  some value $x_{w,e} \in \mathcal{X}$. Suppose that for any subset $S \subseteq V(C_e)$ of
  the vertices of a chain $C_e$, $\bigoplus_{w \in S} x_{w,e}$ can be
  represented on $B$ bits.
  There exists an $O(\Delta^6 \len B)$-round CONGEST algorithm that has the root $r_e$ of the support tree of each edge $e \in U$ learn the value $\bigoplus_{w \in V(C_e)}$.
\end{lemma}
\begin{proof}
  This is implemented following the opposite communication pattern as broadcast. Perform a broadcast operation as in \Cref{lem:broadcast-chain-graph} where each message $M_e$ simply contains the identifier of the uncolored edge $e \in U$.
Each vertex $w$ in the support of $C_e$ remembers $t_{w,e} \in O(\Delta^6 \len \log n)$, the round in which it received the message $M_e$. Since broadcast is performed in trees, each vertex in $V(C_e)$ has exactly one neighbor $w'$ with a lower value $t_{w',e} < t_{w,e}$, except the root which has none. Let $t_{\max} \in O(\Delta^6 \len \log n)$ be an upper bound on the maximum value of all $t_{w,e}$s.

Aggregation is performed in reverse order of the values $t_{w,e}$.
For each vertex $w$ and uncolored edge $e$ such that $w$ is in the support of $e$'s chain, $w$ has a variable $Y_{w,e}$, initialized to $Y_{w,e} \gets x_{w,e}$.
Between rounds $i \cdot \Theta(B / \log n)$ and $(i+1) \cdot \Theta(B / \log n)$, a vertex $w$ with $t_{w,e} = t_{\max} - i$ sends $Y_{w,e}$ to its neighbor $w'$ such that $t_{w',e} = t_{w,e} -1 = t_{\max} - (i+1)$. Upon receiving a message $Y_{w',e}$ from one of its neighbors, a vertex $w$ updates its value $Y_{w,e}$ with it: $Y_{w,e} \gets Y_{w,e} \oplus Y_{w',e}$. At the end of the algorithm, the root $r_e$ of $C_e$ has learned $\bigoplus_{w \in V(C_e)} x_{w,e}$.
\end{proof}

\Cref{lem:broadcast-chain-graph,lem:aggregation-chain-graph} allow to perform operations of the form where each chain $C_e$ broadcasts some value $x_e$ to all of its vertices, each vertex $w$ where $C_e$ intersects with another chain $C_f$ computes some function $g_e(x_e,x_f)$ on the values $x_e$ and $x_f$ it received from the two chains, and an aggregation of the values $g_e(x_e,x_f)$ over $C_e$ allows the chain to learn some information about its neighborhood.
The next lemma is an application of this idea.

\begin{lemma}[Naming ports of a chain]
  \label{lem:port-naming-chain-graph}
  There is an algorithm of complexity $O(\Delta^6 \len \log n)$ for each chain $C_e$ to count its degree $\deg_\chainG(C_e)$ in the multi-graph $\chainG$, and give each vertex handling an edge in $\chainG$ a unique port number between $1$ and $\deg_\chainG(C_e)$.
\end{lemma}
\begin{proof}
  Computation of the degree follows from the procedure for aggregation (\Cref{lem:aggregation-chain-graph}): for each vertex $w \in V(C_e)$, let $x_{w,e}$ be the number of edges handled by $w$ in $\chainG$, i.e., the number of chains other than $C_e$ going through $w$. The sum $\sum_{w \in V(C_e)}x_{w,e}$ of all the values $x_{w,e}$ within the chain $C_e$ is exactly the degree of $C_e$ in the multi-graph $\chainG$.

  For port assignments, note that during the aggregation procedure, each vertex in $C_e$ learns the number of edges handled by the subtrees rooted at each of its descendants in the support tree $T(C_e)$ rooted at $r_e$.
  Thus, if given an adequately sized interval $[a,b] \subseteq [1,\deg_\chainG(C_e)]$ of ports to split between the subtrees rooted at its direct descendants, a vertex can split the interval such that each subtree receives an interval with exactly the number of edges handled within that subtree. The algorithm proceeds as follows: initially, the root builds a message with the interval $[1,\deg_\chainG(C_e)]$ tagged with the identifier of $e$, all fitting on $O(\log n)$ bits. Then, for $O(\Delta^6 \len \log n)$ rounds, each support tree performs a form of broadcast where upon receiving an interval $[a,b]$ tagged with the edge $e$, a vertex $w$ first takes the ports $a$ to $a+x-1$ for itself where $x$ is the number of edges $w$ handles for $C_e$. If $w$ has $k$ direct descendants in $T(C_e)$, it then partitions the remaining ports $[a+x,b]$ into intervals $[a_1,b_1],\dots,[a_k,b_k]$, such that for each $i\in[k]$ the subtree rooted at the $i$th descendant of $w$ in $T(C_e)$ handles a total of $b_i - a_i + 1$ edges in $\chainG$ for $C_e$. The procedure is as fast as a standard $O(\log n)$-bit broadcast.
\end{proof}

The next lemma shows how to perform a round of CONGEST$(B)$ on the graph $\chainG$.
\begin{lemma}[Message passing in the conflict graph]
  \label{lem:message-passing-chain-graph}
  Let $B\geq 1$ be a positive integer. There is an algorithm of
  complexity $O(\Delta^6 \len B \chainDelta)$ for each chain $C_e$ to
  send an individual $B$-bit message to each of its neighbors in
  $\chainG$.

  More precisely, the algorithm simulates a round of CONGEST$(B)$:
  each chain $C_e$ sends a $B$-bit message $m_{e,i}$ to the chain
  $C_f$ that it is connected to through port number $i$, and each
  chain can identify the messages it receives by the ports from which
  it received them.
  \end{lemma}
\begin{proof}
  Consider that we previously ran \Cref{lem:port-naming-chain-graph}
  for each chain $C_e$ to know its degree in the multigraph and have
  ports named $1$ to $\deg_\chainG(C_e)$. The complexity of that
  procedure is much lower than our target complexity of
  $O(\Delta^6 \len B \chainDelta)$ and only needs to be run once.

  Consider all the $B$-bit messages $(m_{e,i})_{i \in
    [\deg_{\chainG}(C_e)]}$ that a chain $C_e$ wants to send to its
  neighbors. The root $r_e$ of $C_e$'s support tree $T(C_e)$ crafts a
  $(B\cdot \deg_{\chainG}(C_e))$-bit message $M_e$ that it broadcasts
  to its support tree. This takes $O(\Delta^6 \len \cdot B \chainDelta)$ rounds by \Cref{lem:broadcast-chain-graph}.

  Next, each $C_e$ performs the aggregation of a $O(B \chainDelta)$-bit value, in $O(\Delta^6 \len \cdot B \chainDelta)$ rounds by
  \Cref{lem:aggregation-chain-graph}. 
  The values we aggregate are from the set $\mathcal{X} = (\set{0,1}^B \cup \set{\bot})^{\chainDelta}$, and represent a $\chainDelta$-tuple of $B$-bit messages, with a special character $\bot$ to denote the absence of a message.
  At a vertex $w \in V(C_e)$, that handles the edge of port $i$ for $C_e$, the $i$th
  element in its tuple $x_{w,e} \in \mathcal{X}$ is the message that $w$ received from
  that edge. All other elements are set to $\bot$. The
  output $(z_i)$ of a $\bigoplus$ operation between two tuples $(x_i)$
  and $(y_i)$ is simply to have $z_i = \max(x_i,y_i)$ for each $i$,
  where $\max(\bot,x) = x$ for any value $x$ and an arbitrary order is
  fixed for other values. After this aggregation, the root $r_e$ knows
  about the message received by $C_e$ on each port.
\end{proof}

\subsection{Computing the Independent Set}
\label{ssec:chain-graph-simulation}
To compute our (sufficiently large) independent set of $\chainG$, we use an algorithm by Faour et al.~\cite{FGGKR_talg25_generalized_rounding} for computing independent sets in CONGEST, which we now recall.

\begin{theorem}[Lemma 4.8 in \cite{FGGKR_talg25_generalized_rounding} (simplified)]
  \label{thm:is-rounding-alg}
  Let $G=(V,E)$ be an $n$-(multi-)graph of maximum degree $\Delta$. Assume that $G$ is equipped with a proper $\zeta$-coloring. Then, there is a deterministic CONGEST$(\log \zeta)$ algorithm to compute an independent set $I$ of size at least $n / (4\Delta)$ in $O(\log^2 \Delta + \log^* \zeta)$ rounds.
\end{theorem}

The theorem is only stated for standard CONGEST in \cite{FGGKR_talg25_generalized_rounding}. However, when looking at the details of its proof, its rounding step (see \cite[Lemma 4.2 and Lemma 2.5]{FGGKR_talg25_generalized_rounding}) uses messages of size $O(\log \zeta)$.
By starting from a $\zeta$-coloring with $\zeta \in \Delta^{O(1)}$, the algorithm only uses $O(\log \Delta)$-bit messages. That is, $\Theta(\log n)$-bit messages are only needed in the case where the maximum degree $\Delta$ is of order $n^{\Theta(1)}$.

To apply \Cref{thm:is-rounding-alg} on our multi-graph $\chainG$, we first compute an $O(\chainDelta^2)$-coloring of $\chainG$. We do so using an adaptation of Linial's classic algorithm to our setting, very similar to a recent adaptation by Barenboim and Goldenberg to the distance-2 setting~\cite{BG_disc24}.

\begin{lemma}[Linial's algorithm on $\chainG$]
  \label{lem:linial-chain-graph}
  There exists an $O(\Delta^6 \len \log n (\log \Delta + \log \log n))$-round deterministic CONGEST algorithm for computing an $O(\chainDelta^2)$-coloring of $\chainG$.
\end{lemma}
\begin{proof}
  Recall the basic operation in Linial's
  algorithm~\cite{linial1987distributive}: to go from a $C$-coloring
  to a $C'$-coloring where $C' < C$, the colors $1\dots C$ are
  interpreted as sets $S_1,\dots,S_C \subseteq [C']$ such that for any
  color $i \in [C]$ and colors $j_1,\dots,j_\Delta \in [C]\setminus
  \set{i}$, $S_i \setminus (\bigcup_{k=1}^{\Delta} S_{j_k}) \neq
  \emptyset$. Such a family of sets is said to be $\Delta$-cover
  free. Each node $v$ updates its color $c_v \in [C]$ to a new color
  $c'_v \in [C']$ by taking its new color inside its set $S_{c_v}$,
  but outside all the sets $S_{c_u}$ of all its neighbors.

  A classic construction for such cover-free families based on
  polynomials is one where the intersection $\card{S_i \cap S_{i'}}$
  of any two distinct sets $S_i,S_{i'}$ by some value $d$, while each
  set $S_i$ has size $\card{S_i} > \Delta \cdot d$.  The key insight
  of~\cite{BG_disc24} is to use the fact that with
  such a cover-free family, nodes no longer need full access to their
  neighbors' sets: a node $v$ can instead perform some binary search
  to find an element $c'_v \in S_{c_v} \setminus \bigcup_{u \in N(v)}
  S_{c_u}$. The search operates as follows: the node $v$ learns
  whether its neighbors' sets have more intersection in the first half
  or the second half of its set $S_{c_v}$. For one of the halves, it
  should be the case that it contains more elements than it has
  intersections with the sets of $v$'s neighbors. This allows to
  recurse within that half. Eventually, after at most $\log(S_{c_v})$
  iterations, $v$ should have found a subset of $S_{c_v}$ that has no
  intersection with its neighbors' sets. With some additional
  optimizations this yields a $O(\log \Delta + \log^*n)$ CONGEST algorithm for
  computing a distance-$2$ $O(\Delta^4)$-coloring. See
  \cite{BG_disc24} for more details.

  We compute an $O(\chainDelta^2)$-coloring of the graph $\chainG$
  with the same techniques, using our communication primitives for
  broadcast and aggregation
  (\Cref{lem:broadcast-chain-graph,lem:aggregation-chain-graph}) to
  perform the binary search. The broadcast primitive is used for the
  root $r_e$ of the support tree of a chain $C_e$ to inform the whole
  support of $C_e$ of its current color, or to indicate in which half
  of a set of colors it is continuing with the binary search. The
  aggregation primitive is used to count the number of intersections
  between a chain's set of potential colors and the colors that its
  neighbors could possibly take. This can be performed by aggregation
  as every vertex $w$ handling an edge of $C_e$ in $\chainG$ knows
  both the color of $C_e$ and of the other endpoint of the edge,
  allowing it to count the intersections of the sets associated with these colors. The
  aggregation simply consists of summing all the intersections
  detected by all the vertices handling an edge for the chain
  $C_e$.

  Each step of the binary search takes $O(\Delta^6 \len \log n)$
  rounds for doing both a broadcast and an aggregation by
  \Cref{lem:broadcast-chain-graph,lem:aggregation-chain-graph}.  The
  sets in which the binary search is performed are of size
  $\poly(\Delta,\log n)$, so the binary search finishes in $O(\log
  \Delta + \log \log n)$ iterations.

  As the target space of colors $O(\chainDelta^2)$ is relatively large
  with respect to $n$ ($\card{\log^* n - \log^* (\chainDelta)} \in O(1)$ since $\chainDelta \in[\log^{\Omega(1)} n, n^{O(1)}]$), only
  a constant number of iterations of Linial's color reduction step are
  needed to reduce an original $n^{O(1)}$-coloring based on the identifiers to the final coloring. This makes for a total round
  complexity of $O(\Delta^6 \len \log n (\log \Delta + \log \log n))$.
\end{proof}

We now prove \Cref{lem:large-is-chain-graph}, which we restate for convenience.

\LemLargeISChainGraph*

\begin{proof}
Recall that $\chainDelta \in O(\Delta^6 \len \cdot \log^2 n) = O(\poly(\Delta)\log^2 n)$ (\Cref{lemma:degree_of_terminating_chain_graph}).
First, we compute a $\zeta$-coloring of the multi-graph $\chainG$ using
\Cref{lem:linial-chain-graph}, where
$\zeta \in O(\chainDelta^2) = \poly(\Delta,\log n)$.
Next, we use \Cref{lem:message-passing-chain-graph} to simulate the algorithm of Faour et al.\ on $\chainG$.
By \Cref{thm:is-rounding-alg}, this algorithm would take $O(\log^2 \chainDelta)$ rounds of direct communication over $\chainG$ using $O(\log \chainDelta)$-bit messages, and it computes an independent set of size at least $\card{U}/(2\chainDelta) \in \Omega(\card{U}/(\Delta^6 \len \log^2 n))$.
By \Cref{lem:message-passing-chain-graph}, every round of $O(\log \chainDelta)$-bit message passing on $\chainG$ can be simulated in $O(\Delta^6 \len \chainDelta \log \chainDelta)$ rounds of CONGEST on $G$, so the whole simulation takes
$O(\Delta^6 \len \chainDelta \log^3 \chainDelta)$ rounds.
It remains to replace $\chainDelta$ by its value. Since $\len\in O(\Delta^{22})$, we have $\chainDelta \in O(\Delta^6\len\log^2 n)\subseteq O(\Delta^{28}\log^2 n)$ and hence $\log\chainDelta\in O(\log\Delta+\log\log n)$, so that $\log^3\chainDelta \in O(\log^3\Delta+\log^3\log n)$. Therefore
\[
O(\Delta^6 \len \chainDelta \log^3 \chainDelta)
= O\parens*{\Delta^{12}\len^2\log^2 n\parens*{\log^3\log n+\log^3\Delta}}
\]
rounds suffice, which is the claimed bound.
\end{proof}

\section*{AI Disclosure}
We used ChatGPT-5.5 and Claude Opus 4.8 to assist with polishing grammar and phrasing throughout \Cref{sec:implementation-phases,sec:analysis-construction-phases,sec:existence-of-MSVC,sec:analysis-selection,sec:comp-indep-chains}. The authors verified the correctness and originality of all content.

\bibliography{references}
\end{document}